\documentclass[reprint,aip,amsmath,amssymb,nofootinbib]{revtex4-1}

\usepackage{color}
\usepackage{nicefrac}
\usepackage{amsmath}
\usepackage{amsthm}
\usepackage{mathrsfs}
\usepackage{graphicx}
\usepackage[capitalise]{cleveref}
\usepackage{algorithm}
\usepackage[noend]{algpseudocode}
\usepackage{natbib}
\usepackage{caption}
\usepackage{subcaption}
\usepackage{minitoc}
\noptcrule 
\usepackage[toc,page,header]{appendix}
\usepackage{tocloft}

\usepackage[normalem]{ulem}	

\newcommand{\eg}{\emph{e.g.,~}}
\newcommand{\ie}{\emph{i.e.,~}}
\newcommand{\bx}{\mathbf{x}}
\newcommand{\by}{\mathbf{y}}
\newcommand{\bp}{\mathbf{p}}

\newcommand{\sket}[1]{\vert #1 )}
\newcommand{\sbraket}[2]{( #1 \vert #2 )}
\newcommand{\ket}[1]{\vert#1\rangle}
\newcommand{\bra}[1]{\langle #1\vert}
\newcommand{\expect}[1]{\langle #1\rangle}
\newcommand{\dg}{^\dagger}
\newcommand{\nn}{\nonumber}

\newcommand{\pd}{{\sf p}}	
\newcommand{\ncoord}[1]{\eta^{#1}}
\newcommand{\PsiDO}{\Psi\text{DO}}
\newcommand{\Op}{\operatorname{Op}}
\newcommand{\Flow}{\Phi}
\newcommand{\cstate}[2][]{{\ket{#1{\psi}^h_{#2}}}}
\newcommand{\dcstate}[2][]{{\sket{#1{\psi}^h_{#2}}}}
\newcommand{\ext}{\operatorname{ext}}
\newcommand{\LPCA}{\operatorname{LPCA}}
\newcommand{\sLap}{\breve{L}_{\epsilon,\lambda}}
\newcommand{\dropcap}[1]{#1}

\newtheorem{theorem}{Theorem}
\newtheorem{lemma}{Lemma}

\newtheorem{prop}{Proposition}

\theoremstyle{remark}

\renewcommand{\thesubsection}{\thesection.\Alph{subsection}}
\setcitestyle{square,numbers}

\begin{document}
\doparttoc
\faketableofcontents

\title{Manifold learning via quantum dynamics}

\author{Akshat Kumar}
\email{akumar@clarkson.edu}
\affiliation{Department of Mathematics, Clarkson University, Potsdam, NY 13699 USA}
\affiliation{Instituto de Telecomunica\c{c}\~{o}es, Lisbon, Portugal}
\author{Mohan Sarovar}
\email{mnsarov@sandia.gov}
\affiliation{Sandia National Laboratories, Livermore, California 94550, USA}

\begin{abstract}
We introduce an algorithm for computing geodesics on sampled manifolds that relies on simulation of quantum dynamics on a graph embedding of the sampled data. Our approach exploits classic results in semiclassical analysis and the quantum-classical correspondence, and forms a basis for techniques to learn the manifold from which a dataset is sampled, and subsequently for nonlinear dimensionality reduction of high-dimensional datasets. We illustrate the new algorithm with data sampled from model manifolds and also by a clustering demonstration based on COVID-19 mobility data. Finally, our method reveals interesting connections between the discretization provided by data sampling and quantization.
\end{abstract}

\maketitle

\dropcap{P}hysical processes have inspired information processing algorithms throughout the history of computing. Perhaps the most prominent examples are sampling algorithms inspired by statistical mechanics such as the Metropolis-Hastings algorithm \cite{hitchcock_history_2003}.
More recently, diffusion and the dynamics of heat flow have provided inspiration for techniques that learn the geometry of data sets and perform nonlinear dimensionality reduction to reduce the complexity of high-dimensional datasets
\cite{chui_special_2006}.
This task is referred to as \emph{manifold learning}, and has become a cornerstone task in unsupervised learning from large datasets.
Underlying manifold learning is the \emph{manifold hypothesis} \cite{fefferman_testing_2016}, which states that most real-world high-dimensional datasets, especially those originating from physical systems constrained by physical laws, are actually constrained to lower dimensional manifolds.
Manifold learning has proven to be a successful means for reducing the complexity of such high-dimensional datasets and enabling tasks such as shape analysis and image denoising \cite{Thorstensen_2011}, feature extraction and classification \cite{Leon-Medina_2020}, anomaly detection \cite{Sipola_2011,Stolz_2020}, and clustering \cite{Levin_Chao_Wenger_Proctor_2020}.

Diffusion-based approaches in spectral geometry
\cite{belkin_laplacian_2003,chui_special_2006,jones_manifold_2008,crane_geodesics_2013}
have been successful to data applications because they are local and correspond to Markov processes defined on finite datasets. The spectral analysis of a Markov generator of diffusion on a dataset is used to access its geometrical features. These approaches exploit two key properties: (i) in the limit of infinite data samples, the Markov generator limits to the Laplace-Beltrami (LB) operator on the data manifold (with possible correction terms that depend on the data sampling distribution), and (ii) the low-energy eigenfunctions of the LB operator, which physically correspond to the long-time limit of diffusion, encode coarse aspects of the geometry of the manifold.

In this paper we propose an entirely new approach to learning the geometry of datasets, motivated by \emph{quantum dynamics}. We develop a data-driven approach to extracting geodesics on data manifolds and show how this can be coupled with embedding techniques to achieve non-linear dimensionality reduction.
The key ingredients to our approach are (i) the approximation of a unitary time evolution operator (propagator) for a quantization of Hamiltonian dynamics on the data manifold, (ii) propagation of localized coherent states that approximate classical point particle trajectories, (iii) a rescaling of the diffusion approximation parameter to support the quantized dynamics. Quantization of the dynamics linearizes the classically non-linear problem of solving for geodesics on a manifold, and moreover, our construction allows us to use results in semiclassical analysis to prove strong asymptotic connections between the dynamics induced by unitary propagation and geodesic distances on the underlying manifold. 
Approximations of geodesic paths have been a central component of several previous algorithms in computer vision and network science \cite{mitchell_discrete_1987, kimmel_computing_1998, peyre_geodesic_2010, crane_geodesics_2013}.
However, connections to quantized dynamics have not been exploited in the past, and our work demonstrates the considerable benefits of utilizing this connection.

We begin with background and asymptotic results in \cref{sec:qdyn_geo}, develop data-driven constructions in \cref{sec:data_dyn},  demonstrate our approach on model manifolds in \cref{sec:eg}, and apply it to real-world data by performing clustering on low-dimensional embeddings derived from the output of our method in \cref{sec:app}.
We discuss the relation to previous work and directions for extending our results in \cref{sec:disc}.

\section{Quantum and classical dynamics}
\label{sec:qdyn_geo}
We begin by defining the dataset $V = \{v_1, v_2, ..., v_N \}$ as $N$ samples from the underlying compact, smooth, and boundaryless Riemannian manifold $\mathcal{M}$ with dimension $\nu$ and metric tensor $g_{ij}$. The samples are drawn according to a fixed distribution that has density $\pd$ with respect to the Riemannian volume on $\mathcal{M}$. The manifold is assumed to be isometrically and smoothly ($C^\infty$) embedded in $\mathbb{R}^n$ with a bounded second fundamental form. We view the index $\ell$ to correspond to a sample on $\mathcal{M}$ and $v_{\ell} \in \mathbb{R}^n$ to specify its \emph{extrinsic coordinates}.

Often, in practice, $n\gg \nu$ and this renders resolving many geometric properties of the underlying manifold, including geodesics, computationally intractable.
The geodesic equation, from a Hamiltonian dynamics point of view, models a ``test particle'' moving freely on phase space, which is defined to be the cotangent bundle $T^*\mathcal{M}$. Upon projection onto the \emph{configuration space} $\mathcal{M}$, the motion of this particle traces out a geodesic path with starting position and direction fixed by the initial conditions to the Hamiltonian flow. 
Simulations of this Hamiltonian flow are difficult to achieve directly since there is no straightforward way to use samples only from configuration space, $\mathcal{M}$, to simulate dynamics that are naturally described in terms of phase space variables, \ie  $\bx=(x^1,...x^\nu)$ and $\bp=(p_1,...p_\nu)$, the canonical configuration (or position) and momentum coordinates, respectively.
Another perspective, coming from geometric optics and modeled by the eikonal equation, is that the phase space dynamics simulate the propagation of a light ray in curved space. 
To obtain geodesics using samples from the manifold the conventional approach involves solving a discrete approximation to the eikonal equation \cite{peyre_geodesic_2010}, usually after a triangulation step that approximates a surface from the point cloud $V$.

In the following, we develop a different approach to obtaining geodesics that does not require solving a nonlinear equation, or triangulation of the point cloud, and instead requires only simple matrix operations on the dataset $V$.
Our approach proceeds through a \emph{quantization} of the dynamical system, which is a linear (albeit infinite-dimensional) system, and propagation of wavepackets defined on $\mathcal{M}$, but localized to a point in $T^*\mathcal{M} \setminus 0$ (while still respecting the uncertainty principle) that quantum mechanically approximate the classical test particles. We then assess certain \emph{statistics} of these propagated wavepackets, which constitutes an effective \emph{dequantization} that enables access to the underlying classical -- and nonlinear! -- dynamics, tracing out the geodesic path with initial position and direction fixed by the localization of the initial wavepacket.

\begin{figure}
\centering
\includegraphics[width=1\linewidth]{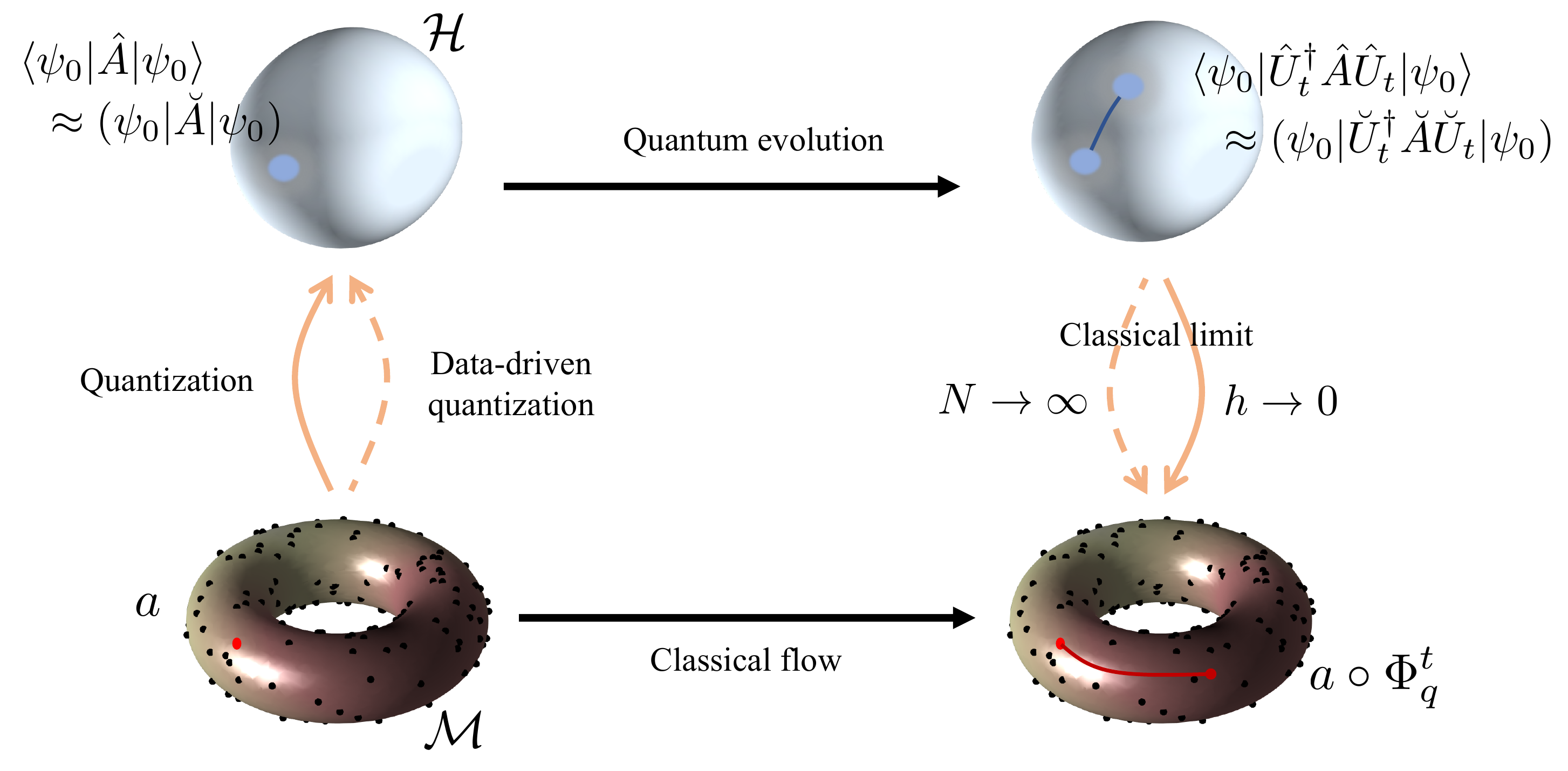}
\caption{A schematic representation of our general approach. Through Egorov's theorem a quantum-classical correspondence can be established where the classical flow on a manifold ($\mathcal{M}$) is related to a quantized evolution in Hilbert space ($\mathcal{H}$). Expectations of time-propagated quantized observables, $\bra{\psi_0}\hat{U}\dg_{t}\hat{A}\hat{U}_t\ket{\psi_0}$, can recover corresponding classical observables under geodesic flow, $a\circ \Phi^t_q$, in the classical limit, $h\to 0$. We show that quantized operators in Hilbert space can be approximated from sample points on the manifold (black dots on the torus above), and construct discrete quantum propagators ($\breve{U})$), states ($|\psi_0)$) and observables ($\breve{A})$) to approximate the time-dependent observable expectation by $(\psi_0|\breve{U}_t\dg\breve{A}\breve{U}_t|\psi_0)$. This approximation converges to the classical flow in the infinite sample limit, $N\to \infty$. }
\label{fig:schematic}
\end{figure}

In general terms, the linearization we are concerned with connects certain dynamics of operators on $L^2(\mathcal{M})$ with Hamiltonian mechanics, or in the functional form, Liouvillian dynamics on $T^*\mathcal{M}$.
The operator dynamics are generated by a linear operator $\hat{Q}_h : C^{\infty}(\mathcal{M}) \to C^{\infty}(\mathcal{M})$ with parameter $h \in (0, h_0]$ for some $h_0 > 0$ and driven by the operator solution $\hat{U}_t$ to $h \partial_t \hat{U}_t = i \hat{Q}_h \hat{U}_t$ with initial condition, $\hat{U}_{t=0} = \hat{I}$.
Then, the evolution of $\hat{A} : L^2(\mathcal{M}) \to L^2(\mathcal{M})$ is given by $\hat{A}_t := \hat{U}_{-t} \hat{A} \hat{U}_t$, which solves $\partial_t \hat{A}_t = \frac{i}{h} [\hat{A}_t, \hat{Q}_h]$ with initial condition $\hat{A}_{t=0} = \hat{A}$.
To bring this to \emph{classical dynamics} on phase space, a connective link is Egorov's theorem \cite{zworski_semiclassical_2012}, which states roughly that if $\hat{Q}_h$ and $\hat{A}$ are the \emph{quantizations} of $q, a \in C^{\infty}(T^*\mathcal{M})$, respectively satisfying certain properties, then $\hat{A}_t$ is the quantization of $a \circ \Flow_q^t$ with $\Flow_q^t$ the Hamiltonian flow given by $q$.
A graphical description of our approach is shown in \cref{fig:schematic}.

Throughout this work, we will connect our discussion to concepts in physics. To aid with this we use the following Dirac notation and physics terminology \footnote{A comprehensive table of notation can be found in SI. Sec. I.}:  operators and quantum states are denoted as $\hat{A}$ and $\ket{\psi}$, respectively, and complex inner products are written as $\bra{\phi}\psi\rangle$, where $\bra{\phi}$ is the dual of $\ket{\phi}$, \ie this inner product is anti-linear in the first component. Further, operator expectations are denoted $\langle \hat{A} \rangle := \bra{\psi} \hat{A} \ket{\psi} = \langle \psi | \hat{A}\psi\rangle$, and we often use generalized distributions (in the sense of Schwartz) so that we may represent states in the position basis as $\psi(\bx) := \bra{\delta_\bx} \psi\rangle$, where $\ket{\delta_\bx}$ a point mass at $\bx \in \mathcal{M}$, hence allowing us to think of them as $L^2$ functions on the manifold.

This first step in our approach requires the notion of \emph{quantization}: as alluded to, this is a map that takes certain classes of smooth functions on phase space, called \emph{symbols}, to linear operators.
The resulting operators are then called \emph{pseudodifferential operators} (or shortly, $\PsiDO$s).
We refer to \cite{zworski_semiclassical_2012} for details, but to give a sense of symbols and $\PsiDO$s, consider the situation in $\mathbb{R}^d$.
On identifying $T^*\mathbb{R}^d \cong \mathbb{R}^{2d}$ for the purposes of analogy, consider a symbol $a \in C^{\infty}(\mathbb{R}^{2d} \times [0, h_0))$ for some $h_0 > 0$ that is Schwartz class in $\mathbb{R}^{2d}$.
Then, its \emph{quantization} is any of the \emph{continuous maps}
\begin{multline*} \label{eq:quantization}
	\Op_{h,\tau}(a) : \mathscr{S}' \ni u \to	\\
		\frac{1}{(2 \pi h)^d} \int_{\mathbb{R}^{2d}} e^{\frac{i}{h} \langle x - y, \xi \rangle} a(\tau x + (1 - \tau )y, \xi ; h) u(y) ~ dy d\xi \in \mathscr{S}
\end{multline*}
for $\mathscr{S}$ the Schwartz space of functions and $\mathscr{S}'$ the space of tempered distributions on $\mathbb{R}^d$ and $\tau \in [0, 1]$.
Particular cases of interest are the \emph{standard quantization} with $\tau = 0$, \emph{Weyl quantization} with $\tau = \nicefrac{1}{2}$ and \emph{right quantization} with $\tau = 1$.
The Schwartz class of symbols is too restrictive, since we lack even the differential operators here upon quantization; to remedy this, one defines a more general \emph{symbol class of order} $(k,m)$, which we denote by $h^k S^m$, to be given by functions $a \in C^{\infty}(\mathbb{R}^{2d})$ such that for all $\beta, \gamma \in \mathbb{N}^d_0$, there is a $C_{\beta,\gamma} \geq 0$ so that for all $h \in [0, h_0]$ and $(\bx,\bp) \in \mathbb{R}^{2d}$, $|\partial_{\bx}^{\gamma} \partial_{\bp}^{\beta} a| \leq C_{\beta,\gamma} h^{-k} (1+\vert\bp\vert^2)^{\nicefrac{(m - |\beta|)}{2}}$.
This symbol class is defined by a certain decay of higher order derivatives with respect to $\bp$, which ensures that the symbols have an asymptotic expansion at large $\vert\bp\vert$ that allows us to distinguish dominant parts of the symbol.
The corresponding image of $\Op_{h,\tau}$ is the space of $\PsiDO$s \emph{of order} $(k,m)$, which we write as $h^k \Psi^m$.

The agreement of two $\PsiDO$s up to \emph{a lower order} signifies that the symbol of their residual decays an order faster \emph{both}, in $h \to 0$ and in $|\bp| \to \infty$; so \eg $\PsiDO$s from two different quantizations, \emph{viz.}, different values of $\tau$, of a common symbol are equal up to a lower order $\PsiDO$.
Now, indeed, we have the polynomials in the symbol classes and in particular, the Laplacian (up to semi-classical scaling), $h^2 \Delta \in h^0 \Psi^2$.

A particularly useful concept is that of the \emph{principal symbol} of a $\PsiDO$ $\hat{A} \in h^k \Psi^m$, which is the equivalence class $[a_0] \in h^k S^m / h^{k + 1} S^{m - 1}$ such that for any $a \in [a_0]$, $\hat{A} - \Op_{h,t}(a) \in h^{k + 1} \Psi^{m - 1}$. Roughly, the principal symbol captures the leading order behavior in the limit $h\to 0$, which is the definition of the \emph{classical limit} in the \emph{quantum-classical correspondence}.
This can be understood as the operator's \emph{high-frequency} component in the following way: if $\phi \in C^{\infty}(\mathbb{R}^d)$ has derivative $d\phi|_{\bx_0} = -\bp_0$ and $q(\bx, \bp) = \sum_{|\beta| \leq D} q_{\beta}(\bx) h^{D-\vert\beta\vert}\bp^{\beta}$ is a symbol, then application of its quantization to a \emph{wave-packet} with phase $\phi$ leads to, $e^{-\frac{i}{h} \phi} \Op_{h,\tau}[q(\bx, \bp)](e^{\frac{i}{h} \phi})(\bx_0) \to \sum_{|\beta| = D} q_{\beta}(\bx_0) \bp_0^{\beta}$ as $h \to 0$.
Another key part of the relationship between $\PsiDO$s and symbols is that in as far as $\PsiDO$s commute up to a lower order operator, the principal symbol resulting from a product of $\PsiDO$s is the product of their principal symbols: $\hat{C} := \Op_h(a) \Op_h(b)$ has principal symbol $a_0 b_0$ with $a_0$ and $b_0$ the principal parts of $a$ and $b$, respectively.
These properties of the principal symbol of a $\PsiDO$ underlie the key relationship in the quantum-classical correspondence: principal symbols are the \emph{classical observables} corresponding to $\PsiDO$s that are interpreted as \emph{quantum observables}.

The semiclassical framework can be ported to $C^{\infty}$ boundaryless Riemannian manifolds by localizing the definition from Euclidean space to small balls, wherein the manifolds are roughly flat themselves and then patching them together via partitions of unity.
A key aspect of this is that the symbol classes $h^k S^m$ are invariant to coordinate transformations \cite[$\S 14$]{zworski_semiclassical_2012}.
Likewise, the concept of principal symbol can be realized by taking \emph{localized} wavepackets on the manifold and applying $\PsiDO$s to them as was sketched in the previous paragraph for general wavepackets in $\mathbb{R}^d$. This is expected since $\mathcal{M}$ locally resembles $\mathbb{R}^d$.
Concretely, we will use \emph{coherent states} to this effect: given $\zeta_0 \in T^*\mathcal{M}$, a \emph{coherent state localized at} $\zeta_0=(\bx_0,\bp_0)$ is the function
\begin{equation} \label{eq:coherent-state}
	\psi^h_{\zeta_0}(\bx) = \frac{1}{h^{\nicefrac{\nu}{4}}} e^{\frac{i}{h} \varphi(\bx, \zeta_0)}
\end{equation}
with $\varphi \in C^{\infty}(\mathcal{M})$ the restriction to $\zeta_0$ of an \emph{admissible phase} \cite{kumar_math} $\phi \in C^{\infty}(\mathcal{M} \times T^*\mathcal{M})$.
This admissibility condition entails roughly that in a neighbourhood of the diagonal $\{ (\bx;  \bx, \bp) \in \mathcal{M} \times T^*\mathcal{M} \}$, $\phi$ can be approximated in local coordinates by, $\langle p, x - y \rangle_{\mathbb{R}^{\nu}} + i ||x - y||^2_{\mathbb{R}^{\nu}}$.
In the context of its interactions with $\PsiDO$s, the \emph{phase space} properties of this function can be understood from its \emph{Husimi function}, as described in \cite{kumar_math}.
These states satisfy the properties of coherent states usually encountered in physics \cite{
gazeau_coherent_2009}; \ie they have minimum uncertainty volume in quantum phase space, satisfying the Heisenberg uncertainty principle, which is equally distributed between position and momentum degrees of freedom, defined through the $\Op_h(\cdot)$ quantization.

Connecting the semiclassical framework and Egorov's theorem with coherent states, the form of quantum-classical correspondence we will utilize is given by,
\begin{theorem}[Ref. \cite{kumar_math}] \label{thm:qc-correspondence}
Let $\hat{Q} = \Op_h(q) \in h^0 \Psi^{-\infty} := \cap_{m \in \mathbb{Z}} h^0 \Psi^m$ with $q \in C_c^{\infty}(T^*\mathcal{M})$ having real principal symbol $q_0$ and $a \in h^0 S^m$ over $h \in (0, h_0]$ for some $h_0 > 0$.
Then, given a coherent state $\psi_h$ localized at $\zeta_0 \in T^*\mathcal{M}$ and a fixed $T \geq 0$, we have for all $h \in (0, h_0]$ and $|t| \leq T$,
\begin{equation} \label{eq:qc-correspondence}
	|\langle \psi_h | e^{-\frac{i}{h} \hat{Q} t} \Op_h(a) e^{\frac{i}{h} \hat{Q} t} | \psi_h \rangle - a \circ \Flow_{q_0}^t(\zeta_0)| = \mathcal{O}(h),
\end{equation}
where $\Phi_{q_0}^t$ is the Hamiltonian flow on $T^*\mathcal{M}$ generated by $q_0$.
\end{theorem}

We often call $\hat{Q}$ the \emph{quantized Hamiltonian}, or even just the \emph{Hamiltonian} when the context is clear with regards to the classical case, in as far as it generates the dynamics on \emph{quantum observables} $\hat{A} \in h^0 \Psi^m$ through the propagator $\hat{U}_t := e^{\frac{i}{h} \hat{Q} t}$.
The expression \cref{eq:qc-correspondence} states that the time evolution of any observable of our quantized system approximates the classical Hamiltonian flow of the corresponding classical observable generated by the principal symbol of the quantized Hamiltonian $\hat{Q}$.
This approximation improves as $h\rightarrow 0$.
The validity of the relationship between the time evolution of the quantized observable and the related classical flow given by \cref{eq:qc-correspondence} is only valid until a fixed amount of time $T > 0$.
This suffices for our application, since we will propagate only for a time smaller than the \emph{injectivity radius} at a given point on $\mathcal{M}$.

The statement in \cref{eq:qc-correspondence} can be roughly understood by expanding the propagated quantum observable by application of the Baker-Campbell-Hausdorff formula:
\begin{align}
\label{eq:At_expansion}
	\hat{A}(t) &= e^{-\frac{i}{h}\hat{Q}t} \hat{A}  e^{\frac{i}{h}\hat{Q}t} = \hat{A} -\frac{it}{h} [\hat{Q}, \hat{A}] + \frac{t^2}{2h^2} [\hat{Q},[\hat{Q},\hat{A}]] + \dots,\nn
\end{align}
with $[\hat{Q},\hat{B}] \in h^{k+1}\Psi^{m-1}$, whenever $\hat{B} \in h^k \Psi^m$ due to the properties of quantization discussed above.
If we dequantize this operator to recover the corresponding propagated symbol, it is clear that even if $a$ has no dependence on $h$ the symbol at $t>0$ will due to the terms in this expansion. These are quantum corrections that decay as $h\to 0$ to yield the flow of the classical observable, $a\circ \Phi_{q_0}^t$.

\subsection{Quantization for manifold learning}
\label{sec:quant_manilearn}
We set some notation for use in the following. Finite-dimensional, data-dependent representations or approximations of operators, $\hat{A}$, and states, $\ket{\psi}$, are denoted $\breve{A}$ and $\vert \psi)$, respectively; $\breve{A} \in \mathbb{C}^{N\times N}$ and $\vert \psi ) \in \mathbb{C}^N$ for some dimension $N$. Expectation values for finite-dimensional approximations of quantum observables are denoted $\expect{\breve{A}} := (\psi \vert \breve{A}\vert\psi)$. $\breve{A}$ and $\vert \psi )$ will be represented in the position basis, unless otherwise specified. Finally, we introduce notation for a global parameterization of points on the manifold as $\iota(\bx) = (\iota_1(\bx), ..., \iota_n(\bx))$.

For manifold learning, a natural choice for the quantized Hamiltonian is $\hat{H}' = h^2 \Delta_g$, where $\Delta_g$ is the LB operator on $\mathcal{M}$, since the principal symbol for this operator is $\mathcal{H}' = |\bp|^2_g = \sum_{i,j} g^{ij}p_i p_j$, which is the classical Hamiltonian for free motion on the manifold.
Thus, the corresponding flow $\Phi_{\mathcal{H}'}^t$ projected onto $\mathcal{M}$ transports a particle along geodesics, however its \emph{speed} depends on the initial momentum magnitude.
On the other hand, if we demand that the quantized Hamiltonian for the system take the form
\begin{equation}
	\hat{H} = h \sqrt{\Delta_g} ,
	\label{eq:qHam}
\end{equation}
then the corresponding classical Hamiltonian is $\mathcal{H} := |\bp|_g$ and its flow, which we denote by $\Phi_t$, projected onto $\mathcal{M}$ transports particles along \emph{unit speed} geodesics, regardless of the initial (non-zero) momentum.
Motivated by recovering intrinsic distances and the exponential map and in support of further applications where this unit-speed regularization is fruitful, it is this flow we wish to approximate.

The form of quantum-classical correspondence given by Theorem \ref{thm:qc-correspondence} is particularly suited to data-driven approximations of the LB operator of a manifold.
A quantum system governed by a Hamiltonian $\hat{Q}$ whose principal symbol $q_0$ is localized in phase space and such that $q_0(\bx,\bp) \equiv |\bp|_{g_{\bx}}$ in a neighbourhood of $(\bx_0, \bp_0) \in T^*\mathcal{M}$ is immediately linked to the geodesic flow $\Phi_t(\bx_0, \bp_0)$ through \cref{eq:qc-correspondence}.
Diffusion operators related to \emph{graph Laplacians} now have a good probabilistic convergence theory from finite samples and we rely on the results of \cite{kumar_math} that convert them to propagators $\hat{U}_{q_0}^t := e^{\frac{i}{h} \hat{Q} t}$.
We now set the stage for implementing this procedure.

Consider the matrix
\begin{equation}
\sLap := \frac{2c_0}{c_2\epsilon}\left(\breve{I}_N - \breve{D}^{-1}_{\epsilon,\lambda}\breve{\Sigma}^{-\lambda}_{\epsilon} \breve{T}_{\epsilon}\breve{\Sigma}^{-\lambda}_{\epsilon}\right),
\label{eq:L_N}
\end{equation}
where $[\breve{T}_{\epsilon}]_{i,j} = k(\nicefrac{\lVert v_i - v_j \rVert^2}{2\epsilon})$, for $k:\mathbb{R}\rightarrow \mathbb{R}$ is a smooth function that exponentially decays in its argument. The diagonal matrices $[\breve{\Sigma}_{\epsilon}]_{i,i} = \sum_{j=1}^N [\breve{T}_{\epsilon}]_{i,j}$ and $[\breve{D}_{\epsilon,\lambda}]_{i,i} = \sum_{j=1}^N [\breve{\Sigma}^{-\lambda}_{\epsilon} \breve{T}_{\epsilon}\breve{\Sigma}^{-\lambda}_{\epsilon}]_{i,j}$ are normalizations and $c_j = \int_{\mathbb{R}^n} k(\vert\vert y \vert\vert^2) y_1^j dy$ for $j = 0, 2$ are constants that are trivial to compute upon a choice for $k$.
The scale parameter $\epsilon > 0$ regulates the order of approximation beyond the infinite sample limit and its optimal choice depends on the density of the samples $V$ on $\mathcal{M}$.
While there is no universal choice for $\epsilon$ at a given sample size, the asymptotic rates of convergence in \cref{prop:prop2} lend a rule-of-thumb.
As the construction of our propagator and its discrete approximation are rooted in diffusion processes, the intuitions for the $\epsilon$ parameter carry over from that literature \cite{nadler2006diffusion}:  one attempts to assign a value such that a $\sqrt{\epsilon}$-ball captures the local neighborhood of any data point.
Thus, $\epsilon$ is a measure of degree of uncertainty or resolution in configuration space data, and below we will show how it is intimately related to the fundamental phase space uncertainty $h$ imposed by quantum mechanics.
Finally, the $\lambda$ parameter that influences the asymptotic (in $N$) convergence properties of $\sLap$, which we now discuss.

The asymptotic properties of $\sLap$ have been thoroughly characterized in recent years,
including the important convergence result  \cite{hein2005graphs} that
$\sLap \xrightarrow{N\rightarrow \infty} \hat{\mathcal{L}}_{\epsilon,\lambda}$, where
\begin{equation}
\hat{\mathcal{L}}_{\epsilon,\lambda} = \left (\Delta_{g} + \frac{2(1-\lambda)}{\pd}\langle\nabla \pd, \nabla\rangle_{g} \right) + \mathcal{O}(\epsilon) .
\label{eq:L_eps}
\end{equation}
The leading term in this operator is the LB operator on the manifold, except for a correction term that depends on the sampling density $\pd$.
The correction term is zero under uniform sampling, and its influence can be tuned by choice of normalization parameter $\lambda$.
While the choice $\lambda=1$ asymptotically removes the dependence on sampling density, this is not necessary for our approach because Theorem \ref{thm:qc-correspondence} only depends on the principal symbol of the Hamiltonian generating the dynamics and these sampling density dependent terms enter as lower-order contributions to the symbol for $h^2 \hat{\mathcal{L}}_{\epsilon,\lambda}$.
Nevertheless, we leave this option for regularization open.

Although $\hat{\mathcal{L}}_{\epsilon,\lambda}$ approximates the LB operator on the manifold as $\epsilon \rightarrow 0$, this approximation is subtle from a quantization perspective since the $\mathcal{O}(\epsilon)$ terms include differential operators of increasing order.
Indeed, since the basis of the approximation (\cref{eq:L_eps}) is a first-order difference approximation to the infinitesimal of a diffusion operator in the parameter $\epsilon$, the corresponding \emph{symbol}, once viewed from the quantization perspective, only locally approximates $|\bp|^2_g$ in phase space \cite{kumar_math}.
The domain and quality of approximation grow as $\epsilon \to 0$ and in order to apply semiclassical methods, this must increase at a certain rate with $h \to 0$ as well.
Further, to use this approximation to obtain quantum dynamics with a well-characterized classical limit, we must view it as applied to coherent states, which as we discussed in the previous section are wave-packets localized in phase space.
The action of the asymptotic operator $\hat{\mathcal{L}}_{\epsilon,\lambda}$ on a coherent state in fact uncovers this crucial relationship between $\epsilon$ and $h$, as we now see.
\begin{prop}
Let $\alpha \geq 1$, $h > 0$, $\zeta_0 := (\bx_0, \bp_0) \in T^*\mathcal{M}$ and let $\psi_{\zeta_0}^h$ be a coherent state localized at $\zeta_0$. Then with $\epsilon := h^{(2 + \alpha)}$ we have,
\begin{align}
 (h^2 \hat{\mathcal{L}}_{\epsilon,\lambda}) \psi^h_{\zeta_0}(\bx_0) = |\bp_0|^2_{g(\bx_0)}\psi^h_{\zeta_0}(\bx_0) + \mathcal{O}(h).
\label{eq:L_coh}
\end{align}
\label{thm:thm1}
\end{prop}
\noindent\textit{Proof}: see Supplementary Information (SI), Sec. I.

\cref{eq:L_coh} states that, in the sense of the discussion in the previous section, the \emph{high-frequency} data in $h^2 \hat{\mathcal{L}}_{\epsilon,\lambda}$ is the symbol of the semiclassical LB operator, $\hat{H}'$ when the \emph{configuration space resolution parameter} $\epsilon$ decays as $h^{2 + \alpha}$ for some $\alpha \geq 1$ with respect to the \emph{phase space resolution parameter} $h$.
In other words, $h^2 \hat{\mathcal{L}}_{\epsilon,\lambda}$ acting on a coherent state approximates the action of the semiclassical LB operator on that coherent state, to order $O(h)$.

In \cite{kumar_math}, $h^2 \hat{\mathcal{L}}_{\epsilon,\lambda}$ is shown to be a $\PsiDO$ with this scaling of $\epsilon$ and $h$ and then, the action of Heisenberg dynamics governed by $\hat{V}_t := e^{i t \sqrt{\hat{\mathcal{L}}_{\epsilon,\lambda}}}$ on coherent states is reduced to the form satisfying Theorem \ref{thm:qc-correspondence}.
In practice we approximate expectations such as those on the left hand side of \cref{eq:qc-correspondence} by discrete inner products that are influenced by the sampling density $\pd$. Therefore, we must account for this density as a multiplier here in the asymptotic regime.
Altogether, we have the quantum-correspondence principle for graph Laplacians, in the asymptotic (continuum) limit:

\begin{theorem}[Ref. \cite{kumar_math}]
Let $\hat{A}$ be an observable for a quantized system and $a \in h^0 S^0$ its corresponding symbol. Let the quantized system have Hamiltonian $\hat{G} = h \sqrt{\hat{\mathcal{L}}_{\epsilon,\lambda}}$, with $h=\epsilon^{\nicefrac{1}{(2+\alpha)}}$ for some $\alpha\geq 1$. The corresponding propagator is $\hat{V}_t = e^{\frac{i}{h}\hat{G}t}$. Let $\Phi_t$ denote the classical Hamiltonian flow generated by the principal symbol of $\hat{H}$: $\mathcal{H}=\vert \bp \vert_g$ (\emph{i.e.}, the homogeneous geodesic flow) and let $\ket{\psi^h_{\zeta_0}}$ be a normalized coherent state localized at $\zeta_0 = (\bx_0, \bp_0)$.
Then, there is an $h_0 > 0$ such that for all $h \in (0, h_0]$ and $|t|$ less than the injectivity radius at $\bx_0$,
\begin{align}
    \Big\vert \bra{\psi_{t,\zeta_0}} \hat{\pd}\hat{A} \ket{\psi_{t,\zeta_0}} - a \circ \Phi_t(\bx_0,\bp_0) \Big\vert = \mathcal{O}(h),
	\label{eq:trace}
\end{align}
where $\ket{\psi_{t,\zeta_0}}:=\hat{V}_t\ket{\psi^h_{\zeta_0}}/\sqrt{\bra{\psi^h_{\zeta_0}}\hat{V}_t\dg\hat{\pd}\hat{V}_t\ket{\psi^h_{\zeta_0}}}$ is the normalized propagated state, and $\hat{\pd}$ is a diagonal operator representing the sampling density $\pd$.
\label{thm:thm2}
\end{theorem}
This theorem tells us that if we propagate coherent states for $|t| < \operatorname{inj}(\bx_0)$ with the asymptotic data-driven propagator, we can approximate observables of the classical Hamiltonian flow, and this approximation improves as $h\rightarrow 0$ ($\epsilon\rightarrow 0$). Any features of the classical flow can be reconstructed by suitable choices for $\hat{A}$. This result is reminiscent of a semiclassical expansion of quantum dynamics around a classical trajectory; \ie for coherent state initial states and short times, time-propagated quantum observables are equal to corresponding classical observables, with $\mathcal{O}(h)$ quantum corrections. Although such expansions are known in specialized settings \cite{Dias_Mikovic_Prata_2006, Polkovnikov_2010}, our result is (i) specific to the manifold learning context, where the generator of evolution is the data-driven approximation $\hat{G}$, and (ii) general in the sense that it applies to dynamics on any Riemannian manifold \emph{satisfying the assumptions of} Section \ref{sec:qdyn_geo}.
In practice, we define coordinates on the manifold, $u:\mathcal{M}\to \mathbb{R}^m$, and compute $\expect{\hat{u}_j} := \bra{\psi_{t,\zeta_0}} \hat{u}_j\hat{\pd} \ket{\psi_{t, \zeta_0}} = \int_\mathcal{M} u_j(\bx) \vert \psi_{t,\zeta_0}(\bx)\vert^2 \pd(\bx)d\bx$, where $\hat{u}_j=\Op_h(u_j)$, which as an operator is simply multiplication by $u_j$. The equality follows from the fact that $\hat{u}_j$ and $\hat{\pd}$ are diagonal operators in position basis.

\subsection{The role of $h$ in manifold learning}
\label{sec:role_h}
What is the role of $h$ in the quantum dynamics formulation of manifold learning? Unlike in physical theories, $h$ is not a constant but a function of the sampled dataset $V$.
According to the proof of \cref{thm:thm1}, setting $h \in \omega(\sqrt{\epsilon})$, for $\epsilon$ the scale parameter in the graph Laplacian construction of the heat kernel, results in a data-driven propagator that approximates the action of the LB operator on coherent states.
This $h(\epsilon)$ relation is intended to set the quantization scale of the quantum system to the resolution of the manifold provided by the sampled data $V$.
The rationale for this is motivated by the following: the graph Laplacian construction approximates a diffusion process that is discretized to $\sqrt{\epsilon}$ balls on the manifold.
This process determines the smallest scale to which we can resolve the manifold and in doing so, uses all of the available resolution provided by the sampled data on configuration space.
With our choice of coherent states as initial states, we redistribute this resolution equally between the position and momentum degrees of freedom.
Therefore, following the uncertainty principle, we must set the minimal phase space resolution for the quantum dynamical system to be above this, \ie $h = \omega(\sqrt{\epsilon})$.
From this perspective, the $h\rightarrow 0$ limit is both the classical limit of the quantum dynamics and also the limit where the sampled data fully covers the manifold.
The practice we have opted for here is to set $h = \epsilon^{\nicefrac{1}{(2 + \alpha)}}$ with $\alpha \geq 1$ as a way to balance errors that arise from coherent state delocalization and from approximating the dynamics, as discussed in the proof of \cref{thm:thm1}.

Another form of the uncertainty principle \cite{fefferman1983uncertainty} dictates that the restriction of a $\PsiDO$ to a region of phase space constitutes its restriction to a portion of the spectrum.
The $h$ parameter thus simultaneously plays the role of setting an implicit \emph{spectral resolution} scale: due to the localization of a coherent state, applying a $\PsiDO$ to it effectively restricts this action to a portion of the spectrum governed by the momentum and the scale $h$.
In the present context, \cite{kumar_math} shows that coherent states are implicitly driven, to $\mathcal{O}(h)$ error, by spectral truncations of the propagator $\hat{V}_t$.
That is, $|\langle \psi_{t,\zeta_0} | \hat{A} | \psi_{t,\zeta_0} \rangle - \langle \psi_{t,\zeta_0,\sigma} | \hat{A} | \psi_{t,\zeta_0,\sigma} \rangle| = \mathcal{O}(h)$ with $\ket{\psi_{t,\zeta_0,\sigma}} = \Pi_{\sigma} \hat{V}_t \Pi_{\sigma} \psi_{t,\zeta_0}$ and $\Pi_{\sigma}$ the spectral projector onto the eigenspaces of $\mathcal{L}_{\epsilon,\lambda}$ with eigenvalues at most $\sigma \sim |\bp_0|_{g(\bx_0)}^2/h^2$.
Thus, the coherent state evolves according to an energy scale set by $|\bp_0|_{g(\bx_0)}$ and $h$.

\section{Data-driven approximation of quantum dynamics}
\label{sec:data_dyn}
This section outlines how the program to extract geodesic distances through quantum propagation described in the previous section can be realized through approximation of the quantum propagator, initial states, and measurements using the point cloud dataset, $V$. Pseudocode describing the procedures in this section is given in the SI, Sec. VII.

\subsection{Propagator}
\label{sec:dd_prop}
We approximate the unitary propagator $\hat{V}_t$ through the graph Laplacian specified in \cref{eq:L_N}. We choose an exponential function for $k(\cdot)$, resulting in a Gaussian kernel $[\breve{T}_{\epsilon}]_{i,j}=e^{-\nicefrac{\vert\vert v_i-v_j\vert\vert^2}{2\epsilon}}$ (in which case, $c_0=\sqrt{\pi}$ and $c_2=\sqrt{\pi}/2$) and in the following, we will assume $\lambda=1$ normalization to reduce the error terms appearing in \cref{eq:L_eps}, for in practice we find it to perform slightly better than $\lambda = 0$, although the results are quite comparable. In the following, if we omit the $\lambda$ subscript, \eg $\breve{L}_{\epsilon}$ and $\hat{\mathcal{L}}_\epsilon$, this implies the choice $\lambda=1$.

In view of \cref{thm:thm2}, we form the matrix $\breve{G} = \sqrt{h^2\sLap}$ with $h=\epsilon^{\nicefrac{1}{(2+\alpha)}}$ for some $\alpha\geq 1$, and then the data-driven approximation of the quantum propagator $\hat{V}_t$ over a time $t$ is formed as $\breve{V}_{t} = e^{\frac{i}{h}\breve{G} t}$.
$\sLap$ is not a symmetric matrix, and therefore in order to compute the above functions of it, we perform a similarity transform of $\sLap$ to a symmetrized version  $\breve{D}^{\nicefrac{1}{2}}_{\epsilon,\lambda}\sLap \breve{D}^{-\nicefrac{1}{2}}_{\epsilon,\lambda}$, perform the required computations spectrally and invert this similarity transformation afterwards.

Note that $\breve{G}$ is an approximation of the Hamiltonian in the position basis of the quantized dynamical system. We can then propagate initial states expressed in the position basis, $\vert \psi_0 )$, as usual, $\vert \psi_t ) = \breve{V}_t \vert \psi_0 )$. Integrating classical dynamics to get time evolved quantities requires working with phase space variables, both positions and momenta. In contrast, by quantizing we are able to work entirely in the position basis, where the data is sampled, an obvious advantage of quantization for our manifold learning application.

\subsection{Initial states} \label{sec:state-prep}
Using the samples from $\mathcal{M}$ in $V$ we will construct approximations of the phase space localized coherent states, in \cref{eq:coherent-state}, when the expected position $\bx_0$ corresponds to a sample point $v_0 := \iota(\bx_0) \in V$.
The definition of a coherent state requires only satisfying the \emph{admissibilty} conditions given in \cite{kumar_math} for its phase and as discussed there, such a state can be readily prepared by using extrinsic coordinates or a local coordinate patch about $\bx_0$.
In practice, the prescribed phase may have a differential vector lying outside of the tangent space of $\iota(\bx_0)$.
Nevertheless, the procedures that follow give a phase sufficiently close to admissible \emph{with high probability} that this issue can be circumvented.

Since $v_i$ represents a point $i$ in a set of global, extrinsic coordinates, an approximation of a coherent state with expected position $\bx_0$ is given by a vector with elements
\begin{align}
	\left[\vert \tilde{\psi}^h_{\zeta_0} )\right]_{\ell} = e^{\frac{i}{h}(v_\ell-v_0)^{\sf T}p_0}e^{-\frac{\lVert v_\ell - v_0 \rVert ^2}{2h}}, \quad 1\leq \ell \leq N ,
	\label{eq:coh_extrinsic}
\end{align}
wherein $p_0 = (v_n-v_0)/||v_n-v_0|| \in \mathbb{R}^n$
is a vector in the direction of $v_n$, which is one of any fixed number $n_{\rm coll}$ nearest neighbors of $v_0$.
Here, the uncertainty parameter $h$ is determined by the choice of $\epsilon$ used in the data-driven construction of the propagator.
This is a finite-dimensional approximation of $\ket{\tilde{\psi}^h_{\tilde{\zeta}_0}} := \ket{e^{\frac{i}{h} \varphi(\cdot ; \tilde{\zeta}_0)}}$, with $\tilde{\zeta}_0 = (\bx_0=\iota^{-1}(v_0), \bp_0)$ and wherein $p_0$ approximates $\bp_0$, the tangent direction between $v_n$ and $v_0$.
We denote the normalized version of $\sket{\tilde{\psi}^h_{\zeta_0}}$ by $\sket{\psi^h_{\zeta_0}}$.

Alternatively, one can also form a coherent state using local information derived through a data-driven approximation of the tangent space at $\bx_0$.
An approach to this is the local principal component analysis procedure (LPCA) as described in \cite{singer2012vector}, which as shown there, yields an approximation to the tangent space up to second order in the size of the local neighborhood\footnote{An application of this feature is that by preparing the state in LPCA coordinates at $\bx_0$ and propagating to time $|t| < \operatorname{inj}(\bx_0)$, one can extend the approximation of normal coordinates from an $O(\sqrt{h})$ neighborhood to the whole of the sampled data within an $O(1)$ radius ball inside the normal neighbourhood at $\bx_0$.}, \emph{with high probability}.
 We perform LPCA on samples points within a Euclidean distance $\delta_{\rm PCA}$ from the  sample point $v_0$ corresponding to the initial point $\bx_0$; \ie on sample points in the set $\mathcal{S}_0 = \{\ell ~\vert~ \vert\vert v_0-v_\ell\vert\vert^2 \leq \delta_{\rm PCA} \}$. The dimension of the approximated tangent space is dictated by the decay of singular values of the LPCA covariance matrix, or can be set if the manifold dimension is known \emph{a priori}. LPCA defines a linear map $\breve{O}:\mathbb{R}^n \to \mathbb{R}^{\tilde{\nu}}$, where $\tilde{\nu}$ is the dimension of the LPCA-determined subspace. Let $\vartheta^{v_0}_\ell \in \mathbb{R}^{\tilde{\nu}}$ be the LPCA coordinates for all $v_\ell$ with $\ell \in \mathcal{S}_0$. Once the LPCA approximation to the tangent space is formed, the initial momentum can be chosen to be a vector in this space. Furthermore, we approximate geodesic distances with the Euclidean distances in LPCA space, and the inner product between tangent vectors that determines the phase of a coherent state is approximated by a Euclidean inner product.

This gives rise to the following initial state :
\begin{align}
	\left[\vert \tilde{\psi}^h_{\zeta_0} )\right]_{\ell} =
	\begin{cases}
		e^{\frac{i}{h}(\vartheta^{v_0}_\ell-\vartheta^{v_0}_0)^{\sf T}p_0}e^{-\frac{\lVert \vartheta^{v_0}_\ell - \vartheta^{v_0}_0 \rVert ^2}{2h}},&   \text{if} \quad \ell \in \mathcal{S}_0 \\
		0,& \text{otherwise},
	\end{cases}
	\label{eq:coherent_state_data}
\end{align}
with $p_0$ a unit vector in $\mathbb{R}^{\tilde{\nu}}$.
Altogether, this is a finite-dimensional, unnormalized LPCA approximation of $\ket{\psi^h_{\tilde{\zeta}_0}}$, with $\tilde{\zeta}_0=(\bx_0=\iota^{-1}(v_0), \bp_0)$, and $p_0$ approximating $\bp_0$ which is an orthogonal projection of $p_0$ onto the tangent space at $v_0$.

As before, the uncertainty parameter $h$ is determined by the choice of $\epsilon$ used in the data-driven construction of the propagator, and we denote the normalized version of this vector by $\vert \psi^h_{\zeta_0} )$. The propagation speed of this state is one due to the normalization of the momentum vector, and hence we often state distances as times in the following. An important point to note is that the formation of this initial state vector does not require computations with the original (possibly high-dimensional) samples.

\subsection{Measurements to extract geodesic distances} \label{sec:geodesic-measurements}

Following \cref{thm:thm2} and the discussion following it, we use the expected position of a propagated state $\vert \psi_{t, \zeta_0})$ to identify a point on $\mathcal{M}$ that is approximately a distance $t$ from the expected position of the initial coherent state. The most direct route to the expected position is through the global extrinsic coordinates (in a larger embedding space) provided by $v_i$.
Let $q_\ell := \vert[\vert \psi_{t,\zeta_0})]_\ell\vert^2, ~ 1\leq\ell\leq N$, be the probability distribution over data points determined by the propagated state.
Then, the mean position of the propagated state is defined as the sample point $\expect{\breve{\bx}_t} \in V$ that minimizes the Euclidean distance in $\mathbb{R}^n$ to the empirical mean $\bar{v}=\sum_{\ell} q_\ell v_l$.

As in the formation of the initial state, the localization of the propagated wavepacket also enables approximation of the expected position without global computations involving the high-dimensional extrinsic coordinates if we proceed through an LPCA approximation of the local neighborhood of the \emph{maximum} of the propagated state.
Namely, we perform LPCA around datapoint $v_{\ell^*}$, where $\ell^* = \arg\max_{\ell} q_{\ell}$, using all data points in $\mathcal{S}_{\ell^*}= \{\ell~\vert~ \vert\vert v_{\ell^*}-v_\ell\vert\vert^2 \leq \delta_{\rm PCA} \}$.
This LPCA computation is decoupled from that of the initial state preparation, hence this neighbourhood size parameter $\delta_{\rm PCA}$ and the resulting coordinate map $\breve{O}$ may also be different from that of the previous section.
The mean position is then defined as the data point in LPCA coordinates, $\expect{\breve{\bx}_t}^{\LPCA} \in \breve{O}[V]$ that minimizes the Euclidean distance in LPCA space to the empirical mean position $\bar{\vartheta} = \sum_{\ell\in\mathcal{S}_{\ell^*}} q_\ell \vartheta^{v_{\ell^*}}_\ell$, where $\vartheta^{v_{\ell^*}}_\ell$ are LPCA coordinates for the data point $\ell$.

The above procedures yield a data point corresponding to $\breve{\bx}_t \in \{ \iota^{-1}(\expect{\breve{\bx}_t}),  \iota^{-1} \circ \breve{O}^{\mathsf T} \expect{\breve{\bx}_t}^{\LPCA} \} \subset \mathcal{M}$.
They are justified by the fact that for $h$ following a certain decay rate, we have with high probability that $\breve{\bx}_t$ is \emph{intrinsically} within distance $h$ from the geodesic point $\bx_t := \pi_{\mathcal{M}} \Flow_t(\bx_0, \bp_0)$: this is,
\begin{prop}
	\label{prop:prop2}
	Let $\beta := (2 + \alpha)(\nu + 1) + 1$ and $\beta_0 := (2 + \alpha)(\nicefrac{5 \nu}{2} + 4) + 2(\nu + 2)$.
	Then, $\exists ~ h_0 > 0$ such that
	\begin{enumerate}
	\item	if the initial state is given by \cref{eq:coh_extrinsic} and $N^{-\nicefrac{1}{\gamma_0}} \lesssim h < h_0$, then with high probability, $d_g(\breve{\bx}_t, \bx_t) < h$ and
	\item	if the initial state is given by \cref{eq:coherent_state_data} and $N^{-\nicefrac{1}{\gamma_1}} \lesssim h < h_0$, then with high probability, $d_g(\breve{\bx}_t, \bx_t) < h$,
	\end{enumerate}
	wherein $\gamma_0 := \max\{ \nu(\nu/4 + 1 + \beta), \beta_0 \}$ and $\gamma_1 := \max\{ (\nu + 2)(\nu/4 + 1 + \beta)/2, \beta_0 \}$ whenever $\nu \geq 2$ and if $\nu = 1$ is possible, these maxima must be taken with $3\nu/2 + 2\beta$.
\end{prop}
\begin{proof}
The proof is given in the SI, Sec. II and relies on the convergence properties established in \cite{kumar_math}.
Essentially, the consistency can be applied almost directly, save for that our initial states have phases that are some \emph{perturbations} away from being admissible.
We show that with high probability, both the extrinsic and LPCA formulations of the phases are sufficiently close to admissible such that the corresponding initial states and their propagations with respect to $\breve{V}_t$ are $\mathcal{O}(h)$ perturbations of the propagations of coherent states.
Then, the properties of localization and means of the propagations of coherent states given in \cite{kumar_math} can be applied directly.
\end{proof}

If we proceed through LPCA, we are using Euclidean approximations of local neighborhoods of the manifold for construction of the initial state and for approximating the expected position of the propagated state. However, since the propagator that connects the local neighborhood of the initial state, $\vert\psi_0)$, and the final state, $\vert \psi_t)$, is global we expect to go beyond linear methods. To confirm this, in \cref{sec:eg} we compare geodesic distances extracted using quantum dynamics to those extracted using a Euclidean approximation of propagation along geodesics (\emph{Euclidean propagation}); specifically, we use the Euclidean parametrization of the neighborhood around $v_0$ provided by LPCA to determine the sampled data point lying closest (in Euclidean norm in LPCA coordinates) to a point $\vartheta^{v_0}_0 + p_0 t$. This point is denoted $\bx_t^{E}$.

Finally, we mention that while the expected value of the position coordinate over the probability distribution determined by the propagated state is the best estimate of position along the geodesic path, due to the localization of the propagated coherent state the extrinsic coordinates of the maximizer $\ell^* := \arg\max_{\ell} q_{\ell}$, are often also a good estimate of the quantities of interest, $\iota_i \circ \Phi_t(\bx_0, \bp_0)$ \cite{kumar_math}.
We have found this to be useful in cases where $N$ is small (see SI, Sec. VIII).

\subsection{Choice of parameters}
\label{sec:deviation}
The computational procedure we have outlined in this section requires choice of several parameters, the most important of which are $\epsilon$ and $\alpha$.
Appropriate choice of both of these parameters is critical to accurate estimation of geodesic distances, and consequently, manifold learning.
This is a familiar situation in many manifold learning approaches and often there are data-driven heuristics to guide the choice of optimal parameters; \eg see Ref. \cite{coifman_random_2008} for an approach to choosing the $\epsilon$ scale parameter in
graph Laplacian applications.
We now show that one can utilize the result in \cref{thm:thm1} to define a procedure for identifying regions in $(\epsilon,\alpha)$ space that lead to good performance for a given dataset.

\cref{thm:thm1} and the semiclassical representation of $\hat{\mathcal{L}}_{\epsilon,\lambda}$ imply the following about its expectation under a coherent state  $h^2 \bra{\psi^h_{\zeta_0}} \hat{\mathcal{L}}_{\epsilon,\lambda}\ket{\psi^h_{\zeta_0}} = \vert\bp_0\vert^2 + \mathcal{O}(h)$ \cite{kumar_math}. If we choose normalized momenta for initial states, this expectation $\approx 1$. Therefore, a heuristic for choosing $\epsilon$ and $\alpha$ is to compute in the finite-dimensional setting, the deviation $\mathcal{D}=\vert h^2(\psi^h_{\zeta_0}\vert \sLap \vert \psi^h_{\zeta_0})-1\vert$ and identify the viable $(\epsilon,\alpha)$ pairs as the ones where this quantity is small. Note that this only requires construction of the matrix $\sLap$ and initial state vector, no spectral computations or propagations. Given that $\vert\psi_{\zeta_0}^h)$ is a restriction of a coherent state to the samples, this deviation condition is effectively evaluating the closeness of matrix elements elements of $\sLap$ (in the $\mathcal{O}(h)$ orthonormal coherent state basis) to the asymptotic quantity $\hat{\mathcal{L}}_{\epsilon,\lambda}$. In practice, we recommend computing the average of this deviation over initial states centered at several sampled data points or to fix the initial state location ($\bx_0$) and to compute $\sLap$ using several samples from the dataset.

\section{Illustrations}
\label{sec:eg}
In this section we demonstrate geodesic distance extraction using quantum propagation, and its application for visualization, on some sample datasets. We use LPCA-based constructions of initial states and position expectations.

\begin{figure*}[t]
\centering
\includegraphics[width=1\linewidth]{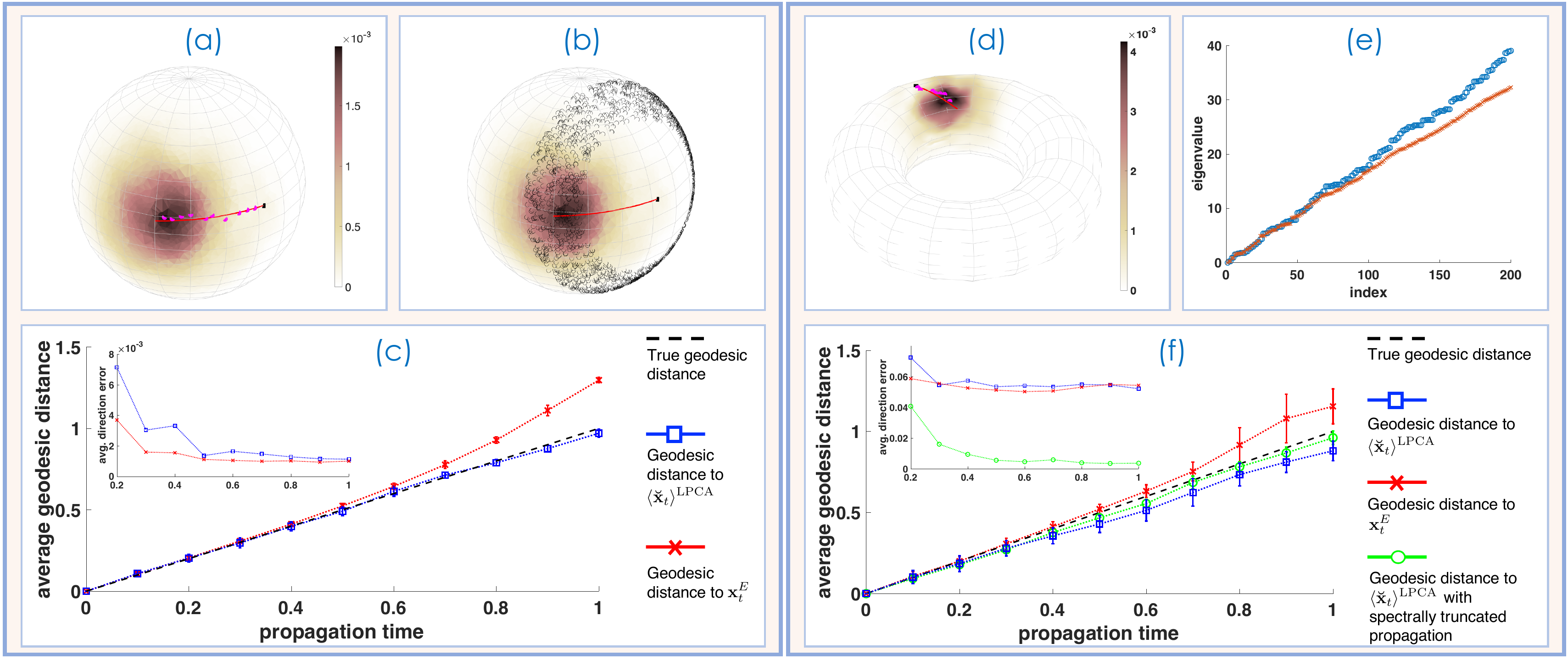}
\caption{Illustration of geodesic extraction through quantum dynamics on two sample manifolds: the sphere ((a)-(c)) and torus ((d)-(f)). In both cases, LPCA is used to construct initial states and compute expected position. \textbf{(a)} Example propagation of a coherent state on the unit sphere with $N=8000$ sample points and $\epsilon=e^{-4.7}, \alpha=1.6, \delta_{\rm PCA}=1.5$. Initial expected position is indicated by the black square. The point cloud is colored according to $q_\ell$, and triangulated using Delaunay triangulation for visualization purposes. The red curve is the true unit-speed geodesic curve up to $t=1$ with the specified initial state, and the magenta dots indicate mean positions determined through data-driven coherent state propagation for ten steps of $\Delta t=0.1$. \textbf{(b)} The propagated state at $t=1$, with the black circles indicating sample points $v_\ell$ such that $d_g(v_\ell, v_0)\leq h$, to show that the propagated state is localized in a region dictated by $h$. \textbf{(c)} The average geodesic distance of the propagated state from the initial point as a function of propagation time. The average is taken over $m=50$ propagations and the spread around the average is shown using one standard deviation error bars. The blue squares correspond to statistics of the geodesic distance to the point extracted through quantum dynamics and the red crosses correspond to Euclidean propagation (see main text). The inset shows the average error in propagation direction as a function of time. \textbf{(d)} Example propagation of a coherent state on a torus with radii ($r=0.8, R=2$) with $N=12000$ sample points and $\epsilon=e^{-3.8}, \alpha=1.8, \delta_{\rm PCA}=2.0$. The point cloud is colored according to $q_\ell$, and an $\alpha$-shape triangulation of the point cloud is used for visualization. The red curve is the true unit-speed geodesic curve up to $t=1$ with the specified initial state, and the magenta dots indicate mean positions determined through data-driven coherent state propagation for ten steps of $\Delta t=0.1$.  \textbf{(e)} A comparison of the first 200 eigenvalues of the true LB operator on the torus (blue circles) and the data-driven reconstruction (red crosses). \textbf{(f)} The average geodesic distance of the propagated state from the initial point as a function of propagation time. The average is taken over $m=50$ propagations and the spread around the average is shown using one standard deviation error bars. The blue squares correspond to statistics of the geodesic distance to the point extracted through quantum dynamics and the red crosses correspond to Euclidean propagation (see main text). The green circles correspond to statistics of geodesic distance calculated using a spectrally truncated quantum propagator. The inset shows the average error in propagation direction on the torus.
\label{fig:sphere_and_torus}}
\end{figure*}

\subsection{2-sphere}
We uniformly sample $N=8000$ points from the unit 2-sphere embedded in $\mathbb{R}^3$ ($n=3, \nu=2$). The quantum propagator $\breve{V}_{\Delta t}$ is formed with $\Delta t=0.1, \epsilon=e^{-4.7}$. \cref{fig:sphere_and_torus}(a) shows an example propagation of an initial coherent state ($\alpha=1.6$) localized a point (black square) with some initial momentum.
The red curve shows the true unit-speed geodesic path, obtained by solving the geodesic equation for the 2-sphere, with the given initial conditions, and the magenta dots show $\expect{\breve{\bx}_t}^{\rm LPCA}$ for $t=i \Delta t$ for integer $i\leq 10$. The color of the point cloud encodes the propagated coherent state $\vert \psi_{t=1, \zeta_0} )$, with the data points colored according to the probability mass at the point, $q_\ell$. The state remains localized as it propagates along the geodesic, and it is concentrated on points a geodesic distance $t$ from the initial point. In \cref{fig:sphere_and_torus}(b) we show a snapshot of the propagated state at $t=1$, with the black circles highlighting sample points $v_\ell$ satisfying $d_g(v_\ell, v_0)<h$ -- where $d_g$ is the geodesic distance on the sphere and $v_0$ is the sample the initial state is localized on -- which clearly shows that the delocalization of $\vert \psi_{t,\zeta_0})$ is set by $h$.

To quantify how well the propagated coherent states track geodesics, we repeat such propagations $m=50$ times with uniformly sampled initial points and initial momenta toward the nearest neighboring point in each case. At each $t=i\Delta t$, we compute $\expect{\breve{\bx}_t}^{\rm LPCA}$ and $\bx_t^{E}$. \cref{fig:sphere_and_torus}(c) shows the average and one standard deviation error bars for the geodesic distance of these positions to the initial point. We see that $\expect{\breve{\bx}_t}^{\rm LPCA}$ corresponds to points that are approximately a geodesic distance $t$ away on average, with very small variation. In contrast, $\bx_t^{E}$, the points computed by the na\"ive local Euclidean approximation of propagation along geodesics begin to deviate from geodesic distance $t$ for $t>0.6$.
The inset to \cref{fig:sphere_and_torus}(c) also shows the average error in direction of propagation for $\expect{\breve{\bx}_t}^{\rm LPCA}$ and $\bx_t^{E}$. For each of the random propagations, this error is computed as $\Delta_{\rm dir}(t) = \vert 1 - p_0^T g_{v_0} p'_0(t) \vert$, where $p_0$ is the initial momentum dictating propagation direction, $g_{v_0}$ is the metric tensor for the 2-sphere at the initial point $v_0$, and $p'_0(t)$ is the initial momentum that results in propagation to the point $\expect{\breve{\bx}_t}^{\rm LPCA}$ or $\bx_t^{E}$. This quantity is just the misalignment between the true momentum vector at $t=0$ and the momentum vector required to propagate to the point determined by either procedure. We evaluate $\Delta_{\rm dir}(t)$ away from $t=0$ since at very short times because of the poor resolution of the manifold given by the finite number of samples, this error can be large; \ie there may not be a data point directly on the geodesic curve and the closest point requires a significantly different initial momentum to propagate to with small $t$. This is an artifact of finite sampling that becomes negligible at larger $t$, as shown by \cref{fig:sphere_and_torus}(c). 

\subsection{Torus}
The sphere is a relatively simple manifold with constant curvature. We next illustrate our method on a more complex manifold with positive and negative curvature, the torus embedded in $\mathbb{R}^3$ ($n=3, \nu=2$). We uniformly sample $N=12000$ points on a torus with radii $r=0.8, R=2$ and form the quantum propagator $\breve{V}_{\Delta t}$ with $\Delta t=0.1, \epsilon = e^{-3.8}$. \cref{fig:sphere_and_torus}(d) shows an example propagation of an initial coherent state (constructed using LPCA, $\alpha=1.8$). The red curve shows the true unit-speed geodesic path, obtained by solving the geodesic equation for the torus, with the given initial conditions, and the magenta dots show mean positions determined through data-driven coherent state propagation for ten steps of $\Delta t$. The color of the point cloud encodes the propagated coherent state $\vert \psi_{t=1,\zeta_0} )$, with the data points colored according to the probability mass at the point, $q_\ell$. To understand more systematically the quality of geodesic distance extraction, we repeat $m=50$ such propagations from uniformly sampled initial points (with initial momenta toward the nearest neighboring point in each case) and show the statistics of the geodesic distance to the extracted point, $\expect{\breve{\bx}_t}^{\rm LPCA}$, as a function of propagation time $t$ in \cref{fig:sphere_and_torus}(f). The blue squares show the average and one standard deviation variation of the geodesic distance extracted up to $t=1$. For comparison we also show the statistics of geodesic distance to the local Euclidean approximation of geodesic propagation $\bx_t^{E}$. The inset in \cref{fig:sphere_and_torus}(f) also shows the error in propagation direction (quantified in the same way as for the sphere example) for the two propagation methods.

We see that the quantum propagation incurs slightly more error for the torus than for the sphere, however, the average error is still $\mathcal{O}(h)$. 
The discussion in \cref{sec:role_h} implies that the propagation of a coherent state with scale set by $h$ is dictated by a restricted portion of the spectrum of $e^{\frac{i}{h}\hat{H}t}$. Thus, to see the potential of the approach with larger number of data samples, it makes sense to compare to propagation under a more accurate spectral truncation of $\hat{H}$.
In SI, Sec. V we use intrinsic information about the torus to directly construct the LB operator for the torus by numerical diagonalization. We truncate the spectral data at 800 total eigenvalues \footnote{The magnitude of the largest eigenvalue used is $156.5$.} (including degeneracies) and eigenvectors and form a \emph{spectrally truncated propagator}, $\tilde{V}_{\Delta t}$, defined over a regular discretization of the torus over $N'=51300$ points.
\cref{fig:sphere_and_torus}(e) compares the first 200 eigenvalues of the data-driven LB reconstruction, and the more accurate reconstruction through numerical diagonalization, and clearly shows the disagreement between spectra after $\sim 100$ eigenvalues.
We then propagate initial coherent states (with $h=e^{-1}$ as above), extract points closest to the expected position of the propagated states, $\expect{\breve{\bx}_t}^{\rm LPCA}$, using the same approach described above, and the statistics of the geodesic distance to these points are shown in \cref{fig:sphere_and_torus}(f) as green circles. The results show extremely low error in geodesic distance of $\expect{\breve{\bx}_t}^{\rm LPCA}$, and low errors in propagation direction (green circles in the inset of \cref{fig:sphere_and_torus}(f)). We note that the reduction in error in direction is also a result of a better estimation of the tangent space at a point due to the larger number of data samples.
Since the phase-space localization $h$ is fixed throughout, we have here an example in which truncation to a better spectral approximation yields tighter geodesic approximation. By the spectral convergence of the graph Laplacian to the LB operator \cite{Trillos_Gerlach_Hein_Slepcev_2018}, we conclude that with more data we could reproduce the extremely low error geodesic extraction demonstrated with $\tilde{V}_{\Delta t}$.

For completeness we also present spectral data for the graph Laplacian for the sphere example in SI, Sec. IV, and in SI, Sec. IX we plot the deviation measure, $\mathcal{D}$, used to inform parameter choices (see \cref{sec:deviation}) for the sphere and torus datasets and indicate the ($\epsilon,\alpha$) values used above.

\subsection{Using geodesic distances to embed}
\label{sec:embed}

Now we demonstrate how the geodesic distances extracted using quantum dynamics can be used to embed datasets in Euclidean space for visualization. To do so, for each sample point $v_\ell$ in the dataset, we propagate $n_{\rm coll}$ coherent states whose position coordinates are initially localized on $v_\ell$.  The coherent states are constructed with momenta in orthogonal directions in the tangent space at $v_\ell$, estimated using LPCA, as detailed above.
Each state in the collection will propagate along a geodesic.
The propagation is done in steps of $\Delta t$, by applying $\breve{U}_{\Delta t}$, for $n_{\rm prop}$ steps. For each step, this enables the extraction of a \emph{geodesic spray} consisting of $n_{\rm coll}$ sample points that are approximately a geodesic distance $i\times \Delta t$ for $1 \leq i \leq n_{\rm prop}$ from the initial state. We populate a geodesic distance matrix $G$ using this data \footnote{The matrix $G$ is a function of ${N,\epsilon, \alpha, \Delta t, n_{\rm prop}, n_{\rm coll},\delta_{\rm PCA}}$, but we do not explicitly notate this dependence for conciseness.} and this distance matrix is used to embed the sample points in three dimensions using force-directed layout \cite{Fruchterman_Reingold_1991}. The matrix $G$ can be considered a sparsified version of the original distance matrix $\breve{T}_{\epsilon}$ based on Euclidean distances, which typically has no strictly zero entries. The embedding based on force-directed layout is not unique, and a potential direction for future research is the study of optimal embedding techniques given the approximate geodesic distances output by quantum propagation.

We follow the above procedure for the sphere and torus datasets and in \cref{fig:embeddings} we show the resulting embeddings. See caption for the parameters used in both cases. In both cases, the reconstruction is a good representation of the original manifold. Of course, there is no dimensional reduction in this case, however the reconstructed embedding clearly reproduces the geometry of the original point cloud dataset.

\begin{figure}
\centering
\includegraphics[width=1\linewidth]{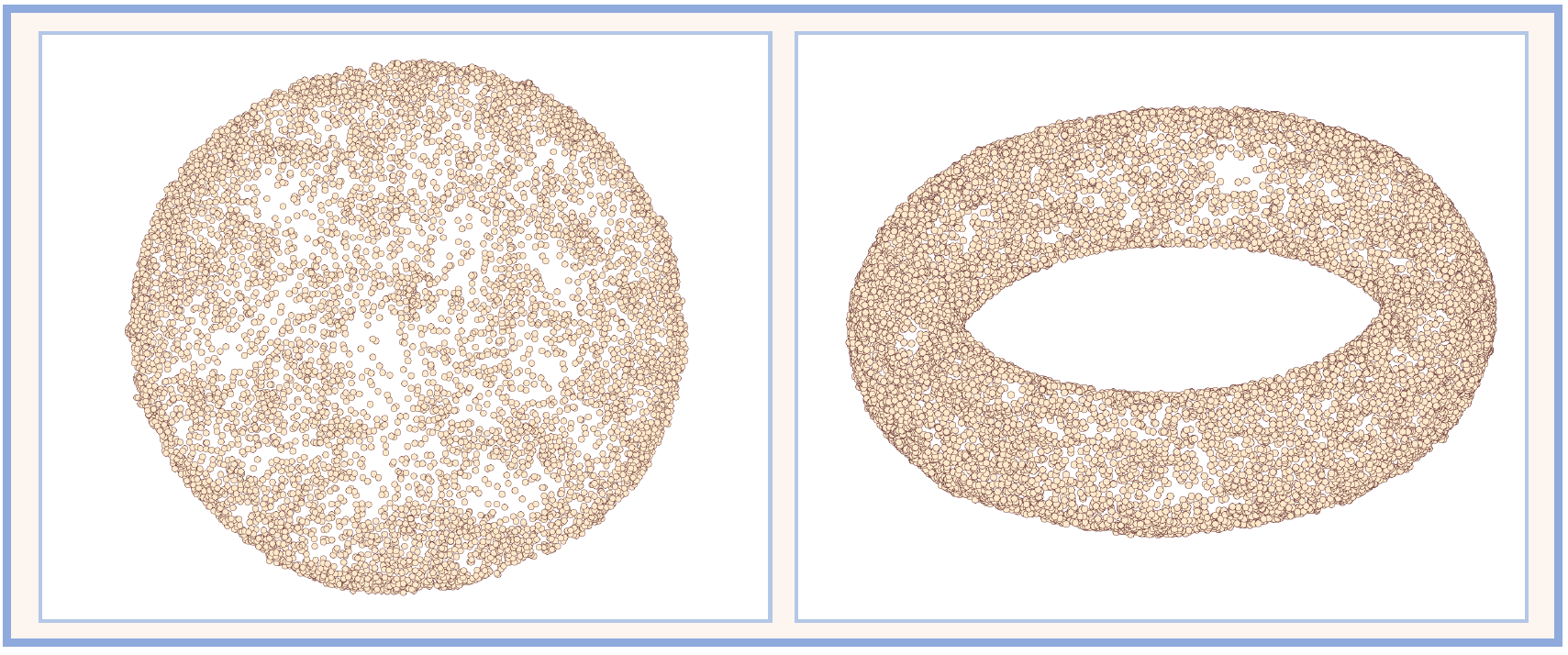}
\caption{
3D embeddings of samples from a sphere and torus, using force-based graph layout based on geodesic distances extracted using our approach. The data is described in Sec. \ref{sec:eg}. The parameters used are: (\textbf{sphere:} $\epsilon=e^{-4.7}, \alpha=1.6, \delta_{\rm PCA}=1.5, \Delta t=0.1, n_{\rm prop}=5,  n_{\rm coll}=40$), (\textbf{torus:} $\epsilon=e^{-3.8}, \alpha=1.8, \delta_{\rm PCA}=2.0, \Delta t=0.4, n_{\rm prop}=1, n_{\rm coll}=40$).}
\label{fig:embeddings}
\end{figure}

\section{An application: COVID-19 mobility data}
\label{sec:app}
We now present an application of our method to a real-world dataset, and demonstrate its utility for visualization and clustering of high-dimensional data. The global COVID-19 pandemic has resulted in massive and unprecedented disruptions to patterns of daily life. Analyzing changes in mobility during the past year has led to insights into social distancing patterns and infection modalities that can inform public health policy and management \cite{Sulyok_Walker_2020, Levin_Chao_Wenger_Proctor_2020, Chang_Pierson_Koh_Gerardin_Redbird_Grusky_Leskovec_2021, Cot_Cacciapaglia_Sannino_2021}. Geolocation information from mobile phones provides a wealth of mobility data, and has been the basis of most of the previous studies on mobility changes during the pandemic. This data can be complex and extremely high dimensional, and traditional manifold learning techniques have been shown to be informative for visualization and clustering of this data \cite{Levin_Chao_Wenger_Proctor_2020}.

In this section, we analyze aggregated mobility data in the Social Distancing Metric dataset from SafeGraph Inc. \cite{safegraph}. This is a fine-grained dataset that collects geolocation information from mobile devices, aggregates it at the census block group (CBG) level in the United States and records it daily for a period of over a year. It contains a wealth of information and we extract a simple metric to gauge the mobility patterns of citizens: we compute a daily stay-at-home (SAH) fraction for each CBG that is the ratio of the \verb|completely_home_device_count| to the \verb|device_count| values. The latter is the number of devices seen during the date whose home is determined to be in the CBG being considered, and the former is the number of these devices that were determined to not have left the home CBG during the day \cite{safegraph}. The SAH fraction for a date provides a measure of how curtailed mobility was within a CBG. To allow comparison to an earlier study with the same dataset \cite{Levin_Chao_Wenger_Proctor_2020} we limit the data to the 117-day time period from February 23, 2020 to June 19, 2020, which provides a snapshot of mobility patterns during the first three months of the pandemic.

In the following, we present results for analysis of mobility data for the state of Georgia (GA). After removing 17 CBGs with poor quality data, there are 5509 CBGs within the state. Therefore, our dataset has $N=5509$ samples, each of dimension 117. We minmax normalize the data samples so that the minimum SAH fraction is zero and maximum is one.
We performed quantum manifold learning on this dataset with parameters $\epsilon = e^{-\nicefrac{1}{2}}, \alpha=1.4, \Delta t=0.1, n_{\rm prop}=20, n_{\rm coll}=60$. See SI, Sec. IX for a plot of the deviation measure $\mathcal{D}$ for this dataset that informs these parameter choices. We compute initial states and $\expect{\breve{\bx}_t}$ using the original extrinsic coordinates and do not employ LPCA for simplicity and since the data dimension is not large enough to prohibit this. After extracting a geodesic distance matrix $G$, we embed the data in three dimensions using a force-directed layout using $G$. The data is then clustered into 5 clusters an in Ref. \cite{Levin_Chao_Wenger_Proctor_2020}. We use the k-means algorithm for clustering with the coordinates of the 3D embedding, as opposed to Gaussian mixture model (GMM) clustering as was done in Ref. \cite{Levin_Chao_Wenger_Proctor_2020}, although the results are similar with GMM clustering, see SI, Sec. VII for details.

Fig. \ref{fig:covid}(a) shows the embedding in 3D with the clusters indicated by colors. In addition, the average SAH fraction time series for each cluster are shown in \cref{fig:covid}(b), which shows that the clustering is meaningful; \ie although almost all CBGs exhibit a similar mobility pattern, the clusters reflect different amounts of overall SAH fraction. The average mobility traces of the clusters are clearly separated. Finally, \cref{fig:covid}(c) shows a map of the CBGs, with the same color coding of clusters as in the other subfigures. This geographic representation clearly shows the rural-urban divide in degree of mobility change during the pandemic; the higher SAH fraction clusters are associated with the urban centers in GA, while the majority of the state exhibited behavior consistent with lower SAH fraction curves (light blue and brown clusters).

\begin{figure*}[t]
\centering
\fbox{\includegraphics[width=0.95\linewidth]{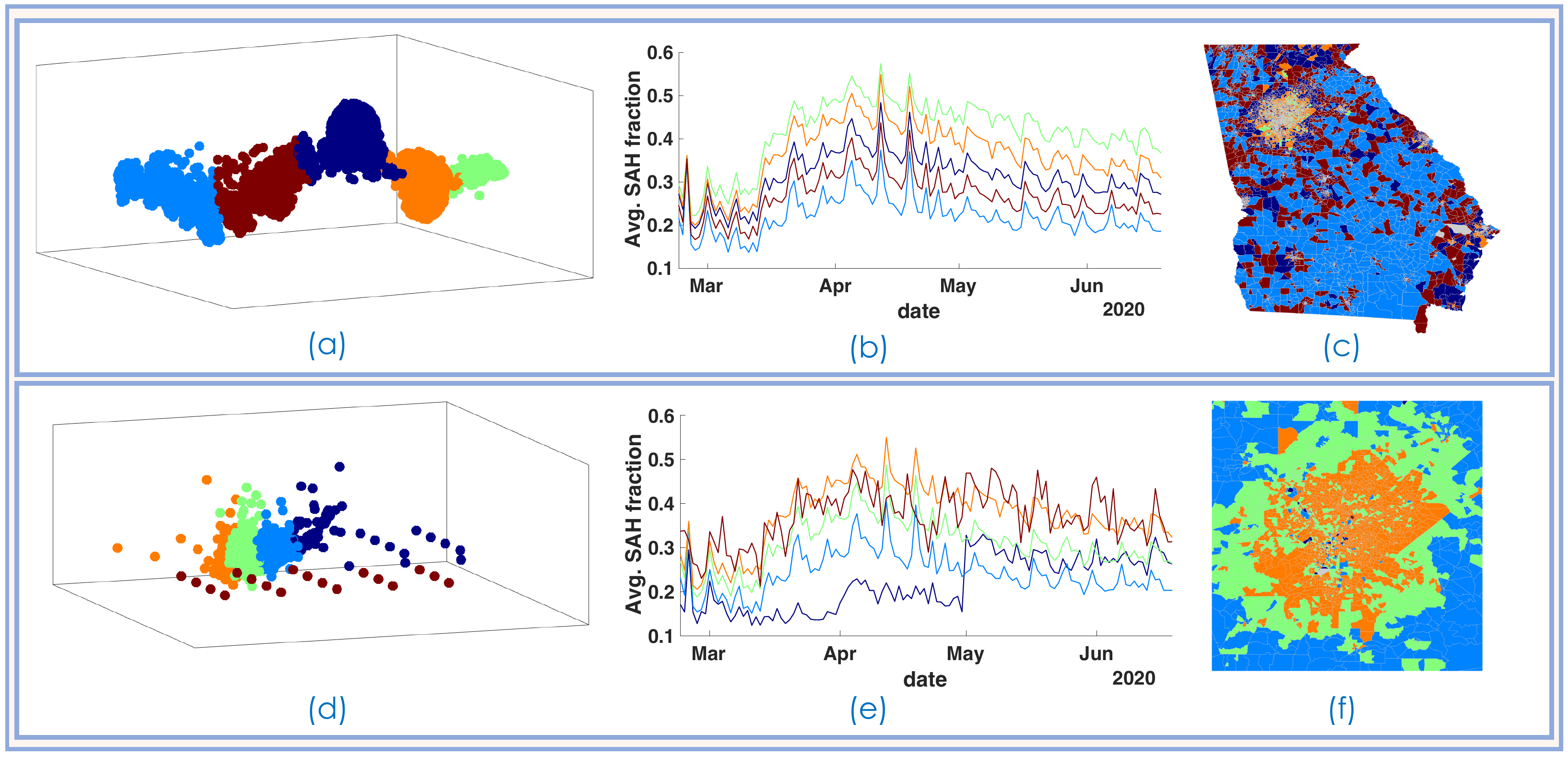}}
\caption{An application of geodesic extraction through quantum dynamics on mobility data from the US state of Georgia (GA). \textbf{(a)} Force-based graph layout in $\mathbb{R}^3$ of CBGs based on extracted geodesic distance matrix, where the location of the propagated states is calculated using the mean position $\expect{\breve{\bx}_t}$. The 3D coordinates from this embedding are used to cluster the CBGs using k-means clustering, and the five clusters are marked by color. \textbf{(b)} The average SAH fraction time series for each of the clusters in (a), with the same color-coding. \textbf{(c)} Geographic location of the clusters in GA. CBGs colored in gray are ones for which there was missing or poor quality data -- these were excluded from the analysis. \textbf{(d)} An $\mathbb{R}^3$ embedding of the same data as in (a), but with the location of the propagated state calculated using the maximum of the associated distribution over data points. Again, clustering is performed using k-means on the embedded coordinates. Poorly connected outlier CBGs are indicated by the brown and dark blue points. \textbf{(b)} Average SAH fraction time series for each of the clusters in (d), with the same color-coding. \textbf{(f)} Geographic location of the CBG clusters in the Atlanta metropolitan area. We present a zoom in into this area because some of the outlier CBGs identified in (d) are located in this area.}
\label{fig:covid}
\end{figure*}

It is instructive to compare the results in \cref{fig:covid} to those in Ref. \cite{Levin_Chao_Wenger_Proctor_2020}, where a similar clustering analysis was done with Laplacian eigenmaps.
In that work, the authors needed to embed into 14 dimensions in order to obtain good clustering, whereas we obtained meaningful clusters that represent distinct classes of SAH degree through a greater dimensional reduction to three dimensions.
An explanation for this might be offered upon looking at the asymptotic situation for manifolds: then, generally speaking, the Laplacian eigenmaps embedding places an eigenfunction of $\hat{\mathcal{L}}_{\epsilon,\lambda}$ at each coordinate axis directly and hence relies on the clustering method to pick up on the geometry exposed by the selected eigenfunctions.
It is by now well understood \cite{nadler2007fundamental} that one may need to look deep into the higher order eigenfunctions with sufficiently small $\epsilon$ to recover non-trivial features of the geometry; \emph{a priori}, this would therefore involve using a large number of the smallest eigenfunctions, especially if there are degeneracies.
On the other hand, the method we have presented assigns neighbours based directly upon the intrinsic geodesic distances, which implicitly involve higher order eigenfunctions \cite{zelditch_2017} so that the dependence on $N$ and $\epsilon$ is only in the quality of the application of Egorov's theorem and hence, of geodesic approximation.
Therefore, our method potentially enables a smaller embedding dimension to identify the same geometric features of data.

Finally, to illustrate another feature of our approach to manifold learning via quantum dynamics, in Fig. \ref{fig:covid}(d) we show 3D embedding and clustering using the same methods described above (and same parameters), except using $\ell^*$, the maximum of the probability distribution dictated by a propagated coherent state, as the estimate of the position of the propagated state. While the mean position $\expect{\breve{\bx}_t}$ enables extraction of clusters with consistent average behavior, using the maximum to determine points a certain geodesic distance away results in a procedure that is sensitive to data outliers. This is illustrated by the brown and dark blue clusters in Fig. \ref{fig:covid}(d), which are disconnected or poorly connected to the remaining data points. When examining the average SAH fraction time series for these clusters in Fig. \ref{fig:covid}(e) we see that they are distinct from the other time series. Particularly remarkable is the dark blue cluster, which shows a marked change in SAH fraction around May 1st 2020. We note that the number of these outliers is small (18 (62) data points in the brown (dark blue) cluster) and identifying them is a challenging task in such a large dataset. These outlier CBGs are distributed across the whole state, but particularly concentrated in urban areas; \eg in Fig. \ref{fig:covid}(f) we show a zoom into the Atlanta metropolitan area, which contains some of the outlier CBGs.

The application presented in this section and subsequent analysis demonstrate the applicability of our manifold learning approach to real-world data sets, even ones that are moderate in size ($N\sim 5000$ points). Remarkably, we have empirically found that our approach is also useful in smaller datasets, where the theoretical properties derived in the large $N$ limit are not necessarily guaranteed. In the SI, Sec. VIII we present results from applying our method to small high-dimensional datasets, one containing COVID-19 mobility data at the global level, and another containing National Basketball Association (NBA) player performance statistics.

\section{Discussion and Related Work}
\label{sec:disc}

Manifold learning has a long history with much of the early work dedicated to approximating geodesics or the distance map via discretization of the Hamilton-Jacobi or eikonal equations in one form or another \cite{mitchell_discrete_1987, kimmel_computing_1998}.
Even Dijkstra's algorithm and its various cousins have their justifications rooted in these non-linear PDEs on configuration space \cite{peyre_geodesic_2010}.
These methods rely either on intrinsic knowledge of the manifold such as a parameterization or some form of regularized meshing such as triangulation; in the case of Dijkstra's algorithm, which remains closer to the \emph{inverse problem} roots of manifold learning, convergence is not always guaranteed \cite{peyre_geodesic_2010}.
More recently, the field has received a resurgence of interest and many new techniques have emerged since discrete graph structures were shown to approximate operators of the form $\Delta + O(\partial^1)$ on manifolds, with high probability \cite{belkin_laplacian_2003, hein2005graphs}.
These works to date, to the best of the authors' knowledge, owe in significant part to the highly local properties of differential operators and the corresponding diffusion processes they generate
\cite{belkin_laplacian_2003,hein2005graphs,chui_special_2006,berry2016local}.
In this work, we have gone \emph{against the grain} in this sense: we have relied on the highly local nature of the diffusion process to approximate its generator, which \emph{semi-classically} has the form $\Delta + O(h \partial^1)$ and proceeded to drive localized initial conditions with \emph{non-local} dynamics.
In doing so, we have developed an entirely linear framework for approximating the highly non-linear geodesic equations, eschewed the need for short-time asymptotics and any gradient-type computations  (both of which have stability issues) as needed, for example in \cite{crane_geodesics_2013} and traded the search for the correct set of eigenfunctions -- a difficult problem in the context of manifold learning applications \cite{cheng2020spectral} -- with the optimization of a continuous parameter, $h$ through the coherent state, for which we provide a region to narrow and guide the choice of optimal parameters. Additionally, while most manifold learning methods are sensitive to the sampling density, $\pd$, as discussed in \cref{sec:quant_manilearn}, our approach has reduced sensitivity because $\pd$ contributes at low order to the principal symbol of the data-driven Hamiltonian.

The methods we employ combine diffusion dynamics on manifolds \cite{hein2005graphs,coifman_geometric_2005},
which provide a necessary step to essentially \emph{bootstrap} from discrete graph structures to differential operators $\hat{\mathcal{L}}_{\epsilon,\lambda}$ that encode geometric data, with semiclassical analysis \cite{zworski_semiclassical_2012} that informs the sense in which these operators encode configuration-space dynamics, the appropriate dynamical operators $\hat{V}_t$ and initial conditions $|\psi_{\zeta_0}^h\rangle$ with which to access this and the length scales $h$ at which this correspondence is valid.
Related to the latter aspect is the large body of work from microlocal analysis \cite{hormander2015analysis,hormander2007analysis,duistermaat1972fourier,guillemin2013semi}, which uses the concept of the wavefront set to identify propagation of singularties along the classical dynamics encoded in the driving operator.
In essence, the spectral function $\hat{V}_t$ of $\hat{\mathcal{L}}_{\epsilon,\lambda}$ is approximating the Fourier Integral Operator (FIO) $\hat{U}^t := e^{i t \sqrt{\Delta_g}}$ and the application of the coherent state $|\psi_{\zeta_0}^h\rangle$ is in fact resolving the wavefront set of the Schwartz kernel of $\hat{U}^t$ to resolution governed by $h$ and for the point $\zeta_0$ in phase space $T^*\mathcal{M}$.
It is conceivable that one could devise other FIOs with the appropriate canonical relation to resolve geodesics, or other function classes to dissect the phase space in search for the wavefront set, or canonical relation of an FIO.
In general, both are hard to deploy, especially in the presence of anisotropy; such a program has been carried out in some part for the case of data on Cartesian grids in \cite{candes2003curvelets, candes2005continuous, candes2007fast}, hence this is in a rather different setting to ours.
To the best of the authors' knowledge, ours is the first work to employ such methods in the sense of the manifold learning inverse problem.

One aspect we merely touch upon here is that of embedding the manifold.
The problem of isometrically embedding a manifold in low dimensions goes back at least to Nash \cite{nash1956imbedding}, who showed with great technique how to do this constructively.
In doing so, the proof methods also show the inherent difficulty of this task: in that setting, all the details of a manifold are assumed -- indeed, Nash's theorems solve a \emph{direct} problem -- and still, this requires highly intractable (to implement) recursive methods.
Approximation methods from the \emph{inverse} problem perspective range from minimizing certain \emph{stress} functionals
\cite{MDS, ISOMAP}
to aligning approximate tangent spaces \cite{LTSA, roweis_nonlinear_2000, HLLE} to using the spectral data of an approximated diffusion operator
\cite{belkin_laplacian_2003,coifman_geometric_2005,jones_manifold_2008}
or more recently, employing geometric controllability of the wave equation \cite{bosi2017reconstruction}.
In essentially all cases, a key role is played by the distance map on the manifold: this is approximated to some order by each method, typically chaining local Euclidean distances from the given, possibly (very) high-dimensional embedding, to global approximations of geodesics.
The quality of these approximations, going from local to global essentially controls the order to which a given embedding will be isometric and close to the original dimension of the manifold.
Thus, the task of approximating geodesics well is key to solving the often-sought embedding problem.

Recently, new deterministic methods have been proposed to resolve the embedding problem with significant improvement in dimension reduction, in a more tractable fashion than Nash's algorithm \cite{fefferman_reconstruction_I,bosi2017reconstruction}.
These stray somewhat from the present context, but there is nevertheless some intersection with \cite{bosi2017reconstruction} in that geometric control of the solution to a spectral approximation to the wave equation is used to embed the manifold from data and in doing so, also to recover the distance map.
In that case, it is the unique extension property of smooth solutions that plays a key role, rather than the specific nature of propagation and phase-space localization properties that we employ.
Other than these cases, most of the effort in manifold embedding coming from dynamics on function spaces has been rooted in diffusion processes:
eigenfunction embedding methods such as \cite{coifman_geometric_2005,jones_manifold_2008} rely heavily on the asymptotically small-time limit of the heat equation to recover geodesics, as per Varadhan's formula \cite{crane_geodesics_2013}.
To compare, on the one hand, our proposed method is only little more than the very original 0-1 graph Laplacian operator as proposed in \cite{belkin_laplacian_2003}.
On the other hand, when this operator marks geodesic neighbours, globally rather than only locally, we can interpret it as an approximation to the \emph{spherical means} operator \cite{sunada1981spherical, zelditch_2017}.
In that sense, when using single neighbour (the 0-1 case), this is a sparse representation of the canonical relation in $\cos(t \sqrt{\Delta_g})$ that encodes geodesics and simultaneously it provides a nearest neighbours approach to embedding.

A common problem in the real-world application of spectral methods is the search for eigenfunctions: a typical use-case is to embed high dimensional data in low dimensions in order to separate the data into meaningful clusters.
The heuristic is the existence of a spectral gap in the diffusion matrix $\breve{T}_{\epsilon,\lambda}:=\breve{D}^{-1}_{\epsilon,\lambda}\breve{\Sigma}^{-\lambda}_{\epsilon} \breve{T}_{\epsilon}\breve{\Sigma}^{-\lambda}_{\epsilon}$ among the top eigenvalues (which are the low-lying eigenvalues of graph Laplacian), both to identify the number of clusters and the dimension of embedding in which to identify them with lower dimensional techniques \cite{von2007tutorial}.
However, such a gap may be dampened or effectively non-existent due to various features of data manifold, such as large variances in cluster sizes or a dominatingly large background cluster, which absorbs most of the diffusion process and thus localizes many of the leading eigenvectors to itself \cite{cheng2020spectral,nadler2007fundamental}.
This is of course particularly problematic for outlier detection \cite{cheng2020spectral,miller2015spectral}.
In contrast, our approach works with the high-energy eigenfunctions of the graph Laplacian by the nature of $\hat{V}_t$ that drives the dynamics: indeed, the propagation of localized states, or in the microlocal setting, singularities, by the half-wave propagator it approximates is governed by the high-energy eigenfunctions of $\Delta_g$ \cite{zelditch_2017}.
These eigenfunctions are, by their very nature, highly \emph{delocalized}
\cite{zelditch_2017} and tend to concentrate on geodesics -- which are global -- rather than concentrating on dense regions, either due to sampling concentration or geometry, as would be the case for slow mixing of diffusion processes.
Moreover, the use of coherent states $|\psi_{\zeta_0}^h\rangle$ has another effect from the clustering perspective: the parameter $h$ sets that on propagation with this state, the operator is implicitly projected onto eigenvalues to order $1/h$.
We have suggested a simple test to give a regime in which the choice of $(\epsilon,h)$ yields good embeddings for all of the examples we have considered here.
Thus, we are employing a large range of the spectrum at once and $h$ provides an \emph{implicit} parameter for spectral truncation, with the ability to \emph{a priori} ascertain optimal regions.
This may explain the applicability of our approach to the various tests provided here, both in large and small data settings: in particular, we note that in Figure \ref{fig:covid} (d-f) we witness the existence of an outlier statistic with accompanying cluster even with a very rudimentary embedding with extremely low dimension (three), which is not seen to be present with the diffusion based approaches until at least eight eigenvectors are used and clustering is done in eight dimensions.
Further, in the SI, Sec. VIII we show more embeddings from other real-world datasets that exhibit meaningful separation and clustering.

In addition to establishing a new framework for manifold learning our work motivates several future directions. Firstly, while we used graph Laplacian convergence results in order to construct the unitary quantum propagator, it would be interesting to use other constructions with convergence guarantees.
As discussed in the 2-torus example, the slow spectral convergence of graph Laplacians motivates the search for other approaches to constructing the unitary propagator that provide direct access to the matrix elements of the FIO $\hat{U}^t = e^{i t \sqrt{\Delta_g}}$.
Secondly, while in this work we used force-based graph layout as the embedding method,  it would be fruitful to evaluate other embedding techniques, \emph{e.g.} \cite{Puchkin_2020}, and construct new ones that specifically take advantage of sparse geodesic distance matrices.
Another interesting direction is to consider the integration of our approach with techniques for data-driven modeling of dynamical systems.
In essence, our approach reconstructs the dynamical flow of a system from samples of configuration space.
Considering the convergence properties in \cite{kumar_math}, the methods can readily approximate the geodesic flow of rather general phase space observables, hence enabling
the simulation of observable time traces, $f_i \circ \Phi^t$ for many $i$.
These can be input to dynamic mode decomposition (DMD)
\cite{Kutz_2016} to construct approximations of the Koopman operator $f \mapsto f \circ \Phi^t$
\cite{Rowley_2009} for the system and the corresponding analysis of \emph{coherent structures} may be useful for applications \cite{Budisic_2012}.

Finally, our work naturally suggests the investigation of quantum computing algorithms for manifold learning. Quantum computers are widely believed to provide an exponential speedup for the simulation of quantum dynamics \cite{Georgescu_2014}, which is the core ingredient in our approach. Thus, a quantum algorithm for manifold learning based on the approach presented here could provide the ultimate route to scalable manifold learning on extreme-size datasets.

\begin{acknowledgements}
This work was supported by the Laboratory Directed Research and Development program at Sandia National Laboratories, a multimission laboratory managed and operated by National Technology and Engineering Solutions of Sandia, LLC., a wholly owned subsidiary of Honeywell International, Inc., for the U.S. Department of Energy's National Nuclear Security Administration under contract {DE-NA-0003525}. MS was also supported by the U.S. Department of Energy, Office of Science, National Quantum Information Science Research Centers.	
\end{acknowledgements}

\section*{References}
\bibliography{bib}
\pagebreak

\setcounter{section}{0}
\setcounter{part}{-1}
\setcounter{prop}{0}

\renewcommand{\thesection}{\Roman{section}}
\addtolength{\cftsecnumwidth}{30pt}
\addtolength{\cftsubsecnumwidth}{40pt}

\renewcommand{\thesubsection}{\thesection.\Alph{subsection}}

\onecolumngrid
\part{}
\linespread{1.5}
\parttoc
\linespread{1}
\raggedbottom
\pagebreak

\section{Notation}
In this section, we collect and explain the notation used in this work.

\begin{table*}[hbt]
  \begin{tabular}{|l|p{13cm}|}
\hline
  $\mathcal{M}$ & Manifold that data is sampled from. \\
\hline
	$g$ & Metric of manifold $\mathcal{M}$. \\  
\hline
  $\nu$ & Intrinsic dimension of $\mathcal{M}$. \\
\hline 
	$d_g(\bx_1,\bx_2)$ & Intrinsic, geodesic distance between points $\bx_1,\bx_2 \in \mathcal{M}$. \\
\hline
	$\Delta_g$ & Laplace-Beltrami operator on $\mathcal{M}$. \\
\hline
   $n$ & Dimension of extrinsic coordinates for points sampled from $\mathcal{M}$.\\
\hline
   $V=\{v_i\}_{i=1}^N$ & Dataset, consisting of $N$ vectors in $v_i\in\mathbb{R}^n$. \\
\hline
   $\mathsf{p}$ & Sampling distribution dictating sampling from $\mathcal{M}$.\\
\hline
	$\iota(\bx)$ & Extrinsic coordinates for a point $\bx \in \mathcal{M}$. \\
\hline
	$\eta^{\bx_0}(\bx)$ & Normal coordinates centered at $\bx_0$ for a point $\bx \in \mathcal{M}$. \\
\hline
	$\mathscr{S}$ & Schwartz space of functions on $\mathbb{R}^d$. \\
\hline
	$\mathscr{S}'$ & Space of tempered distributions on $\mathbb{R}^d$. \\
\hline
	$\bx=(x^1, ... x^\nu),\bp=(p_1, ..., p_\nu)$ & Position and momentum phase space variables defined on $T^* \mathcal{M}$. \\
\hline
	$\zeta = (\bx,\bp) \in T^* \mathcal{M}$ & A point in phase space. \\
\hline
	$\Phi_q^t$ & Hamiltonian flow representing time evolution associated with symbol $q$. \\
\hline
   $a \in C^\infty(T^*\mathcal{M})$ & A symbol representing a physical observable in phase space. \\
\hline
	$h^kS^m$ & Symbol class of order $(k,m)$. \\
\hline
	$h^k\Psi^m$ & Space of pseudo-differential operators of order $(k,m)$.\\
\hline
	$\hat{A}: C^\infty(\mathcal{M}) \to C^\infty(\mathcal{M})$ & Operator representing a quantum observable. \\		
\hline
	$\Op_h(a)$ & Quantization of a symbol $a$ to an operator $\hat{A}$.\\
\hline
	$\hat{U}_t: C^\infty(\mathcal{M}) \to C^\infty(\mathcal{M})$ & Unitary propagator representing time evolution map for a quantum system. \\
\hline
	$\ket{\psi} \in C^\infty(\mathcal{M})$ & An (infinite-dimensional) quantum state. \\
\hline
	$\ket{\psi^h_{\zeta_0}} \in C^\infty(\mathcal{M})$ & A coherent state centered at phase space point $\zeta_0$. \\
\hline
	$|\psi)\in \mathbb{C}^{N}$ & An $N$-dimensional complex vector that depends on the dataset $V$. If a corresponding $\ket{\psi}$ is defined, $|\psi)$ is the restriction of $\ket{\psi}$ to $V$. Sometimes we will abuse this notation and denote both the vector and a $C^\infty$ extension by $|\psi)$. \\
\hline 
	$\breve{A}\in \mathbb{C}^{N\times N}$ & An $N\times N$ complex matrix that depends on the dataset $V$. If a corresponding $\hat{A}$ is defined, $\breve{A}$ is a finite-dimensional approximation of $\hat{A}$. \\
\hline
	inj($\bx_0$) & Injectivity radius of point $\bx_0 \in \mathcal{M}$. \\
\hline
$\tilde{\nu}$ & Dimension for local principal component analysis. \\
\hline
	$\breve{O}:\mathbb{R}^n \to \mathbb{R}^{\tilde{\nu}}$ & Local principal component analysis (LPCA) linear map \\
\hline
	$\vartheta_\ell^{v_0} \in \mathbb{R}^{\tilde{\nu}}$ & LPCA coordinates for data point $\ell$, with LPCA performed at point $v_0$.\\
\hline
	$\epsilon>0$ & Scale parameter dictating neighborhood size for sampled data points on $\mathcal{M}$. \\
\hline
	$0\leq \lambda\leq 1$ & Normalization parameter dictating convergence properties of data-driven Markov generator. \\
\hline
	$\alpha>0$ & Scale parameter dictating uncertainty of coherent states constructed from dataset $V$. \\
\hline
	$\Delta t>0$ & Time step taken by quantum propagator. \\
\hline
	$n_{\rm prop} \in \mathbb{Z}$ & Number of $\Delta t$ time steps to propagate.  \\
\hline
	$n_{\rm coll} \in \mathbb{Z}$ & Number of initial states to propagate. \\
\hline
	$\delta_{\rm PCA}>0$ & Radius of Euclidean ball that dictates the samples points over which LPCA is performed.\\
\hline
  \end{tabular}
\end{table*}

\section{Theoretical proofs}
\label{sec:proofs}
\subsection{Proof of Proposition 1}
\begin{prop}
Let $\alpha \geq 1$, $h > 0$, $\zeta_0 := (\bx_0, \bp_0) \in T^*\mathcal{M}$ and let $\psi_{\zeta_0}^h$ be a coherent state localized at $\zeta_0$. Then with $\epsilon := h^{(2 + \alpha)}$ we have,
\begin{align}
 (h^2 \hat{\mathcal{L}}_{\epsilon,\lambda}) \psi^h_{\zeta_0}(\bx_0) = |\bp_0|^2_{g(\bx_0)}\psi^h_{\zeta_0}(\bx_0) + \mathcal{O}(h).
\label{eq:L_coh_si}
\end{align}
\label{thm:thm1_si}
\end{prop}

\noindent\textit{Proof}: The expansion of $\hat{\mathcal{L}}_{\epsilon,\lambda}$ defined in the main text is pointwise a Taylor series and can be found in \cite{hein2005graphs} to be, on application to $u \in C^{\infty}$ in local coordinates,
\begin{align}
	\hat{\mathcal{L}}_{\epsilon,\lambda}[u](\bx) = \left( \Delta_g + \frac{2(1 - \lambda)}{\pd} \langle \nabla \pd, \nabla\cdot\rangle_{g_{\bx}} \right)[u](\bx)	
		+ \sum_{j=2}^{\infty} \frac{\epsilon^{j-1}}{j!} G_j(\bx,D)[u](\bx),
\end{align}
with each $G_j$ a partial differential operator (PDO) of order $2j$ with smooth coefficients in the variable $\bx$.
The issue with applying this to a coherent state is that we must let $h \to 0$ while also having $\epsilon \to 0$, but the applications of the PDOs to $\psi^h_{\zeta_0}$ bring about terms of negative orders in $h$, which may lead to the effect of higher order PDOs becoming asymptotically dominant.
To curtail this, we must scale $h$ appropriately according to $\epsilon$ and also rescale $\hat{\mathcal{L}}_{\epsilon,\lambda}$ (in a manner common in semiclassical analysis).
For generality, we let $\epsilon := \epsilon(h)$ be a function of $h > 0$.

On letting $s(z) := \exp_{\bx}(z)$ provide normal coordinates, we find that on application to a coherent state,
\[
	\hat{\mathcal{L}}_{\epsilon\lambda}[\psi^h_{\zeta_0}](\bx) = h^{-2} \psi^h_{\zeta_0}(\bx) \langle \nabla_z, \nabla_z \rangle \big|_{z = 0} \phi \circ s
		+ \mathcal{O}(h^{-1}) \psi^h_{\zeta_0}(\bx) + \sum_{j=2}^{\infty} \frac{\epsilon^{j-1}}{j!} \mathcal{O}(h^{-2j}),
\]
wherein $\phi$ denotes the complex phase of $\psi^h_{\zeta_0} = e^{\frac{i}{h} \phi(\bx ; \zeta_0)}$ and $\mathcal{O}(h^{-k})$ denotes a polynomial in $1/h$ of degree $k$ with coefficients depending only on derivatives of $\phi$ and the smooth functions appearing in the coefficients of the operators $G_j$.
We wish to recover the dominant term of the semi-classical application
\begin{align*}
	h^2 \Delta_g[\psi^h_{\zeta_0}](\bx)
		&= \psi^h_{\zeta_0}(\bx) \langle \nabla_z, \nabla_z \rangle \big|_{z = 0} \phi \circ s + \mathcal{O}(h)	\\
		&= |d \phi|_{\bx}|^2 \psi^h_{\zeta_0}(\bx) + \mathcal{O}(h),
\end{align*}
which follows from $ds|_{z = 0} = I$, as the dominant term from the $\hat{\mathcal{L}}_{\epsilon,\lambda}$ approximation.
To do so, we must balance the asymptotic scaling in the terms of higher order: namely, on setting $\epsilon \in \mathcal{O}(h^{2 + \beta})$ with $\beta \geq 0$, then we have $\epsilon^{j-1} \mathcal{O}(h^{-2(j-1)}) \in \mathcal{O}(h^{\beta})$ for all $j \geq 2$, so that
\begin{align*}
	 (h^2 \hat{\mathcal{L}}_{\epsilon,\lambda}) [\psi^h_{\zeta_0}](\bx_0)
		&=  (h^2 \Delta_g) \psi^h_{\zeta_0}(\bx_0) + \mathcal{O}(h) + \mathcal{O}(h^{\beta})	\\
		&= |\bp_0|^2_{g(\bx_0)}\psi^h_{\zeta_0}(\bx_0) +\mathcal{O}(h) + \mathcal{O}(h^{\beta}),
\end{align*}
since $\phi(\bx_0 ; \zeta_0) = 0$ and $d\phi|_{\bx = \bx_0} = -\bp_0$ by definition of the admissibility condition on $\phi$.
If now we use $\epsilon = h^{2 + \alpha}$ with $\alpha \geq 1$, then the remainder term is just $\mathcal{O}(h)$ as in the statement of the Proposition.

\subsection{Proof of Proposition 2}

We restate the data-driven constructions of coherent states and positions expectations detailed in the main text here for ease of reference. 

The coherent state definition using extrinsic coordinates provided by $v_i \in \mathbb{R}^n$, with expected position corresponding to the point $v_0$ is given by a vector with elements
\begin{align}
	\left[\vert \tilde{\psi}^h_{\zeta_0} )\right]_{\ell} = e^{\frac{i}{h}(v_\ell-v_0)^{\sf T}p_0}e^{-\frac{\lVert v_\ell - v_0 \rVert ^2}{2h}}, \quad 1\leq \ell \leq N
	\label{eq:coh_extrinsic_si}
\end{align}
$p_0 = (v_n-v_0)/||v_n-v_0|| \in \mathbb{R}^n$ is a unit norm vector between the initial point and a nearby neighboring point according to the original Euclidean distances in $\mathbb{R}^n$. 

Alternatively, the coherent state can be formulated using local coordinates constructed using LPCA.
We perform LPCA on samples points within a Euclidean distance $\delta_{\rm PCA}$ from the  sample point $v_0$ corresponding to the initial point $\bx_0$; \ie on sample points in the set $\mathcal{S}_0 = \{\ell ~\vert~ \vert\vert v_0-v_\ell\vert\vert^2 \leq \delta_{\rm PCA} \}$. Let $\vartheta^{v_0}_\ell \in \mathbb{R}^{\tilde{\nu}}$ be the LPCA coordinates for all $v_\ell$ with $\ell \in \mathcal{S}_0$. Once the LPCA approximation to the tangent space is formed, the initial momentum can be chosen to be a vector in this space. Then the coherent state is defined as the vector:
\begin{align}
	\left[\vert \tilde{\psi}^h_{\zeta_0} )\right]_{\ell} =
	\begin{cases}
		e^{\frac{i}{h}(\vartheta^{v_0}_\ell-\vartheta^{v_0}_0)^{\sf T}p_0}e^{-\frac{\lVert \vartheta^{v_0}_\ell - \vartheta^{v_0}_0 \rVert ^2}{2h}},&  \quad \text{if} \quad \ell \in \mathcal{S}_0 \\
		0,& \quad \text{otherwise},
	\end{cases}
	\label{eq:coherent_state_data_si}
\end{align}
with $p_0$ a unit vector in $\mathbb{R}^{\tilde{\nu}}$.

To calculate the expected position using extrinsic coordinates, we start with $q_\ell := \vert[\vert \psi_{t,\zeta_0})]_\ell\vert^2, ~\forall \ell$, which is the probability distribution over data points determined by the propagated state. Then, the expected position calculated using extrinsic coordinate is the point closest (in $L^2$ distance in $\mathbb{R}^n$) to $\bar{v}=\sum_{\ell} q_\ell v_l$. We denote this data point $\expect{\breve{\bx}_t}$.

To calculate this expectation in LPCA coordinates, we first perform LPCA around datapoint $v_{\ell^*}$, where $\ell^* = \textrm{argmax}_\ell q_\ell$, using all data points in $\mathcal{S}_{\ell^*}= \{\ell~\vert~ \vert\vert v_{\ell^*}-v_\ell\vert\vert^2 \leq \delta_{\rm PCA} \}$. The mean position is then defined as the data point that is closest (according to Euclidean distance in the LPCA space) to the empirical LPCA mean position $\bar{\vartheta} = \sum_{\ell\in\mathcal{S}_{\ell^*}} q_\ell \vartheta^{v_{\ell^*}}_\ell$, where $\vartheta^{v_{\ell^*}}_\ell$ are LPCA coordinates for data point $\ell$. We notate this approximated expected position as $\expect{\breve{\bx}_t}^{\rm LPCA}$.

Finally, for use in the following, we define the Riemannian normal coordinates centered at a point $\bx_0$ as $\ncoord{\bx_0}(\bx) = (\ncoord{\bx_0}_1(\bx), ..., \ncoord{\bx_0}_\nu(\bx))$.

Then, restating Proposition 2,
\begin{prop}
	\label{prop:prop2_si}
	Let $\beta := (2 + \alpha)(\nu + 1) + 1$ and $\beta_0 := (2 + \alpha)(\nicefrac{5 \nu}{2} + 4) + 2(\nu + 2)$.
	Then, $\exists ~ h_0 > 0$ such that
	\begin{enumerate}
	\item	if the initial state is given by \cref{eq:coh_extrinsic_si} and $N^{-\nicefrac{1}{\gamma_0}} \lesssim h < h_0$, then with high probability, $d_g(\breve{\bx}_t, \bx_t) < h$ and
	\item	if the initial state is given by \cref{eq:coherent_state_data_si} and $N^{-\nicefrac{1}{\gamma_1}} \lesssim h < h_0$, then with high probability, $d_g(\breve{\bx}_t, \bx_t) < h$,
	\end{enumerate}
	wherein $\gamma_0 := \max\{ \nu(\nu/4 + 1 + \beta), \beta_0 \}$ and $\gamma_1 := \max\{ (\nu + 2)(\nu/4 + 1 + \beta)/2, \beta_0 \}$ whenever $\nu \geq 2$ and if $\nu = 1$ is possible, these maxima must be taken with $3\nu/2 + 2\beta$.
\end{prop}

This proposition is a statement about the accuracy of the calculated expected positions with respect to the position determined by the classical geodesic flow. There are in total four combinations of state preparation and expectation calculation to consider:
\begin{enumerate}
\item	The state $\sket{\psi^h_{\zeta_0}}$ is prepared using extrinsic coordinates as per \cref{eq:coh_extrinsic_si} and the mean is also computed in extrinsic coordinates, giving $\bar{v}$.	\label{app:thm3-ext-ext-mean}
\item	The state is still \cref{eq:coh_extrinsic_si}, while the mean is computed in LPCA coordinates, giving $\bar{\vartheta}$.	\label{app:thm3-ext-lpca-mean}
\item	The state is $\sket{\psi^h_{\zeta_0}}$ prepared using LPCA coordinates as per \cref{eq:coherent_state_data_si} and the mean is also computed in LPCA coordinates, $\bar{\vartheta}$. 	\label{app:thm3-lpca-lpca-mean}
\item	The state is still \cref{eq:coherent_state_data_si}, while the mean is now computed in extrinsic coordinates, giving $\bar{v}$.	\label{app:thm3-lpca-ext-mean}
\end{enumerate}
In all cases, we are using states whose phases are some perturbation away from being admissible.
When the initial states are true coherent states, it is shown in \cite[$\S 6$]{kumar_math} that the corresponding version of this proposition holds in the case of extrinsic means with a better convergence rate and similarly, a better convergence rate than stated here follows for the LPCA case upon ensuring that the projected coordinates are diffeomorphic in a neighborhood of the maximum of the propagated coherent state.
In the proof of this proposition below, we extend these results to the case of initial states that are perturbations of sufficiently low order, of coherent states.
Thus, the main argument rests on criteria for our constructed initial states $\sket{\psi^h_{\zeta_0}}$ to be such perturbations, which is the content of the following two preliminary lemmas.

We begin with a criterion for the phases of the initial states, which is where the approximations ultimately land and we record it as,

\begin{lemma} \label{lem:cs-approx}
Let $\varphi \in C^{\infty}(\mathcal{M} \times T^*\mathcal{M})$ be an admissible phase and $\zeta_0 \in T^*\mathcal{M}$.
Then, for $\ket{\psi^h_{\zeta_0}}$ a (normalized) coherent state localized at $\zeta_0$ with phase $\varphi$ and given $\phi \in C^{\infty}(\mathcal{M})$, if $\varepsilon > 0$ and there is a constant $C\geq 0$ such that $||\phi(\cdot) - \varphi(\cdot ; \zeta_0)||_{\infty} \leq C h^{\nicefrac{\nu}{4} + 1} \varepsilon$, then there is $\varepsilon_0 > 0$ so that for all $\varepsilon \in (0, \varepsilon_0]$, $\breve{\psi}^h_{\zeta_0}(\cdot) := e^{\frac{i}{h} \phi(\cdot)}/|| e^{\frac{i}{h} \phi(\cdot)} ||_{L^2(\mathcal{M})} = \psi^h_{\zeta_0}(\cdot) + \mathcal{O}(\varepsilon)$.
\end{lemma}

\begin{proof}
We denote by $\varphi$ the specification $\varphi(\cdot ; \zeta_0)$ and $\upsilon := \phi - \varphi$.
Then, a Taylor expansion of $e^{\frac{i}{h} z}$ yields that there is a constant $C' > 0$ such that $||e^{\frac{i}h \phi} - e^{\frac{i}{h} \varphi}||_{\infty} \lesssim ||e^{\frac{i}{h} \upsilon} - 1||_{\infty} \leq C' h^{\frac{\nu}{4}} \varepsilon$ and hence, $||e^{\frac{i}{h} \phi} - e^{\frac{i}{h} \varphi}||_{L^2} \leq ||e^{\frac{i}{h} \upsilon} - 1||_{\infty} ||e^{\frac{i}{h} \varphi}||_{L^2} \leq C'^2 h^{\frac{\nu}{2}} \varepsilon$, so $| \, ||e^{\frac{i}{h} \phi}||_{L^2} - ||e^{\frac{i}{h} \varphi}||_{L^2} \, | \leq C' h^{\frac{\nu}{2}} \varepsilon$.
Thus,
\begin{align*}
	| \breve{\psi}^h_{\zeta_0} - \psi^h_{\zeta_0} | 
		& = \left| \frac{e^{\frac{i}{h} \phi} ||e^{\frac{i}{h} \varphi}||_{L^2} -  e^{\frac{i}{h} \varphi} ||e^{\frac{i}{h} \phi}||_{L^2}}{||e^{\frac{i}{h} \varphi}||_{L^2} ||e^{\frac{i}{h} \phi}||_{L^2}} \right|	\\
		&\leq \frac{|e^{\frac{i}{h} \phi} - e^{\frac{i}{h} \varphi}| ||e^{\frac{i}{h} \varphi}||_{L^2} + |e^{\frac{i}{h} \varphi}(||e^{\frac{i}{h} \varphi}||_{L^2} - ||e^{\frac{i}{h} \phi}||_{L^2})|}{||e^{\frac{i}{h} \varphi}||_{L^2} ||e^{\frac{i}{h} \phi}||_{L^2}}	\\
		& \leq C' \varepsilon \frac{||e^{\frac{i}{h} \varphi}||_{L^2} h^{\frac{\nu}{4}} + |e^{\frac{i}{h} \varphi}| h^{\frac{\nu}{2}}}{||e^{\frac{i}{h} \varphi}||_{L^2} (||e^{\frac{i}{h} \varphi}||_{L^2} - C' h^{\frac{\nu}{2}} \varepsilon)}\\
		& \lesssim \varepsilon
\end{align*}
whenever $\varepsilon$ is sufficiently small.
\end{proof}

We wish to see that the states prepared in practice, as in Sec. 2.B of the main text, are themselves perturbations of bonafide coherent states.
That this is achievable with high probability is given in,

\begin{lemma} \label{lem:model-cs}
Let $\sket{\psi^h_{\zeta_{\operatorname{ext}}}}$ be prepared as in \cref{eq:coh_extrinsic_si} with $\zeta_{\operatorname{ext}} := (v_0, p_0) \in \mathbb{R}^{2n}$ and $\sket{\psi^h_{\zeta_{\operatorname{LPCA}}}}$ be prepared as in \cref{eq:coherent_state_data_si} with $\zeta_{\operatorname{LPCA}} := (\vartheta_0^{v_0}, p_0) \in \mathbb{R}^{2 \tilde{\nu}}$ and $\tilde{\zeta}_{\operatorname{LPCA}} := (v_0, \tilde{p}_0) \in \mathbb{R}^{2n}$ its realization in the extrinsic space with $\tilde{p}_0$ emanating from $v_0$ and residing in the LPCA subspace approximation of $T_{v_0} \iota(\mathcal{M})$.
Then, setting $\bx_0 := \iota^{-1}(v_0)$, we have for all $\beta \geq 0$ that
\begin{enumerate}
\item	there is an $L^2$ normalized coherent state $\cstate{\zeta_0}$ localized at $\zeta_0 := (\bx_0, \ncoord{\bx_0} \iota^{-1}(v_n)) \in T^*\mathcal{M}$ and a constant $h_0 > 0$ such that for all $h \in (0, h_0]$,
	\[
		\Pr[||d\iota|_{\bx_0} \zeta_0 - \zeta_{\ext}||_{\mathbb{R}^n} > h^{\nicefrac{\nu}{4} + 1 + \beta}
			\wedge ||[\dcstate{\zeta_{\ext}}] - [\cstate{\zeta_0}]||_{\infty} > h^{\beta}]
		  \leq e^{-\Omega(N h^{\nu (\nicefrac{\nu}{4} + 1 + \beta)})} + e^{-\Omega(N h^{\nicefrac{3 \nu}{2} + 2\beta})}
	\]
	and
\item	there is an $L^2$ normalized coherent state $\cstate{\zeta_0}$ localized at $\zeta_0 = (\bx_0, \zeta_{\bp_0}) \in T^*\mathcal{M}$ and a constant $h_0 >0$ such that if $\delta > 0$ and $\delta N^{-\frac{2}{(\nu + 2)(\nicefrac{\nu}{4} + 1 + \beta)}} \leq h \leq h_0$, then
	\[
		\Pr[||d\iota|_{\bx_0} \zeta_0 - \tilde{\zeta}_{\LPCA}||_{\mathbb{R}^n} > h^{\nicefrac{\nu}{4} + 1 + \beta}
			\wedge ||[\dcstate{\zeta_{\LPCA}}] - [\cstate{\zeta_0}]||_{\infty} > h^{\beta}]
		\leq e^{-\Omega(\delta^2)} + e^{-\Omega(N h^{\nicefrac{3 \nu}{2} + 2\beta})} .
	\]
\end{enumerate}
\end{lemma}

\begin{proof}
The discussion in \cite[$\S 3.1$]{kumar_math} 
shows that $\dcstate[\tilde]{\zeta_{\ext}}$ is an $\mathcal{O}(h^{\nicefrac{\nu}{4} + \beta})$ perturbation of $\cstate[\tilde]{\zeta_0}$ whenever $v_n$ is chosen such that $||v_0-v_n||_{\mathbb{R}^n} \lesssim h^{\nicefrac{\nu}{4} + 1 + \beta} < \kappa$.
Then, $d_g(\iota^{-1}(v_0), \iota^{-1}(v_n)) \lesssim h^{\nicefrac{\nu}{4} + 1 + \beta}$ and by a Chernoff bound, such a point belongs to our sample set with probability at least $1 - e^{-\Omega(N h^{\nu (\nicefrac{\nu}{4} + 1 + \beta)})}$.
Further, by 
\cite[Lemma 19]{kumar_math}, 
$| ||\dcstate[\tilde]{\zeta_{\ext}}||_N^2 - ||\dcstate[\tilde]{\zeta_{\ext}}||_{L^2}^2 | \lesssim h^{\nicefrac{\nu}{2} + \beta}$ with probability at least $1 - e^{-\Omega(N h^{\nicefrac{3 \nu}{2} + 2\beta})}$.
Combining these events with \cref{lem:cs-approx} gives,
\begin{align*}
	| \dcstate{\zeta_{\ext}} - \cstate{\zeta_0} |
		&\leq \left| \dcstate{\zeta_{\ext}} - \frac{\dcstate[\tilde]{\zeta_{\ext}}}{||\dcstate[\tilde]{\zeta_{\ext}}||_{L^2}} \right|
		+ \left| \frac{\dcstate[\tilde]{\zeta_{\ext}}}{||\dcstate[\tilde]{\zeta_{\ext}}||_{L^2}} - \cstate{\zeta_0} \right|	
		\lesssim h^{\beta}
\end{align*}
and with a union bound we arrive at the probabilistic bound in the statement of the first part of the Lemma.

Now, for the second part, let $\iota : \mathcal{M} \hookrightarrow \mathbb{R}^n$ denote the isometric embedding such that for each $1 \leq \ell \leq N$, there is $\bx_{\ell} \in \mathcal{M}$ so that $v_{\ell} = \iota(\bx_{\ell}) \in V$.
Further, for each $\bx \in \mathcal{M}$, let $\Pi_{\bx} : \mathbb{R}^n \to T_{\iota(\bx)} \iota(\mathcal{M}) \subset \mathbb{R}^n$ denote the orthogonal projection onto the subspace of $\mathbb{R}^n$ tangent to the submanifold at $\iota(\mathcal{M})$ and set  $\varphi(\bx ; \zeta_{\bx}, \zeta_{\bp}) := \langle d\iota|_{\zeta_{\bx}} \zeta_{\bp}, \Pi_{\zeta_{\bx}} [\iota(\zeta_{\bx}) - \iota(\bx)] \rangle + \frac{i}{2} \frac{\langle d\iota|_{\zeta_{\bx}} \zeta_{\bp} \rangle}{\langle p_0 \rangle} ||\Pi_{\zeta_{\bx}} [\iota(\zeta_{\bx}) - \iota(\bx)]||_{\mathbb{R}^n}^2$ with $\zeta := (\zeta_{\bx}, \zeta_{\bp}) \in T^*\mathcal{M}$.
Then, $\varphi \in C^{\infty}(\mathcal{M} \times T^*\mathcal{M})$ and satisfies,
\[
	d\varphi|_{\bx = \zeta_{\bx}} = -\zeta_{\bp} , \quad\quad d^2 \Im(\varphi)|_{\bx = \zeta_{\bx}} = \frac{\langle d\iota|_{\zeta_{\bx}} \zeta_{\bp} \rangle}{\langle p_0 \rangle} d\iota|_{\zeta_{\bx}}^{\sf T} d\iota|_{\zeta_{\bx}} ,
\]
hence it is an admissible phase.
The corresponding coherent state $\ket{\psi^h_{\zeta}}$ with phase $\varphi$ will serve as the \emph{model} for the LPCA prepared state \cref{eq:coherent_state_data_si}.

The LPCA procedure \cite[$\S 2$]{singer2012vector} yields a linear map, $\breve{O} : \mathbb{R}^n \to \mathbb{R}^{\tilde{\nu}}$ for $1 \leq \tilde{\nu} \leq n$.
Its matrix representation is given by $\breve{O} := [u_1, \ldots, u_{\tilde{\nu}}]^{\sf T}$ such that for each $1 \leq j \leq \tilde{\nu}$, the column vector $u_j \in \mathbb{R}^n$ is the left singular vector corresponding to $j$-th largest singular value of a local neighbourhood matrix at $v_0$.
Suppose, without loss of generality, that the embedding $\iota(\mathcal{M})$ is appropriately translated and rotated such that $v_0 = 0 \in \mathbb{R}^n$ and the subspace $T_{v_0} \iota(\mathcal{M}) \subset \mathbb{R}^n$ is aligned so that the first $\nu$ vectors of the standard basis $\{ e_1, \ldots, e_n \}$ (meaning $e_j$ has $j$-th entry equal to $1$ and is otherwise zero) of $\mathbb{R}^n$ form an orthonormal basis for $T_{v_0} \iota(\mathcal{M})$ and there are $\zeta_{\bp_1}, \ldots, \zeta_{\bp_{\nu}} \in T^*_{\zeta_{\bx_0}}(\mathcal{M})$ such that for each $1 \leq j \leq \nu$, $d\iota|_{\zeta_{\bx_0}} \zeta_{\bp_j} = e_j$.
Now assume the event, call it $\mathcal{D}$, that $\tilde{\nu} = \nu$ and recall that the phase constructed in \cref{eq:coherent_state_data_si} is given by
\[
	\phi(\bx) := \langle p_0, \breve{O}[\iota(\bx_0) - \iota(\bx)] \rangle_{\mathbb{R}^{\nu}} + \frac{i}{2} ||\breve{O}[\iota(\bx_0) - \iota(\bx)]||_{\mathbb{R}^{\nu}}^2 .
\]
Since the singular vectors $u_1, \ldots, u_{\nu}$ are orthonormal, we have $\breve{O} \breve{O}^{\sf T} = I_{\nu}$, so
\[
	\phi(\bx) = \langle \tilde{p}_0, \hat{\Pi}[\iota(\bx_0) - \iota(\bx)] \rangle_{\mathbb{R}^n} + \frac{i}{2} ||\hat{\Pi}[\iota(\bx_0) - \iota(\bx)]||_{\mathbb{R}^n}^2
\]
with $\tilde{p}_0 := \breve{O}^{\sf T} p_0$ and $\hat{\Pi} := \breve{O}^{\sf T} \breve{O}$.

Further suppose we are in the event, call it $\mathcal{A}_{\varepsilon}$ for $\varepsilon > 0$, that there are $w_1, \ldots, w_{\nu} \in \mathbb{R}^n$, which form an orthonormal basis for $T_{\iota(v_0)} \iota(\mathcal{M})$ such that for each $1 \leq j \leq \nu$, $||u_j - w_j||_{\infty} \leq \varepsilon$.
Then, $||\hat{\Pi} - \Pi_{\zeta_{\bx_0}}|| \lesssim \varepsilon$ and setting $\zeta_{\bp_0} := d\iota|_{\zeta_{\bx_0}}^{\sf T} \Pi_{\zeta_{\bx_0}} \tilde{p}_0$, we have
\begin{align*}
	|&\phi(\bx) - \varphi(\bx ; \zeta_{\bx_0}, \zeta_{\bp_0})|	\\
		&\quad \leq |\langle \tilde{p}_0, (\hat{\Pi} - \Pi_{\zeta_{\bx_0}}) \cdot \iota(\bx) \rangle
			+ \langle (I_n - \Pi_{\zeta_{\bx_0}}) \tilde{p}_0, \Pi_{\zeta_{\bx_0}} \cdot \iota(\bx) \rangle|
			 + \frac{1}{2} \left| ||\Pi_{\zeta_{\bx_0}} \iota(\bx)||^2 - ||\hat{\Pi} \iota(\bx)||^2 \right|	\\
		&\quad \leq ||p_0|| \cdot ||(\hat{\Pi} - \Pi_{\zeta_{\bx_0}}) \cdot \iota(\bx)|| + ||(\hat{\Pi} - \Pi_{\zeta_{\bx_0}}) \tilde{p}_0|| \cdot ||\Pi_{\zeta_{\bx_0}} \cdot \iota(\bx_0)||	
		 + \frac{1}{2} \left| ||\Pi_{\zeta_{\bx_0}} \iota(\bx)|| + ||\hat{\Pi} \iota(\bx)|| \right| \cdot ||(\hat{\Pi} - \Pi_{\zeta_{\bx_0}}) \cdot \iota(\bx)||	\\
		&\quad \lesssim \varepsilon ,
\end{align*}
wherein we have used that $\hat{\Pi} \tilde{p}_0 = \tilde{p}_0$.

We now wish to apply \cref{lem:cs-approx}, so we require that $\varepsilon \lesssim h^{\nicefrac{\nu}{4} + 1 + \beta}$.
Further, if $\delta N^{-\frac{2}{(\nu + 2)(\nicefrac{\nu}{4} + 1 + \beta)}} \leq h$, then from the proof of \cite[Theorem B.1]{singer2012vector}, we find that the corresponding event $\mathcal{A}_{\varepsilon}$ happens with probability at least $1 - e^{-\Omega(\delta^2)}$.
Moreover, within the same event, we find that the singular values resulting from LPCA and corresponding to $\breve{O}$ have the following \emph{gap property}: if $\sigma_1 \geq \cdots \geq \sigma_n \geq 0$ are the singular values corresponding to $u_j, \ldots, u_n$, then for $1 \leq j \leq \nu$, $\sigma_j \sim N \varepsilon^{\nicefrac{\nu}{2} + 1}$, while for each $\nu + 1 \leq m \leq n$, $\sigma_m \lesssim N \varepsilon^{\nicefrac{\nu}{2} + 2}$.
Therefore, for all $(j,m) \in \{ 1, \ldots, \nu \} \times \{ \nu + 1, \ldots, n \}$, we have $\sigma_j/\sigma_m \gtrsim h^{-(2 + \nicefrac{\nu}{4})}$, so within the same probabilistic event we can determine the dimension $\nu$, \emph{viz}., the event $\mathcal{D}$ occurs simultaneously with $\mathcal{A}_{\varepsilon}$.
By the proof of \cref{lem:cs-approx}, we so far have that $||e^{\frac{i}{h} \phi} - e^{\frac{i}{h} \varphi}||_{L^{\infty}} \lesssim h^{\nicefrac{\nu}{4} + \beta}$ and $||e^{\frac{i}{h} \phi}/||e^{\frac{i}{h} \phi}||_{L^2} - \cstate{\zeta_0}||_{L^{\infty}} \lesssim h^{\beta}$.
An application of 
\cite[Lemma 19]{kumar_math}
gives that $| ||\dcstate[\tilde]{\zeta_{\LPCA}}||_N^2 - ||e^{\frac{i}{h} \phi}||_{L^2}^2 | \lesssim h^{\nicefrac{\nu}{2} + \beta}$ with probability at least $1 - e^{-\Omega(N h^{\nicefrac{3 \nu}{2} + 2\beta})}$.
Thus, proceeding as in the previous part we find that assuming also this event gives $||\dcstate{\zeta_{\LPCA}} - \cstate{\zeta_0}||_{L^{\infty}} \lesssim h^{\beta}$, whence upon taking a union bound over all considered events, we arrive at the probabilistic bound as stated in the second part of the present Lemma.
\end{proof}

Now that we may associate to each prepared state a proper coherent state, we wish to run the convergence results of 
\cite[$\S 6$]{kumar_math} 
on the corresponding coherent state and pull the results back to the originally prepared state, up to perturbative terms.
This entails knowing that the propagated \emph{density} $|\dcstate{t, \tilde{\zeta}_0}|^2$, for $\tilde{\zeta}_0 \in \{ \zeta_{\ext}, \zeta_{\LPCA} \}$, is also a perturbation of $|\hat{V}_t \cstate{\zeta_0}|^2$ with $\zeta_0$ related to $\tilde{\zeta}_0$ by the previous Lemma.
As an intermediary step, we will bound $|| |\breve{V}_t \cstate{\zeta_0}|^2 - |\dcstate{t,\tilde{\zeta}_0}|^2 ||_{\infty}$ with the help of a \emph{generic} $L^{\infty}$ bound on the propagation of $C^{\infty}$ functions as discussed in 
\cite[$S 6.2$]{kumar_math}. 
Such an $L^\infty$ bound makes sense because $\breve{V}_t$ is a finite-dimensional operator that extends to the continuum by Nyst\"om extension (discussed in detail in Sec. 6 of \cite{kumar_math}): 
\begin{align}
\breve{V}_t [u](\bx) = f(0)u(\bx) + \breve{T}_{\epsilon,\lambda} [Df(\breve{T}_{\epsilon,\lambda})[u]](\bx),	
\end{align}
where $f(z)=e^{it\sqrt{\frac{2c_0}{c_2}\frac{(1-z)}{\epsilon}}}$, $Df(z) = \nicefrac{(f(z) - f(0))}{z}$, and $\breve{T}_{\epsilon,\lambda}$ is the diffusion operator defined in the main text.
If $u \in C^{\infty}$, the generic $L^{\infty}$ bound mentioned above is
\begin{equation}	\label{eq:discrete-prop-sup-bound}
	|\breve{V}_t[u](\bx)|
		\quad \lesssim ||u||_{\infty} +  \frac{||u||_{N,2} t \epsilon^{-\frac{\nu + 1}{2}} ||k||_{\infty} ||p_{N,\epsilon,\lambda}||_{N,\infty}^{\nicefrac{1}{2}}}{(\inf p_{N,\epsilon}^{\lambda}) (\min_{j \in [N]} p_{N,\epsilon}^{\lambda}) \min_{j \in [N]} p_{N,\epsilon,\lambda}^{\nicefrac{3}{2}}} ,
\end{equation}
wherein $||\cdot||_{N,2}^2 := \frac{1}{N} \sbraket{\cdot}{\cdot}$, $||\cdot||_{N,\infty}$ denotes the maximum of a function over the sample set and $p_{N,\epsilon}, p_{N,\epsilon,\lambda} \in C^{\infty}(\mathcal{M})$ are given by
\begin{gather*}
	p_{N,\epsilon}(\bx) := \frac{1}{N} \sum_{j=1}^N k_{\epsilon}(\bx,\bx_j),
	\quad p_{N,\epsilon,\lambda}(\bx) := \frac{1}{N} \sum_{j=1}^N k_{N,\epsilon,\lambda}(\bx, \bx_j),	\\
	k_{N,\epsilon,\lambda}(\bx,\by) := \frac{k_{\epsilon}(\bx,\by)}{(p_{N,\epsilon}(x) p_{N,\epsilon}(y))^{\lambda}} ,
	\quad k_{\epsilon}(\bx,\by) := \epsilon^{-\nicefrac{\nu}{2}} k(||\iota(\bx) - \iota(\by)||^2/\epsilon);
\end{gather*}
these are related to the matrix definitions from Sec. 1.A of the main text via
\begin{gather*}
	[\breve{T}_{\epsilon}]_{i,j} = N k_{\epsilon}(\iota^{-1} v_i, \iota^{-1} v_j),	\\
	[\breve{\Sigma}_{\epsilon}]_{j,j} = N p_{N,\epsilon} \circ \iota^{-1}(v_j) .
\end{gather*}
Therefore, once we have a sampling based bound on the factors involving $p_{N,\epsilon}$ and $p_{N,\epsilon,\lambda}$ and consistency of the $||\cdot||_{N,2}$ norms for $\cstate{\zeta_0}$ and $\dcstate{\tilde{\zeta}_0}$, we can establish the deviation of their corresponding propagated densities.
This leads to the recovery of geodesic propagation on $\mathcal{M}$ up to errors decaying in $h$ and $N$ simultaneously, with high probability, which we now establish in,

\begin{proof}[Proof of Proposition 2]
We proceed with the cases \ref{app:thm3-ext-ext-mean} - \ref{app:thm3-lpca-ext-mean}.
Recall that $\expect{\breve{\bx}_t}$ and $\expect{\breve{\bx}_t}^{\LPCA}$ denote the mean approximated geodesic point as computed in extrinsic or LPCA coordinates, respectively and that this is independent of how the corresponding initial state is prepared.
We will denote by $\expect{\bx_t}$ either one of these points and specify a particular one with its specific notation when needed.
Regarding the initial state, let $\tilde{\zeta}_0 \in \{ \zeta_{\ext}, \zeta_{\LPCA} \}$ and let $\cstate{\zeta_0}$ be the coherent state corresponding to $\dcstate{\tilde{\zeta}_0}$ as constructed in \cref{lem:model-cs}.
We will work with the approximations with respect to $\cstate[\breve]{t,\zeta_0} := \breve{V}_t \cstate{\zeta_0}$, so we denote
\begin{gather*}
	\breve{q}_{\ell} := |[\cstate[\breve]{t, \zeta_0}]_{\ell}|^2 ,	\\
	\breve{v}_{\ell} := \sum_{\ell \in S_{\ell^*}} v_{\ell} \breve{q}_{\ell} ,
	\quad	\breve{\vartheta}_{\ell} :=  \sum_{\ell \in S_{\ell^*}} \vartheta_{\ell}^{v_{\ell^*}} \breve{q}_{\ell}
\end{gather*}
and we denote $\expect{\bx_t}_{\zeta_0} \in \{ \expect{\breve{\bx}_t}_{\zeta_0}, \expect{\bx_t}^{\LPCA}_{\zeta_0} \}$ with $\expect{\breve{\bx}_t}_{\zeta_0}$ the closest point, with respect to $||\cdot||_{\mathbb{R}^n}$, to $\breve{v}_{\ell}$ in $\iota(\mathcal{M})$ and likewise, by $\expect{\bx_t}^{\LPCA}_{\zeta_0}$ the closest point to $\breve{\vartheta}_{\ell} \in \mathbb{R}^{\tilde{\nu}}$ with repsect to the Euclidean norm in $\mathbb{R}^{\tilde{\nu}}$.
Further, we will use the notation from the proof of \cref{lem:model-cs} to let $\breve{O} : \mathbb{R}^{n} \to \mathbb{R}^{\tilde{\nu}}$ denote the LPCA map.
We will also need the continuum counterparts to the discretely constructed \emph{degree functions} $p_{N,\epsilon,\lambda'}$ for $\lambda' \in \{ 0, \lambda \}$, namely: $p_{\epsilon}(\bx) := \int_{\mathcal{M}} k_{\epsilon}(\bx, \by) \pd(\by) d\by$, $k_{\epsilon,\lambda}(\bx,\by) := k_{\epsilon}(\bx,\by)/(p_{\epsilon}(\bx) p_{\epsilon}(\by))^{\lambda}$ and $p_{\epsilon,\lambda}(\bx) := \int_{\mathcal{M}} k_{\epsilon,\lambda}(\bx,\by) \pd(\by) d\by$.
 
Now set $\beta, \delta_0 \geq 0$, $\varepsilon > 0$ and assume the occurrence of the following events:
\begin{alignat*}{4}
&	\mathcal{A}_0(h^{\beta} ; \tilde{\zeta}_0) && : &\quad& ||\dcstate{\tilde{\zeta}_0} - \cstate{\zeta_0}||_{L^{\infty}} \lesssim h^{\beta} ,	\\
&	\mathcal{E}_{\lambda'}(\delta_0) && : &\quad& ||p_{N,\epsilon,\lambda'} - p_{\epsilon,\lambda'}||_{\infty} \lesssim \delta_0	\\
&	\mathcal{G}_{\ext} && : &\quad& d_g(\iota^{-1}(\expect{\breve{\bx}_t}_{\zeta_0}), \bx_t) \leq h && \text{if } \expect{\bx_t} = \expect{\breve{\bx}_t}	\\
&	\mathcal{G}_{\LPCA}(\varepsilon) && : &\quad&
		\begin{cases}
			d_g(\iota^{-1} \breve{O}^{\sf T} \expect{\bx_t}_{\zeta_0}^{\LPCA}, \bx_t) \leq h ,	\\
			\tilde{\nu} = \nu \wedge ||\hat{\Pi} - \Pi|| \leq \varepsilon
		\end{cases}
	&& \text{if } \expect{\bx_t} = \expect{\bx_t}^{\LPCA} .
\end{alignat*}

We now see that these events imply that the computed mean is intrinsically $\mathcal{O}(h)$ close to $\bx_t$.
Let $\underline{C}_{\lambda'} := \inf p_{\epsilon,\lambda'} > 0$ and note that if $\delta_0 = \min\{ \underline{C}_0^{2\lambda}, \underline{C}_{\lambda}^{\nicefrac{3}{2}} \}/2$, then by \cref{eq:discrete-prop-sup-bound} and the events $\mathcal{A}_0(h^{\beta} ; \tilde{\zeta}_0)$ and $\mathcal{E}_{\lambda'}(\delta_0)$,
\begin{align*}
	||  |\cstate[\breve]{t, \zeta_0}|^2 - |\dcstate{t, \tilde{\zeta}_0}|^2 ||_{\infty}	
		&\leq || |\cstate[\breve]{t, \zeta_0}| - |\dcstate{t, \tilde{\zeta}_0}| ||_{\infty} (||\cstate[\breve]{t, \zeta_0}||_{\infty} + ||\dcstate{t, \tilde{\zeta}_0}||_{\infty})	\\
		& \lesssim h^{\beta}  \epsilon^{-\frac{\nu + 1}{2}} (h^{-\nicefrac{\nu}{4}} + \epsilon^{-\frac{\nu + 1}{2}}) .
\end{align*}
Therefore, setting $\beta = (2 + \alpha)(\nu + 1) + 1$ gives $|| |\cstate[\breve]{t, \zeta_0}|^2 - |\dcstate{t, \tilde{\zeta}_0}|^2 ||_{\infty} \lesssim h$.
With this, we have that for any $u \in C^{\infty}$,
\[
	|\sum_{\ell \in S_{\ell^*}} u_{\ell} (q_{\ell} - \breve{q}_{\ell}) | \lesssim h
\]
and hence upon denoting
\begin{gather*}
	\bar{v}_{\ell} := \sum_{\ell \in S_{\ell^*}} v_{\ell} q_{\ell} ,
	\quad	\bar{\vartheta}_{\ell} :=  \sum_{\ell \in S_{\ell^*}} \vartheta_{\ell}^{v_{\ell^*}} q_{\ell} ,
\end{gather*}
we have in particular,
\[
	|| \bar{v}_{\ell} - \breve{v}_{\ell} ||_{\mathbb{R}^n} \lesssim h ,
	\quad	|| \bar{\vartheta}_{\ell} - \breve{\vartheta}_{\ell} ||_{\mathbb{R}^{\nu}} \lesssim h,
\]
wherein the second inequality follows from event $\mathcal{G}_{\LPCA}$.
Hence, $||\expect{\breve{\bx}_t} - \expect{\breve{\bx}_t}_{\zeta_0}||_{\mathbb{R}^n} \lesssim h$ so under the event $\mathcal{G}_{\ext}$ with $h < \kappa$, we have $d_g(\iota^{-1} \expect{\breve{\bx}_t}, \bx_t) \lesssim h$.
Likewise, $||\expect{\bx_t}^{\LPCA} - \expect{\bx_t}_{\zeta_0}^{\LPCA}||_{\mathbb{R}^n} \lesssim h$ and there is a neighbourhood $V \subset \mathbb{R}^{\nu}$ of the origin, a neighbourhood $\mathcal{O} \subset \mathcal{M}$ of $\iota^{-1} \breve{O}^{\sf T} \expect{\bx_t}^{\LPCA}$ and a constant $\varepsilon_0 > 0$ such that for all $\varepsilon \in (0, \varepsilon_0]$, the event $\mathcal{G}_{\LPCA}(\varepsilon)$ gives
that $\iota^{-1} \breve{O}^{\sf T} : V \to \mathcal{M}$ is a diffeomorphism.
Therefore, we have also in the LPCA case that $d_g(\iota^{-1} \breve{O} \expect{\bx_t}^{\LPCA}, \bx_t) \lesssim h$.

The event $\mathcal{A}(h^{\beta} ; \tilde{\zeta}_0)$ happens with probability bounded below by that given in \cref{lem:model-cs}.
The event $\mathcal{E}_{\lambda'}$ happens with probability at least $1 - e^{-\Omega(N \epsilon^{\nicefrac{\nu}{2}})}$, 
by the arguments in the proof of
\cite[Lemma 17]{kumar_math}. 
Since $\varepsilon \in (0, \varepsilon_0]$ can be taken to be constant, following the proof of \cref{lem:model-cs}, the second event of $\mathcal{G}_{\LPCA}(\varepsilon)$ happens with probability at least $1 - e^{-\Omega(N^{\nicefrac{4}{(\nu + 2)}})}$.
The event $\mathcal{G}_{\ext}$ and the first event of $\mathcal{G}_{\LPCA}$ happen with probabilities bounded below by
\cite[Proposition 4]{kumar_math}: 
namely, there is a constant $h_0 > 0$ such that if $N^{-\nicefrac{1}{\beta_0}} \lesssim h \leq h_0$ with $\beta_0 := (2 + \alpha)(\nicefrac{5 \nu}{2} + 4) + 2(\nu + 2)$, then
\begin{equation} \label{eq:mean-plbd}
	\Pr[d_g(u(\expect{\bx_t}_{\zeta_0}), \bx_t) \leq h] > 1 - e^{-\Omega(N h^{\beta_0})}
\end{equation}
with $u := \iota^{-1}$ for $\expect{\bx_t}_{\zeta_0} = \expect{\breve{\bx}_t}_{\zeta_0}$ and $u := \iota^{-1} \breve{O}^{\sf T}$ for $\expect{\bx_t}_{\zeta_0} = \expect{\bx_t}_{\zeta_0}^{\LPCA}$.
Then, a union bound for each of the cases among the combinations of values of $\tilde{\zeta}_0$ and $\expect{\bx_t}$ gives the consistencies in the statement of the Proposition, \emph{with high probability}.
\end{proof}

\section{Tangent space estimation via local PCA}
\label{sec:tspace}
Local principal component analysis (LPCA) is utilized to formulate the initial coherent state and to calculate the expected position of the propagated state using only local data. Both of these procedures rely on the fact that LPCA provides an estimate of the tangent space of a manifold from sampled data, and LPCA coordinates match normal coordinates on the manifold up to third order in the size of the local neighborhood \cite{10.1093/imaiai/iat003}. 

In this section we illustrate the quality of LPCA constructed tangent spaces for the sphere and torus model manifolds for which intrinsic global coordinates are known and analytical expressions for all geometric quantities can be derived. We parameterize both two-dimensional manifolds by angular coordinates $(\theta, \phi)$ (explicit forms of the coordinatizations given below).

We quantify the accuracy of LPCA tangent space estimates as a function of the LPCA neighborhood size, $\delta_{\rm PCA}$ using two error metrics.
First, we compare the tangent space estimated by LPCA to the true tangent space at a point on the manifold using the subspace angle error metric:
\begin{align}
	\Delta_{TS} = \arccos(\vert V_1^{\mathsf{T}} V_2\vert),
\end{align}
where $V_1 (V_2)$ is a matrix whose columns are vectors that span the LPCA estimated tangent space (true tangent space). The second error metric we employ quantifies the error in the estimated norm of vectors in the tangent space at a point ($\theta_0, \phi_0)$. We compute $\vert \bp \vert = \vert \vartheta_0 - \vartheta_j\vert$, where $\bp$ is a vector between the point of interest $(\theta_0, \phi_0)$ with LPCA coordinates $\vartheta_0$, and some nearby point $\vartheta_j$. This quantity is compared against the tangent vector norm $\hat{n} = \sqrt{[d\theta, d\phi] g_{\theta, \phi}(\theta_0, \phi_0)[d\theta, d\phi]^{\mathsf{T}}}$, with $d\theta=\theta_j-\theta_0$ and $d\phi =\phi_j-\phi_0$, as:
\begin{align}
	\Delta_{N} = \vert|\vert \bp \vert - \hat{n}\vert|
\end{align}
This is an important error to quantify because the momenta for initial coherent states are normalized in LPCA coordinates, and in order to ensure unit speed propagation we require $\vert \bp \vert=1$.

\subsection{Sphere}
Consider a sphere unit radius parameterized by angles $0\leq\theta < \pi$, $0\leq \phi < 2\pi$  with radii $R, r$ $(r<R)$ and parameterized by angles $\theta$, $\phi$, with Euclidean coordinates in a 3D embedding given by:
\begin{align}
	(x,y,z) = (\sin(\theta)\cos(\phi), \sin(\theta)\sin(\phi), \cos(\theta))
\end{align}
The tangent space for this manifold at a point $\theta_0, \phi_0$ is a plane in the 3D embedding space spanned by the two vectors
\begin{align}
	e_1 &= [ -r\sin(\theta_0)\cos(\phi_0), -r\sin(\theta_0)\sin(\phi_0), r\cos(\theta_0) ]\nn \\
	e_2 &= [ -(R + r\cos(\theta_0))\sin(\phi_0), R + r\cos(\theta_0)\cos(\phi_0), 0]
\end{align}

We uniformly sample $N=8000$ points from a 3D embedding of the sphere and evaluate the error metrics $\Delta_{TS}$ and $\Delta_N$ at all points, as a function of $\delta_{\rm PCA}$, the Euclidean distance cutoff that dictates the LPCA neighborhood size. Fig. \ref{fig:lpca}(a) shows two examples of the estimated versus true tangent space at a point on sphere for $\delta_{\rm PCA}=0.3, 1.7$. Then in Fig. \ref{fig:lpca}(b) we show the average value of these error metrics (averaged over all 8000 sampled points) and one standard deviations error bars.

\begin{figure*}[t]
\centering
	\includegraphics[width=1\linewidth]{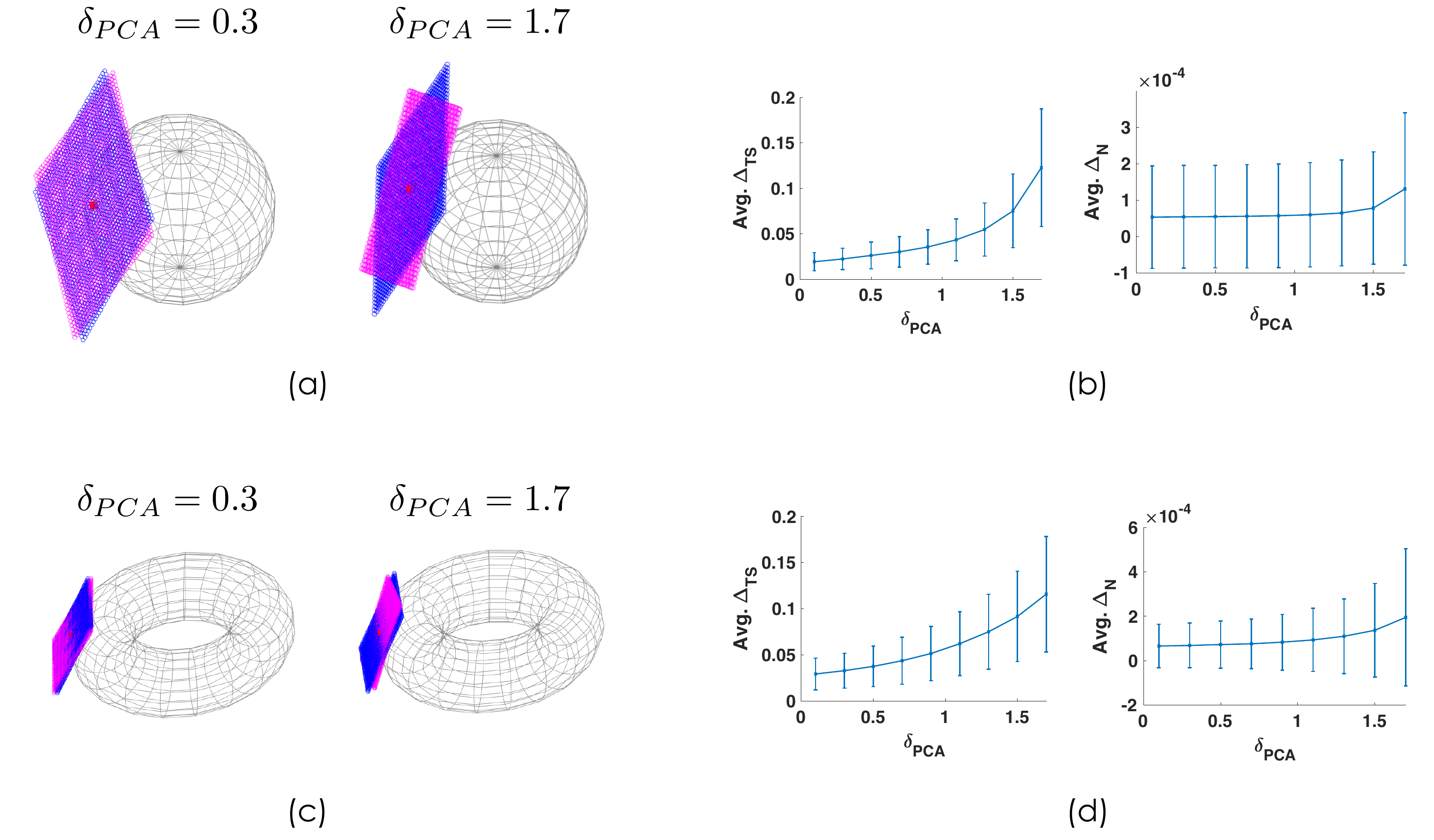}
	\caption{Error metrics for LPCA-based tangent space estimation for two model manifolds.\label{fig:lpca}}
\end{figure*}

\subsection{2-torus}
Consider a 2-torus with radii $R, r$ $(r<R)$ and parameterized by angles $0\leq \theta < 2\pi$, $0 \leq \phi < 2\pi$, with Euclidean coordinates in a 3D embedding given by:
\begin{align}
		(x,y,z) = ((R+r\cos(\theta))\cos(\phi), (R+r\cos(\theta))\sin(\phi), r\sin(\theta))
		\label{eq:torus_p}
\end{align}
The tangent space for this manifold at a point $\theta_0, \phi_0$ is a plane in the 3D embedding space spanned by the two vectors
\begin{align}
	e_1 &= [ -r\sin(\theta_0)\cos(\phi_0), -r\sin(\theta_0)\sin(\phi_0), r\cos(\theta_0) ]\nn \\
	e_2 &= [ -(R + r\cos(\theta_0))\sin(\phi_0), R + r\cos(\theta_0)\cos(\phi_0), 0]
\end{align}

We uniformly sample $N=12000$ points from a 3D embedding of a 2-torus with $R=2, r=0.8$ and evaluate the error metrics $\Delta_{TS}$ and $\Delta_N$ at all points, as a function of $\delta_{\rm PCA}$, the Euclidean distance cutoff that dictates the LPCA neighborhood size. Fig. \ref{fig:lpca}(c) shows two examples of the estimated versus true tangent space at a point on the 2-torus for $\delta_{\rm PCA}=0.3, 1.7$. Then in Fig. \ref{fig:lpca}(d) we show the average value of these error metrics (averaged over all 12000 sampled points) and one standard deviations error bars.

As is clear the statistics of these error metrics for the sphere and torus examples, the LPCA estimate of the tangent space is best for small $\delta_{\rm PCA}$. This creates a trade-off for our application: to get the best estimate of the tangent space a small $\delta_{\rm PCA}$ is preferred, but in order to obtain a good construction of an initial coherent state, we require a local neighborhood ($\mathcal{S}_0$ in the main article) that is not too small. To resolve this trade-off we fix $\delta_{\rm PCA}$ to be large enough to construct good fidelity coherent states, but construct those states using LPCA coordinates extracted from an LPCA computed in a neighborhood of size $\gamma\delta_{\rm PCA}$, for a scale factor $\gamma \leq 1$. In the illustrations shown in Sec. 3 of the main text, we used $\gamma=0.1$ for the sphere, and $\gamma=0.05$ for the torus.

Note that in general, a different choice of $\delta_{\rm PCA}$ can be made for the initial state construction and the expectation calculation (if both are done using LPCA). In all our illustrations in this work however, we choose the same value for both.

\section{Approximation of spectrum for sphere}
In the main article we showed how the spectrum of the data-driven approximation of the LB operator constructed from $N=12000$ samples from the 2-torus only matches the spectrum for the true LB operator on this manifold up to $\sim 100$ eigenvalues. Similarly, \Cref{fig:sphere_spec} shows the spectra for the LB operator on the sphere (blue circles) and the eigenvalues of the data-driven approximation constructed from $N=8000$ points (uniformly) sampled from the sphere (red crosses). We see again that these quantities begin to disagree at higher eigenvalues -- of particular note is the inability of the data-driven construction to match the degeneracies of the high eigenvalues.

\begin{figure*}[t]
\centering
	\includegraphics[width=0.5\linewidth]{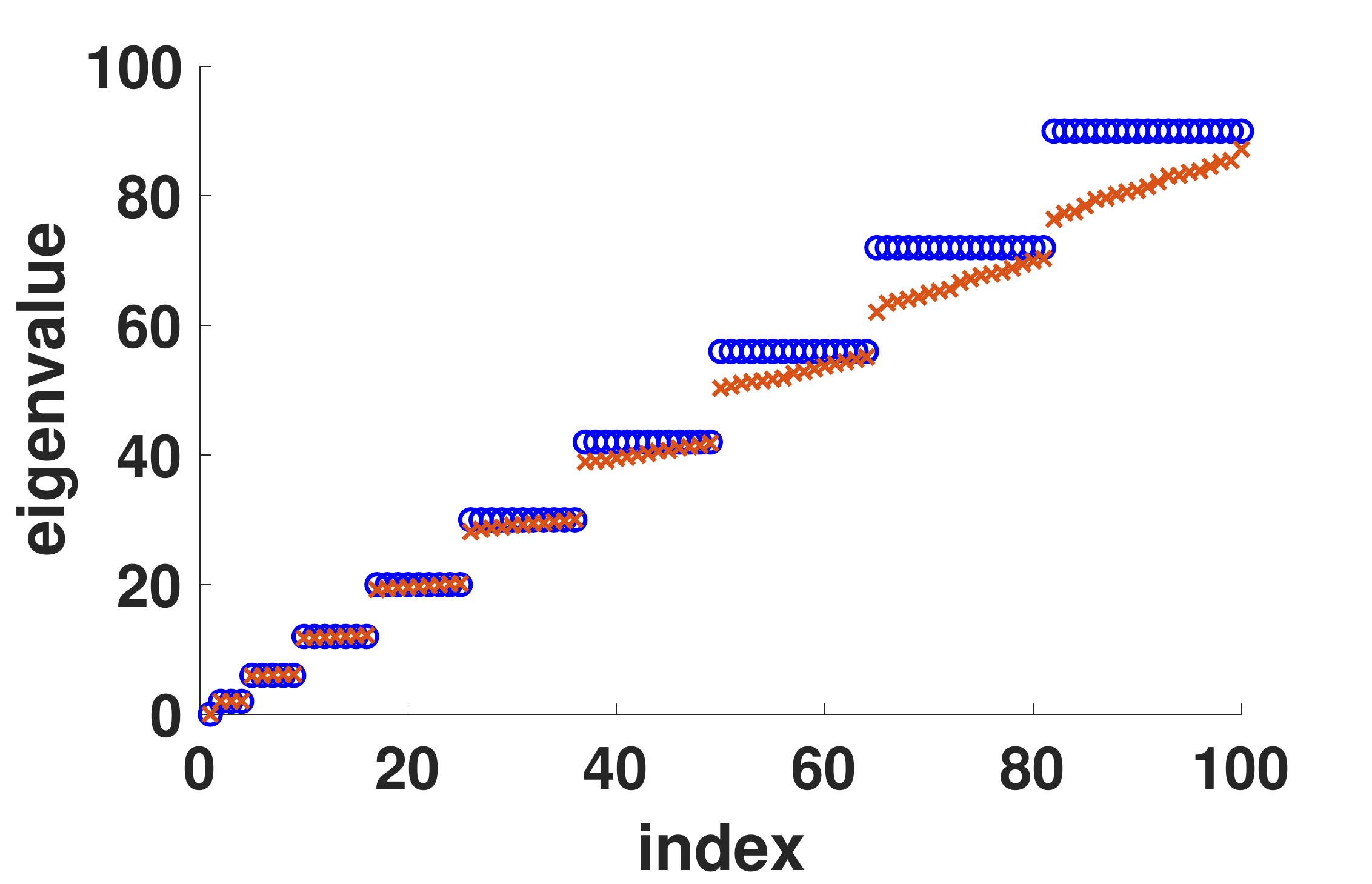}
	\caption{Eigenvalues of the LB operator on the sphere (blue circles) compared to the eigenvalues of the data-driven approximation of the LB operator constructed from $N=8000$ samples from the sphere with $\epsilon=e^{-4.8}$. \label{fig:sphere_spec}}
\end{figure*}

\section{Spectrally truncated propagator for the 2-torus}
In the main article we present the quantum evolution of coherent state by a spectrally truncated propagator for the 2-torus, as a way of illustrating the performance of our approach to extracting geodesics when the amount of data allows very accurate approximation of the Laplace-Beltrami (LB) operator on the data manifold. Here we present details of how the spectrally truncated propagator is constructed and additional data on propagation for long times with this propagator.

We construct the spectrally truncated propagator by numerically solving for the spectrum and eigenvectors of the LB operator on the 2-torus. Consider the explicit parameterization of a 2-torus with radii $r,R, r<R$ embedded in $\mathbb{R}^3$ given in \cref{eq:torus_p}. In these coordinates, the LB operators can be be explicitly written as:
\begin{align}
	\Delta = -\frac{1}{r^2\rho(\theta)}\frac{\partial}{\partial\theta}\rho(\theta)\frac{\partial}{\partial \theta} - \frac{1}{\rho(\theta)^2}\frac{\partial^2}{\partial\phi^2}
\end{align}
We wish to solve for the eigenfunctions and eigenvalues of this operator, \ie solve
\begin{align}
	-\frac{1}{r^2\rho(\theta)}\frac{\partial}{\partial\theta}\rho(\theta)\frac{\partial}{\partial \theta} \varphi_n(\theta, \phi) - \frac{1}{\rho(\theta)^2}\frac{\partial^2}{\partial\phi^2} \varphi_n(\theta,\phi) = \lambda_n \varphi_n(\theta,\phi),
	\label{eq:eval_eq}
\end{align}
for $\lambda_n$ and $\varphi_n(\theta,\phi)$. Taking into account the symmetries of the torus, we can separate variables and write the eigenfunctions as $\varphi_n(\theta,\phi) = \psi_n(\theta)e^{ik_n\phi}$, for $k_n \in \mathbb{N}_0$, and consequently \cref{eq:eval_eq} can be simplified to an equation for $\psi_n$ as
\begin{align}
	& -\frac{1}{r^2\rho(\theta)}\frac{d}{d\theta}\rho(\theta)\frac{d}{d\theta}\psi_n(\theta) + \frac{k_n^2}{\rho(\theta)^2}\psi_n(\theta) = \lambda_n \psi_n(\theta),\nn \\
	\Rightarrow &  \left[ \frac{-\rho'(\theta)}{r^2\rho(\theta)} \frac{d}{d\theta} - \frac{1}{r^2}\frac{d^2}{d\theta^2} + \frac{k_n^2}{\rho(\theta)^2} \right]\psi_n(\theta) = \lambda_n \psi_n(\theta)
	\label{eq:eval_psi_only}
\end{align}
with boundary conditions $\psi_n^{(j)}(0)=\psi_n^{(j)}(2\pi)$ on $\psi_n$ and all of its derivatives. 

We can approximate solutions to \cref{eq:eval_psi_only} by iterating over integers $k_n$, and for each value of $k_n$, discretizing the $\theta$ domain into $N_\theta$ points, and approximating the quantity in the square brackets by a matrix, $M(k_n)$. Then numerically solving a matrix eigenproblem of the form $M(k_n)\Psi = \lambda_n\Psi$ yields discrete approximations of the LB eigenfunctions and eigenvalues, with error scaling as $\mathcal{O}(\nicefrac{1}{N_\theta^2})$. We choose a fine enough discretization of the $\theta$ domain and take enough iterations of $k_n$ so that the numerically approximated spectral data is converged. 
The discretization used for the $r=0.8$ torus studied in the main text was 1023 (uniformly spaced) points in the domain $[0,2\pi)$.
We apply this procedure to approximate the first 400 eigenvalues, $\lambda_n$, and associated partial eigenvectors $\psi_n(\theta)$, of the 2-torus LB operator. 
This data can be used to construct full eigenvectors, $\varphi_n(\theta,\phi)$, by discretizing the $\phi$ coordinate and evaluating the phase $e^{ik_n\phi}$ at each point, for each $n$. Note that for each $k_n$ that generates a solution to the eigensystem in \cref{eq:eval_psi_only}, so does $-k_n$. Therefore, for each $\lambda_n, \psi_n(\theta)$ pair, we construct two $\varphi_n(\theta, \phi)$, with $\pm k_n$ phase factors.
Using this eigendata we form approximate unitary propagator as $\tilde{V}_{\Delta t} = X \exp(i\sqrt{\Lambda}\Delta t) X\dg$, where $\Lambda$ is a matrix of the first 800 eigenvalues (with degeneracies) and the columns of $X$ are the corresponding eigenvectors, $\varphi(\theta, \phi)$. In order to have a numerically tractable propagator $\tilde{V}_{\Delta t}$, we downsample the 1023 points in $\theta$ to $171$ and discretize $\phi$ to $300$ (uniformly spaced) points, yielding an approximation of the propagator evaluated on $N'=51300$ total points on the torus.

\subsection{Propagation details}
Propagation with the spectrally truncated propagator is performed in almost the same manner as described in the main text. Initial states are formed using LPCA coordinates with $\delta_{\rm PCA}=2.0$, and to allow comparison to the data-driven propagation, we fix the coherent state uncertainty parameter at the value dictated by the $\epsilon,\alpha$ choice made in the main text for the data-driven propagator with $N=12000$ points; \ie $h=e^{-1}$. An important difference to the procedure described in the main text is that for the spectrally truncated propagation, the initial momentum $p_0$ is chosen to be a uniformly random 2-vector in LPCA coordinates, instead of being a vector towards the nearest neighbor to any point. By the proof of Proposition 2 in \cref{sec:proofs} this is valid because the LPCA coordinates approximate the tangent space when there is sufficient density of local sampling. 

\subsection{Long-time propagation}
To support the 2-torus numerics in the main article, in Fig. \ref{fig:torus_exactprop} we show an example of long-time propagation with this spectrally truncated propagator. The propagation is remarkably accurate -- the coherent state remains localized and the propagation is along a unit-speed geodesic curve as the distance propagated plotted in \cref{fig:torus_exactprop}(b) shows.

\begin{figure*}[t]
\centering
	\includegraphics[width=1\linewidth]{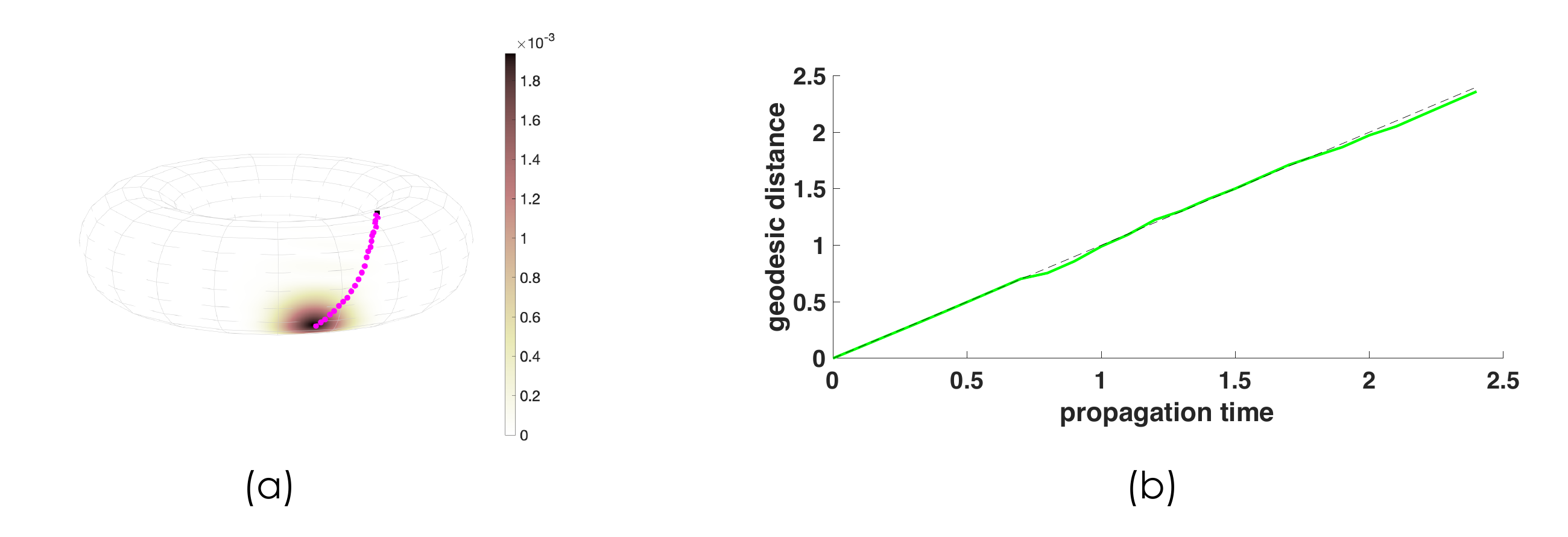}
	\caption{An example of long time propagation with a spectrally truncated propagator on the 2-torus. {\bf(a)} The path of a coherent state propagated for $t=2.4$ that traverses a significant portion of the torus. Note that the coherent state remains localized and of roughly Gaussian profile even after such a long propagation. {\bf(b)} The geodesic distance of $\expect{\breve{\bx}_t}^{\rm LPCA}$, the mean position of the propagated state at time $t$ from the initial point. We see that the mean remains a geodesic distance $t$ away even for long times. \label{fig:torus_exactprop}}
\end{figure*}

\raggedbottom
\pagebreak
\section{Pseudocode}
\label{app:code}
In this section we provide pseudocode that implements the core component of the quantum manifold learning program described in the main article, namely, the computation of a geodesic distance matrix $D$.

The inputs to the main function $\textsc{ComputeGeodesicMatrix}$ are the parameters:
\begin{enumerate}
\setlength\itemsep{0em}
\item $V \in \mathbb{R}^{N \times n}$ is the dataset consisting of $N$ vectors, each of dimension $n$.
\item $\epsilon>0$ is the scale parameters used in constructing the graph Laplacian.
\item $\alpha>0$ is the scaling parameters that determines the relationship between $h$ and $\epsilon$.
\item $0\leq \lambda\leq 1$ determines the graph Laplacian normalization. For all the illustrations in this paper, $\lambda=1$.
\item $\Delta t > 0$ is the time step used in constructing the unitary quantum propagator.
\item $n_{\rm prop} \in \mathbb{N}$ is the number of time steps to propagate for.  Geodesic distances are determined in steps of $\Delta t$ up to a distance $n_{\rm prop}\Delta t$.
\item $n_{\rm coll} \in \mathbb{N}$ is the number of distinct momenta to propagate from each data point on the manifold.
\item \texttt{PCA} is a boolean that flags whether LPCA is used in the algorithm, or states and measurements are constructed through extrinsic coordinates in $V$.
\item $\delta_{\rm PCA}>0$ is the Euclidean distance bound that determines the points used to construct the LPCA space around each data point (only used if \texttt{PCA}=1).
\item $0<\gamma<1$ is the scale factor that determines the LPCA projection (see discussion at the end of \cref{sec:tspace}).
\end{enumerate}
\raggedbottom

\begin{algorithm}[H]
\caption{Algorithm to compute geodesic distance matrix via quantum dynamics. Input parameters are explained in the text.}\label{algorithm}
\begin{algorithmic}[1]
	\State {\bf Inputs:} $V, \epsilon, \alpha, \lambda, \Delta t, n_{\rm prop}, n_{\rm coll}$, \texttt{PCA}, $\delta_{\rm PCA}, \gamma$
	 
    \Procedure{ComputeGeodesicDistanceMatrix}{}
      \State  $\breve{V}_{\Delta t} \leftarrow$ \textsc{ComputePropagator}($V, \epsilon, \lambda, \Delta t$) \Comment{Compute unitary quantum propagator from data}
      \State Set $h = \epsilon^{\frac{1}{(2+\alpha)}}$
      \State $D$ $\leftarrow$ \texttt{ZeroMatrix}($N,N$) \Comment{Initialize $N\times N$ geodesic distance matrix with zeros}
      \If{\texttt{PCA}}
      	\For{$1 \leq j \leq N$} \Comment{Loop over all points in data set}
      		\State $\vartheta_{\ell}^j \leftarrow$ Local PCA coordinates for all points $\ell \in \mathcal{S}_0$, a radius $\delta_{\rm PCA}$ ball around $j$ \Comment{Compute LPCA coordinates}
      		\For{$1 \leq m \leq n_{\rm coll}$} \Comment{Formulate $n_{\rm coll}$ initial states centered at point $j$, each with different momentum}
      			\State $p_0 = \vartheta_k^j - \vartheta_j^j$ for point $k$ that is $m$th closest to point $j$ in LPCA coordinates
      			\State $L^2$ normalize $p_0$
      			\State Form the vector $\vert \psi_j )$ with elements $e^{\frac{i}{h}((\vartheta_\ell^j - \vartheta_j^j)^{\sf T}p_0}e^{-\frac{\vert\vert\vartheta_\ell^j - \vartheta_j^j\vert\vert^2}{2h}}$ if $\ell \in \mathcal{S}_0$ and zero otherwise.
      			\State $L^2$ normalize $\vert \psi_j )$
      			\For{$1 \leq n \leq n_{\rm prop}$} \Comment{Propagate each initial state}
      				\State Compute $\vert \psi_j^n ) = \breve{V}_{\Delta t}^n \vert \psi_j )$
      				\State $L^2$ normalize $\vert \psi_j^n )$ 
      				\State Compute $q^n = |\vert \psi_j^n )|^2$
      				\State Let $\ell^* = {\rm argmax}~ q^n$
      				\State $\vartheta_{\ell}^{\ell^*} \leftarrow$ Local PCA coordinates for all points $\ell \in \mathcal{S}_{\ell^*}$, a radius $\delta_{\rm PCA}$ ball around $\ell^*$
      				\State Compute $\bar{\vartheta}$ \Comment{Compute mean position dictated by propagated state in LPCA coordinates}
      				\State Compute point $\ell^{d}$ such that $\vert\vert \bar{\vartheta} - \vartheta^{\ell^*}_{\ell^d}\vert \vert$ is minimized
      				\If{$[D]_{j,\ell^d}==0$ or $[D]_{j,\ell^d}> n\Delta t$}
      					 \State Set $[D]_{j,\ell^d}=  [D]_{\ell^d,j} = n\Delta t$ \Comment{Set geodesic distance from $j$ to $\ell^d$ to be $n\Delta t$ unless it is already closer}
      				\EndIf
      			\EndFor	
      		\EndFor
      	\EndFor
      \Else
      	\For{$1 \leq j \leq N$} \Comment{Loop over all points in data set}
      		\For{$1 \leq m \leq n_{\rm coll}$} \Comment{Formulate $n_{\rm coll}$ initial states centered at point $j$, each with different momentum}
      			\State $p_0 = v_k - v_j$ for point $k$ that is $m$th closest to point $j$ in original $n$-dimensional coordinates
      			\State $L^2$ normalize $p_0$
      			\State Form the vector $\vert \psi_j )$ with elements $e^{\frac{i}{h}((v_\ell - v_j)^{\sf T}p_0}e^{-\frac{\vert\vert v_\ell - v_j\vert\vert^2}{2h}}$, $ \forall 1 \leq \ell N$
      			\State $L^2$ normalize $\vert \psi_j )$
      			\For{$1 \leq n \leq n_{\rm prop}$} \Comment{Propagate each initial state}
      				\State Compute $\vert \psi_j^n ) = \breve{V}_{\Delta t}^n \vert \psi_j )$
      				\State $L^2$ normalize $\vert \psi_j^n )$ 
      				\State Compute $q^n = |\vert \psi_j^n )|^2$
      				\State Compute $\bar{v}$ \Comment{Compute mean position dictated by propagated state}
      				\State Compute point $\ell^{d}$ such that $\vert\vert \bar{v} - v_{\ell^d}\vert \vert$ is minimized
      				\If{$[G]_{j,\ell^d}==0$ or $[G]_{j,\ell^d}> n\Delta t$}
      					 \State Set $[G]_{j,\ell^d}=  [G]_{\ell^d,j} = n\Delta t$ \Comment{Set geodesic distance from $j$ to $\ell^d$ to be $n\Delta t$ unless it's already closer}
      				\EndIf
      			\EndFor	
      		\EndFor
      	\EndFor
      
      \EndIf
	  \State Set all diagonal elements of $G$ to zero
      \State \textbf{return} $G$
    \EndProcedure
    
    \Procedure{ComputePropagator}{$V, \epsilon, \lambda, \Delta t$}
    	\State Compute $\left[\breve{T}_{\epsilon}\right]_{i,j} = e^{-\nicefrac{||v_i - v_j||^2}{2\epsilon}}$  \Comment{Compute Gaussian kernel matrix}
    	\State Compute $\breve{L}_{\epsilon,\lambda} = \frac{4(\breve{I}_N-\breve{D}_{\epsilon,\lambda}^{-1} \breve{\Sigma}_{\epsilon}^{-\lambda} \breve{T}_{\epsilon} \breve{\Sigma}_{\epsilon}^{-\lambda})}{\epsilon}$ \Comment{Compute normalized graph Laplacian}
    	\State \textbf{return} $\breve{V}_{\Delta t} = \exp(-i \Delta t \sqrt{\breve{L}_{\epsilon,\lambda}})$ \Comment{Compute quantum propagator}
	\EndProcedure
    \end{algorithmic}
\end{algorithm}

\section{Gaussian mixture modeling clustering for COVID-19 mobility data}
In the main article we presented results when the geodesic distance based embeddings were clustered using the k-means algorithm. To show that our conclusions are insensitive to this choice of clustering algorithm we show here clustering of the same data using the Gaussian mixture modeling (GMM), which was the clustering algorithm used in previous work with the same dataset \cite{Levin_Chao_Wenger_Proctor_2020}.

We begin with the same 3D force-based graph embeddings of the 117 dimensional data points produced in the main article, using geodesic distances computed using our approach with location of the propagated state estimated via the the expected position or the maximum of the probability distribution dictated by the propagated state. Then we perform GMM clustering into five clusters. The results for both the mean and max cases are shown in \cref{fig:gmm}. Although the clusters are slightly different from those found by k-means (both k-means and GMM are probabilistic and will produce slightly different clusters for each execution) the key conclusions hold: (i) in the case where the location is estimated via the expected position (mean), the average mobility time traces are clearly separated and thus each cluster has a distinct average mobility pattern, and (ii) in the case where the location is estimated via the maximum (max) the outlier CBGs, that exhibit an abrupt change in mobility around May 1, 2021, are grouped into a separate cluster.

\begin{figure}[t]
\centering
	\begin{subfigure}[b]{0.45\textwidth}
	\centering	
	\includegraphics[width=\textwidth]{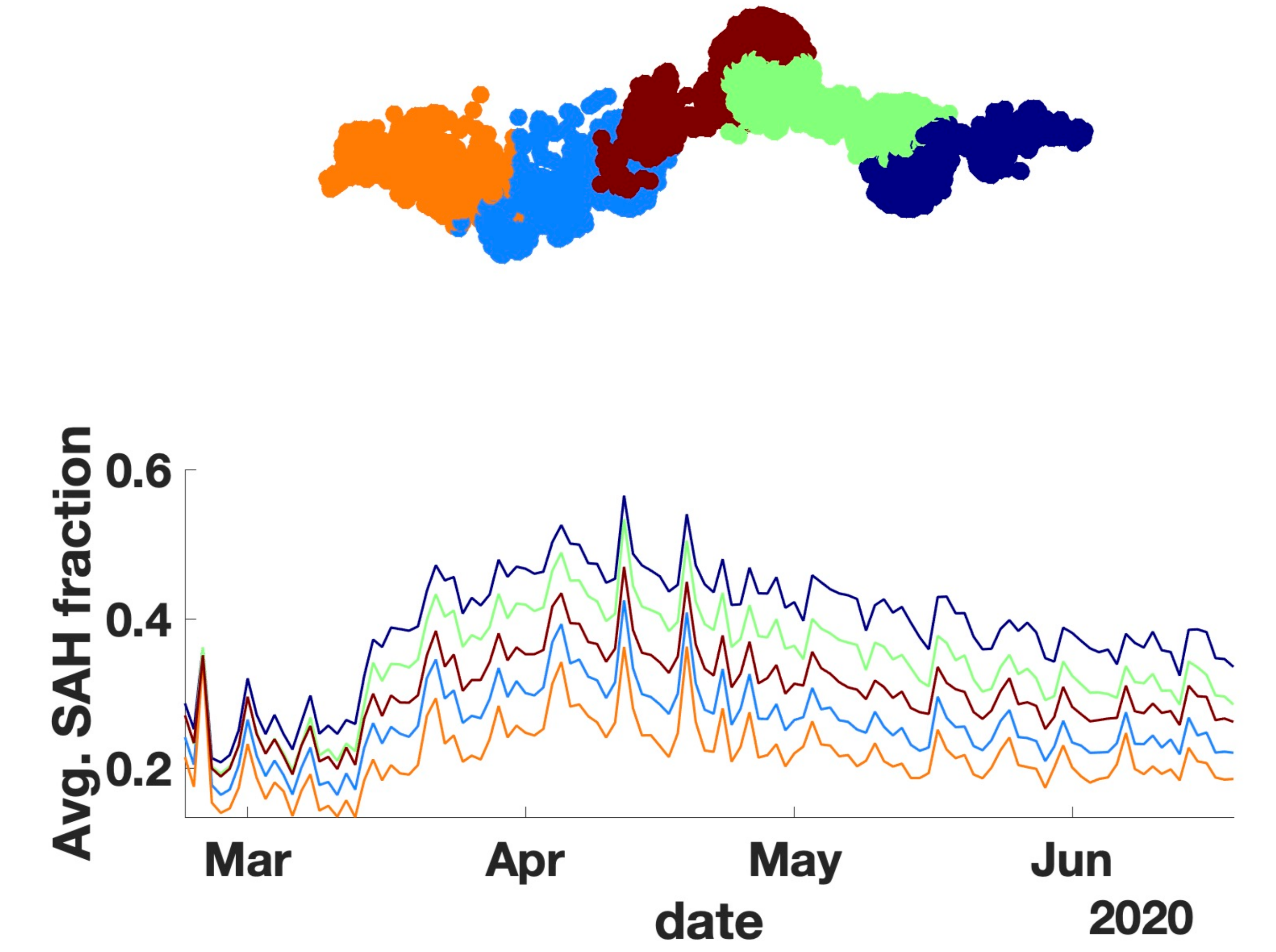}
	\caption{Mean}
	\end{subfigure}
	\hfill
	\begin{subfigure}[b]{0.45\textwidth}
	\centering	
	\includegraphics[width=\textwidth]{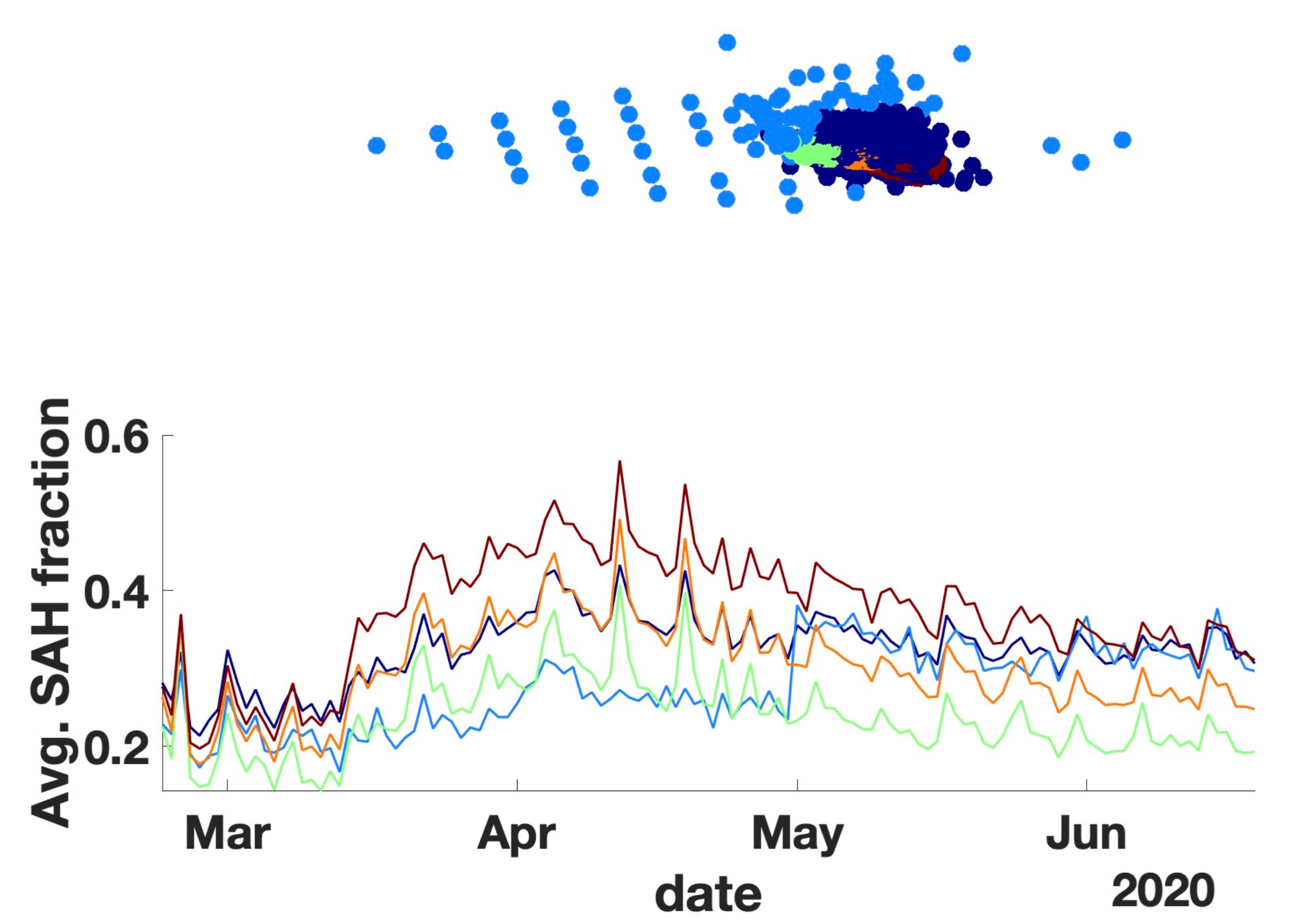}
	\caption{Max}
	\end{subfigure}
	\caption{GMM-based clustering for the same data presented in Fig. 5 of the main article, for cases where the location of the propagates states are extracted via \textbf{(a)} an expected position computation, $\expect{\breve{\bx}_t}$, or \textbf{(b)} the maximum of the resulting distribution over points, $\ell^*$.\label{fig:gmm}}
\end{figure}

\section{Small $N$ clustering examples}
In this section we present two examples of application of our manifold learning approach to clustering of small high-dimensional datasets. The theoretical results presented in the main article provide guarantees on the properties of our algorithm in the large $N$ limit, however, as we will show in this section it also reveals interesting information for small $N$ datasets, especially in the context of clustering.

\subsection{Global COVID-19 mobility patterns}
In the main article we analyzed mobility information during the first three months of the COVID-19 pandemic in the United States at a fine-grained geographic level by considering mobility data for individual census block groups in Georgia. Here we consider another dataset with mobility information, but this time analyze it at a much more coarse-grained geographic scale.

We consider the COVID-19 Community Mobility Reports dataset from Google \cite{google-covid} that aggregates mobile phone location information to capture mobility trends. The dataset charts movement trends in almost every country, including categorizing according to subregions/states in several countries, and covers dates from Feb. 15, 2020 to present day (as of the writing of this manuscript). In addition, the mobility information catalogued across six different categories of places: Retail and Recreation, Grocery and Pharmacy, Parks, Transit Stations, Workplaces and Residences. The key relevant quantity in the dataset is the \emph{percentage change in mobility} from a baseline value for each category and region calculated during the period Jan. 3rd, 2020 to Feb. 6, 2020. The baseline value is calibrated according to day of the week.

In the following, we analyze the dataset over the time period Feb. 15, 2020 to Jan. 24, 2021 (345 days). We first preprocess the data in a couple of steps. First, we coarse-grain the dataset to country-level information; \ie we consider mobility change patterns for each country. Out of the total 132 countries in the dataset we drop 6 countries with poor quality or significant missing data, leaving a dataset with $N=126$ countries. Second, we ignore the Residential category in the dataset since (i) it has different units than the other categories and the documentation specifically says it cannot be treated on the same footing as the other categories \cite{google-covid}. This leaves 5 categories across which mobility changes are tracked for 345 days per country, which defines a time series for each category. We concatenate all 5 time series for each country yielding data vectors of dimension $n=345\times 5 = 1725$.  We normalize each data sample by the maximum $L^2$ norm of the $N=126$ data samples so that each sample can be viewed as a vector in the unit ball in $R^{1725}$.

We calculate geodesic distances between countries in this dataset and then use the geodesic distance matrix to embed the data in $\mathbb{R}^3$ using the same embedding technique employed in the main article (force-directed graph layout). Due to the small size of the dataset, we form initial states and calculate expected positions using extrinsic coordinates (no local PCA) and using the maximum of the propagated state as the estimate of the expected position. The parameters used are: $\epsilon=e^{-0.6}, \alpha=1, dt=0.2, n_{\rm prop}=4, n_{\rm coll}=5$. After this embedding (which is a reduction from 1725 dimensions to 3) we perform k-means clustering using the embedding coordinates to cluster the countries into 6 clusters (we experimented with varying numbers of clusters and empirically found that 6 was optimal in terms of grouping clusters according to mobility patterns). The resulting embedding and color-coded clusters are shown in \cref{fig:globalcovid_graph}. 

The geodesic distance based embedding is meaningful is several ways. It clearly distinguishes different geographic regions, with countries in a region (\eg northern Europe) being close to each other. This is reasonable since social distancing policies and the timing of their introduction during the pandemic varied across geographic regions. In addition, the mobility patterns are highly correlated for countries within a cluster. In \cref{fig:globalcovid_traces} we show mobility patterns for some sample countries within each cluster (top) and also the average percentage change in mobility for the countries in each cluster (bottom). Both of these are shown for each of the five categories considered. Each cluster captures a distinct mobility pattern over the 345 days and the countries within a cluster show remarkably similar mobility patterns. Both of these features indicate a meaningful clustering of this high-dimensional dataset.

\begin{figure*}[t]
\centering
	\includegraphics[width=0.8
	\linewidth]{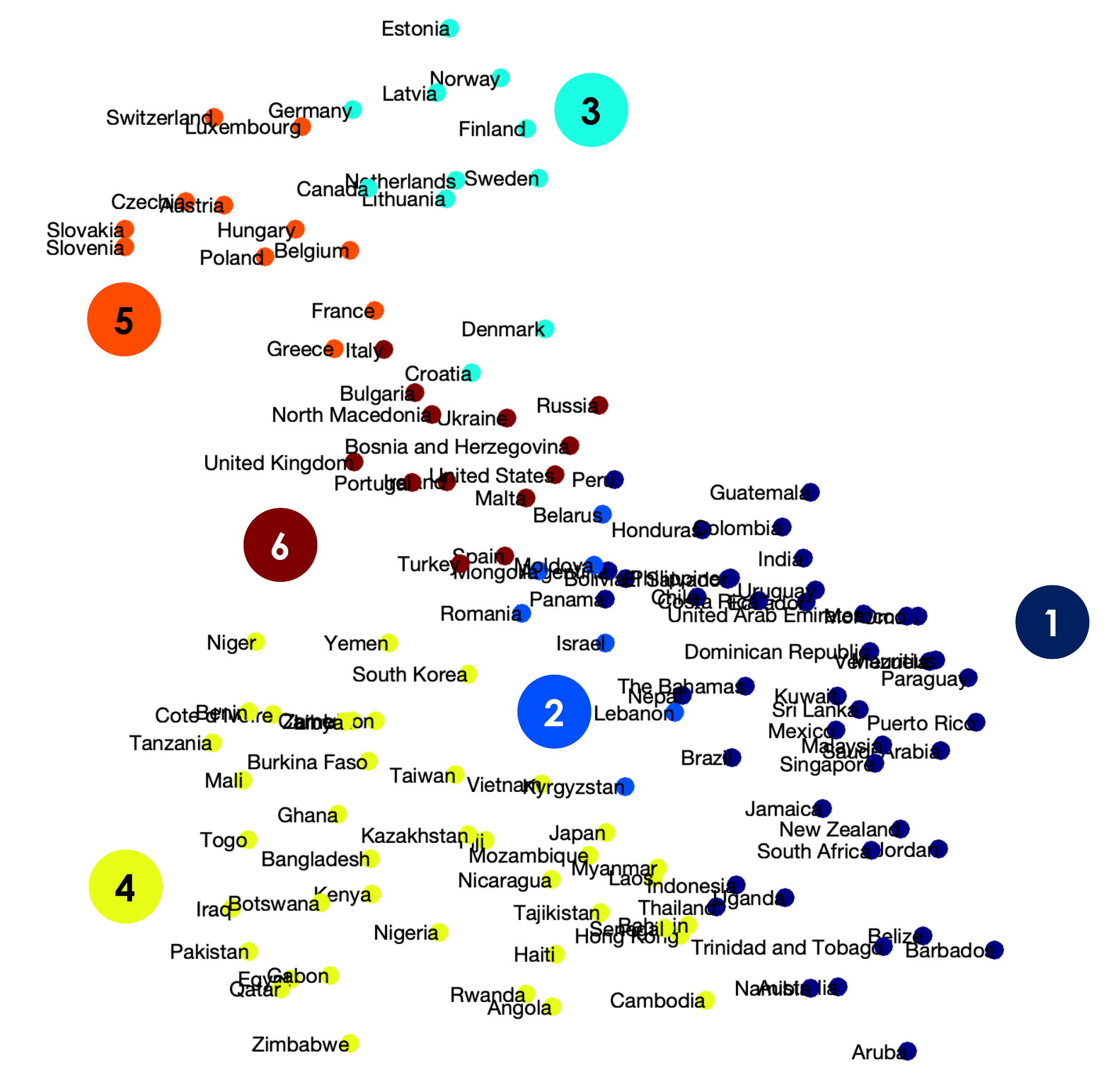}
	\caption{Force-based graph layout of countries based on geodesic distance matrix extracted using our approach applied to mobility pattern data in these countries during the COVID-19 pandemic. The colors indicate clusters found using k-means clustering based on the 3D coordinates determined by the embedding. \label{fig:globalcovid_graph}}
\end{figure*}

\begin{figure*}[t]
\centering
	\includegraphics[width=0.8
	\linewidth]{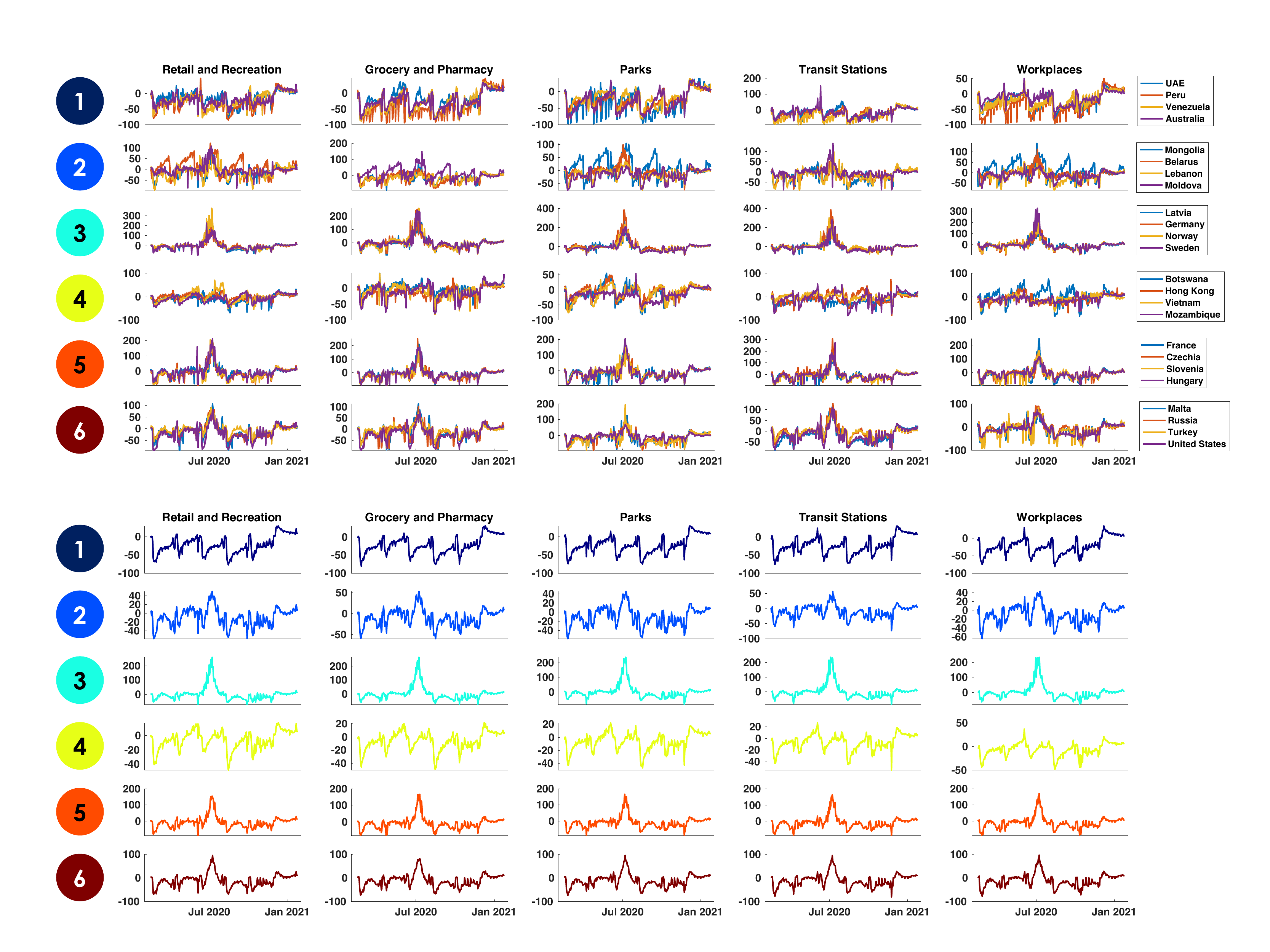}
	\caption{{\bf (Top)} Percentage mobility changes over the monitored time period (Feb 15, 2020 to Jan 24, 2021) for four randomly sampled countries in each cluster (rows), in each of the categories collected in the Google Community Mobility Reports (columns). The countries are listed to the right of each row. Note the similarity of the mobility patterns for countries in the same cluster. {\bf (Bottom)} The average percentage mobility change for all countries within each cluster, and separated across the categories (columns). Each cluster shows a distinct average mobility pattern.  \label{fig:globalcovid_traces}}
\end{figure*}

\subsection{NBA player performance}
The second small $N$, high-dimensional dataset we consider is the NBA Players dataset on Kaggle \cite{nba}, which contains over two decades of statistics on basketball players in the National Basketball Association (NBA). The data was scraped and collated by Justinas Cirtautas using the NBA Stats API. 

We first preprocess the dataset and remove any players without at least five consecutive seasons of data, which leaves $N=920$ samples in the dataset. Then we focus on three key statistics for each player, over the first five years of their career in the NBA (or first five years for which there are records in the NBA Players dataset): average points per game, average rebounds per game, and average assists per game. This yields $3\times 5=15$ features per data point. The data samples are minmax normalized so that the minimum statistic is zero and maximum is one.

We calculate geodesic distances between players in this dataset and then use the geodesic distance matrix to embed the data in $\mathbb{R}^3$ using the same embedding technique employed in the main article (force-directed graph layout).
Due to the small size of the dataset, we form initial states and calculate expected positions using extrinsic coordinates (no local PCA) and using the maximum of the propagated state as the estimate of the expected position. The parameters used are: $\epsilon=e^{-1.3}, \alpha=1.2, \Delta t=0.1, n_{\rm prop}=6, n_{\rm coll}=40$. As with the previous example, we perform k-means clustering based on the coordinates from the 3D embedding. In \cref{fig:nba_graph} we show the embedding and clustering for five clusters. Some of the nodes are also labeled with the player names. 

The embedding and clustering shown in \cref{fig:nba_graph} groups players according to performance and playing style. In \cref{fig:nba_stats} shows the average statistics over five years for each of the clusters. Each cluster clearly shows a distinct pattern in terms of average statistic over the five years. The highest performers, or superstars who excel in all three categories (points, rebounds, assists), are grouped in the light blue cluster, while the orange cluster contains specialist centers and forwards who excel at rebounding and also score points, and the brown cluster contains specialist guards who excels at assists and scoring points. 

As with the global COVID-19 mobility pattern dataset, the low-dimensional embedding and subsequent clustering of this dataset using our approach demonstrates its utility on even small datasets.

\begin{figure*}[t]
\centering
	\includegraphics[width=0.8
	\linewidth]{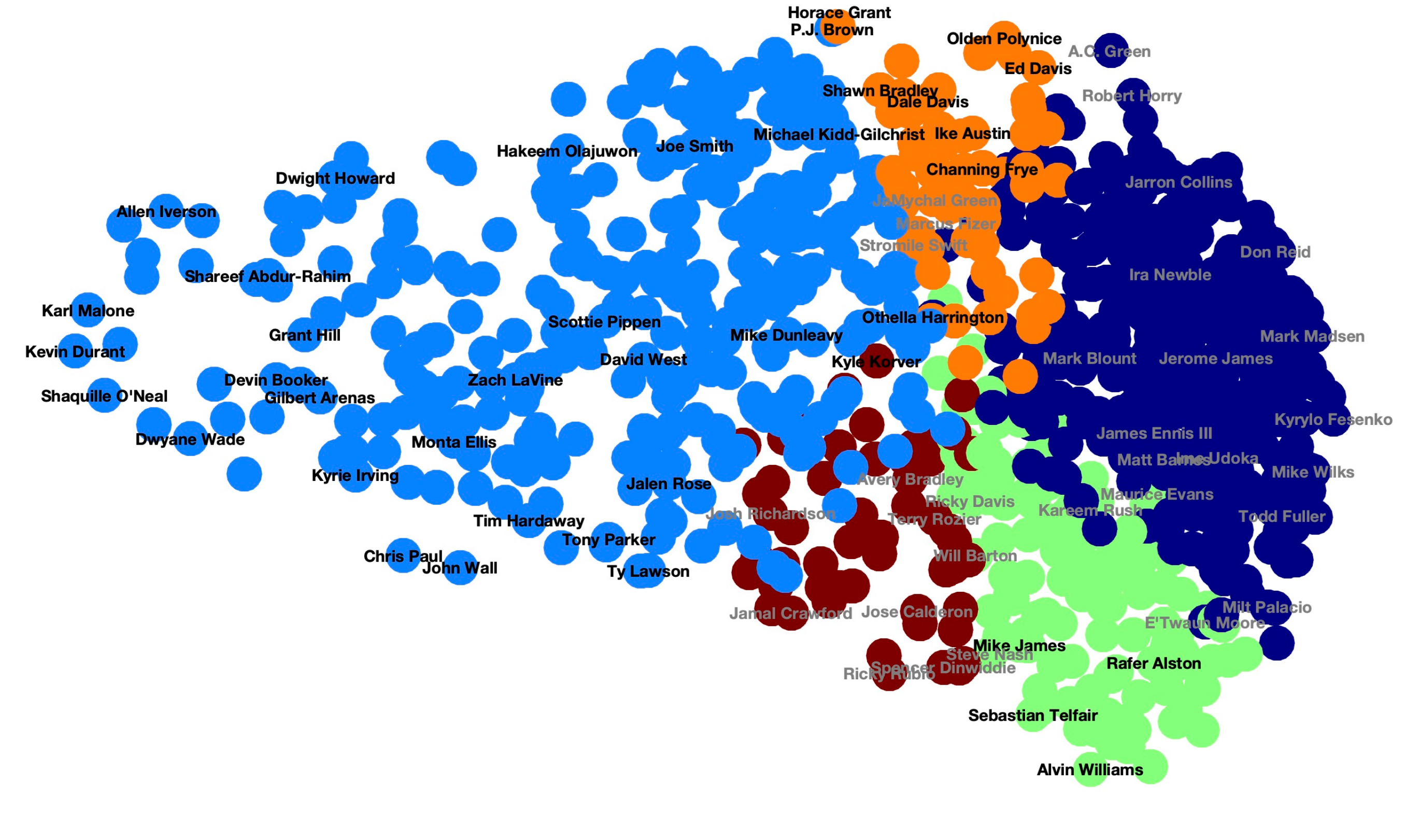}
	\caption{Force-based graph layout of 920 NBA players based on playing statistics over five years. The colors indicate clusters found using k-means clustering based on the 3D coordinates determined by the embedding. Some of the player names are listed to show representative players from each cluster (listing all names would make the plot illegible) \label{fig:nba_graph}}
\end{figure*}

\begin{figure*}[t]
\centering
	\includegraphics[width=0.6
	\linewidth]{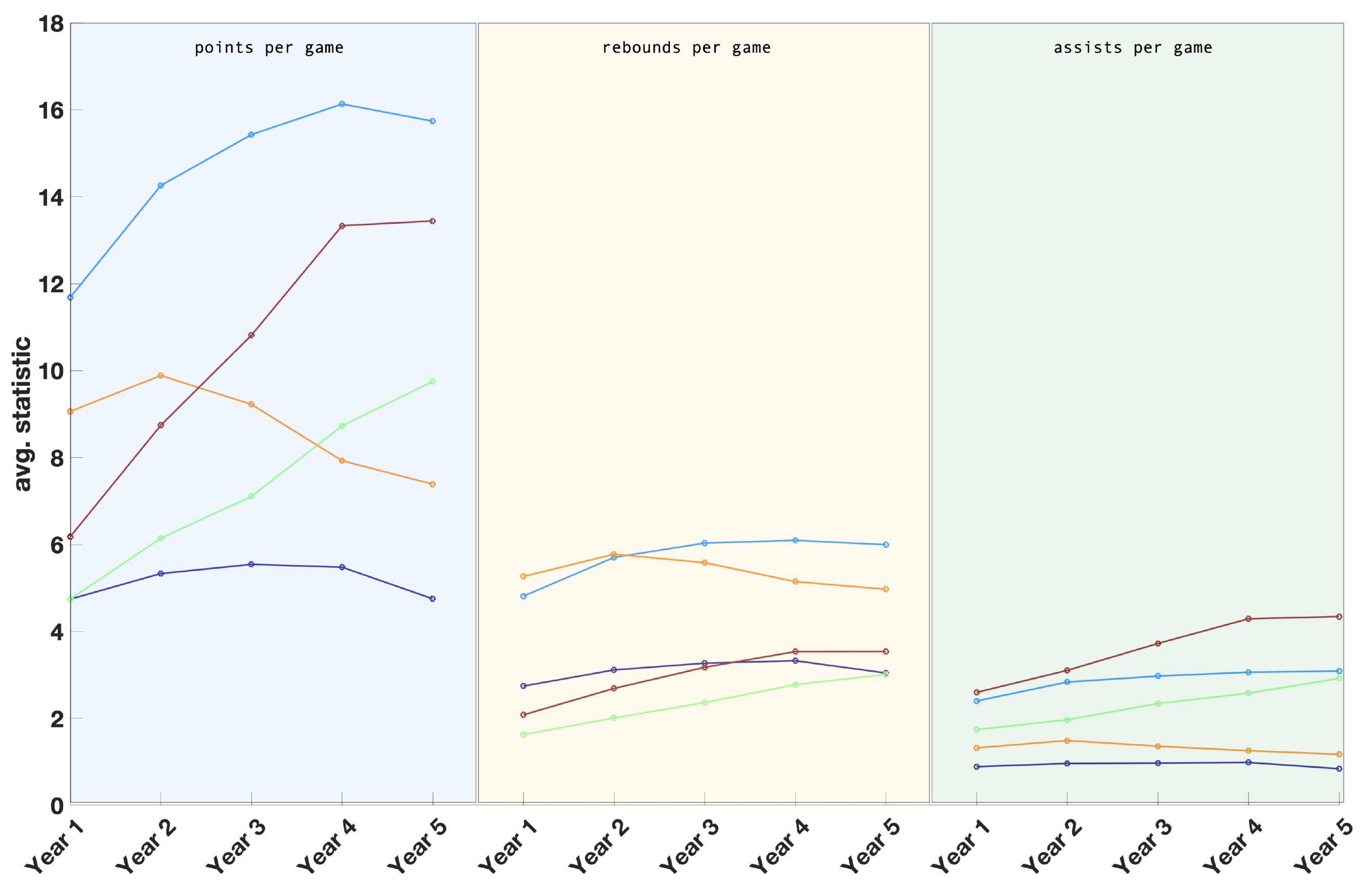}
	\caption{Average statistics, when averaged over all players in a given cluster. The colors of the averages correspond to the cluster color-codes in \cref{fig:nba_graph}. The panels show from left to right, the average points per game, average rebounds per game, and average assists per game, all over five years. \label{fig:nba_stats}}
\end{figure*}

\section{Deviation measure $\mathcal{D}$ for datasets}
The main article defined a useful quantity, $\mathcal{D}=\vert h^2(\psi^h_{\zeta_0}\vert \breve{L}_{\epsilon,\lambda} \vert \psi^h_{\zeta_0})-1\vert$, that can be used a heuristic for choosing good values of $\epsilon$ and $\alpha$ for a dataset. For ideal constructions of the propagator and coherent states, this quantity should be $\mathcal{O}(h)$. In this section we plot this deviation measure as a function of $\epsilon$ and $\alpha$ for all the datasets analyzed in this work. We average this quantity over 28 choices of initial point, $\bx_0$, on each dataset to avoid sampling-related bias around any particular point.

\begin{figure}[t]
\centering
	\begin{subfigure}[b]{0.3\textwidth}
	\centering	
	\includegraphics[width=\textwidth]{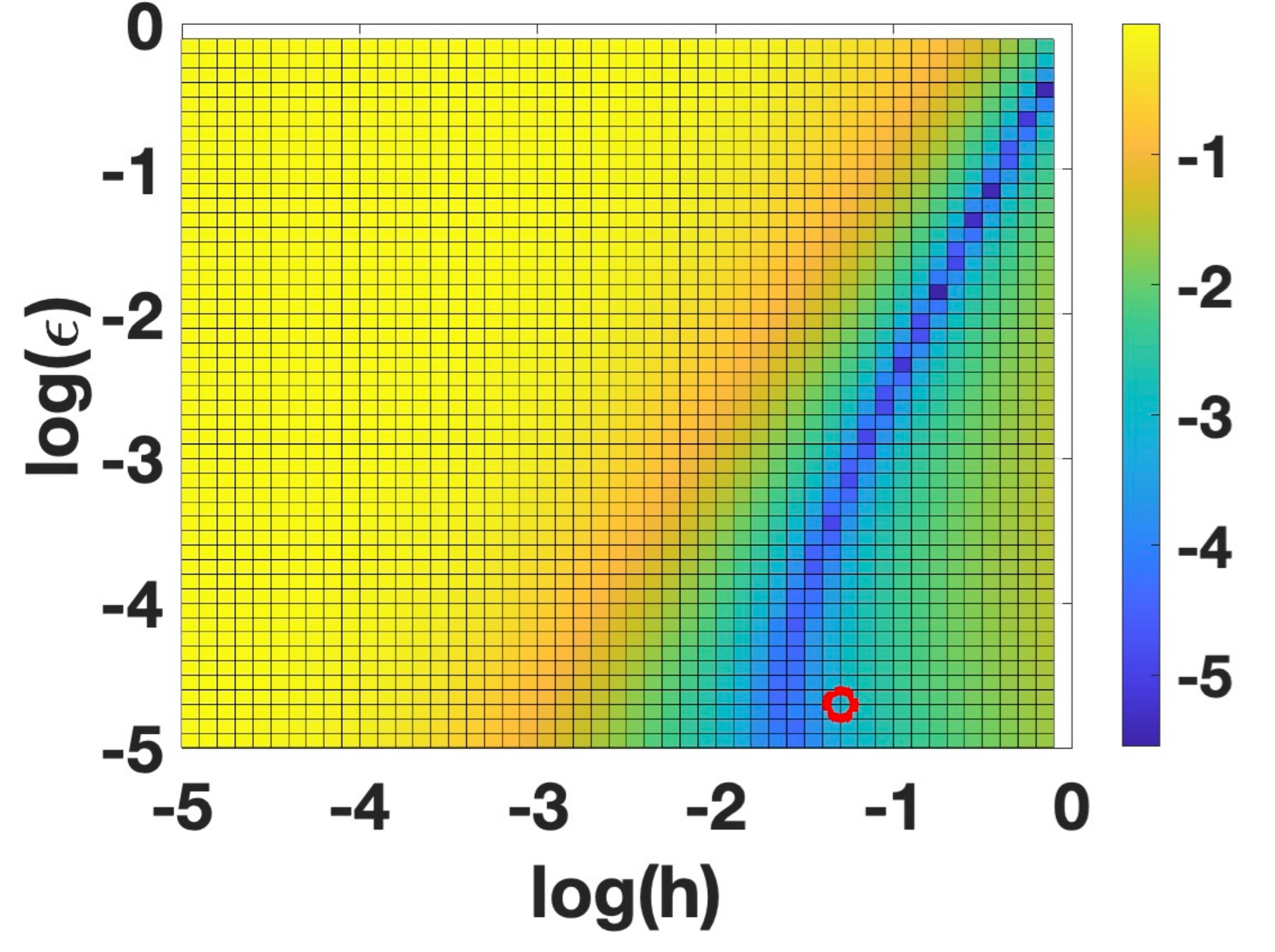}
	\caption{Sphere}
	\end{subfigure}
	\begin{subfigure}[b]{0.3\textwidth}
	\centering	
	\includegraphics[width=\textwidth]{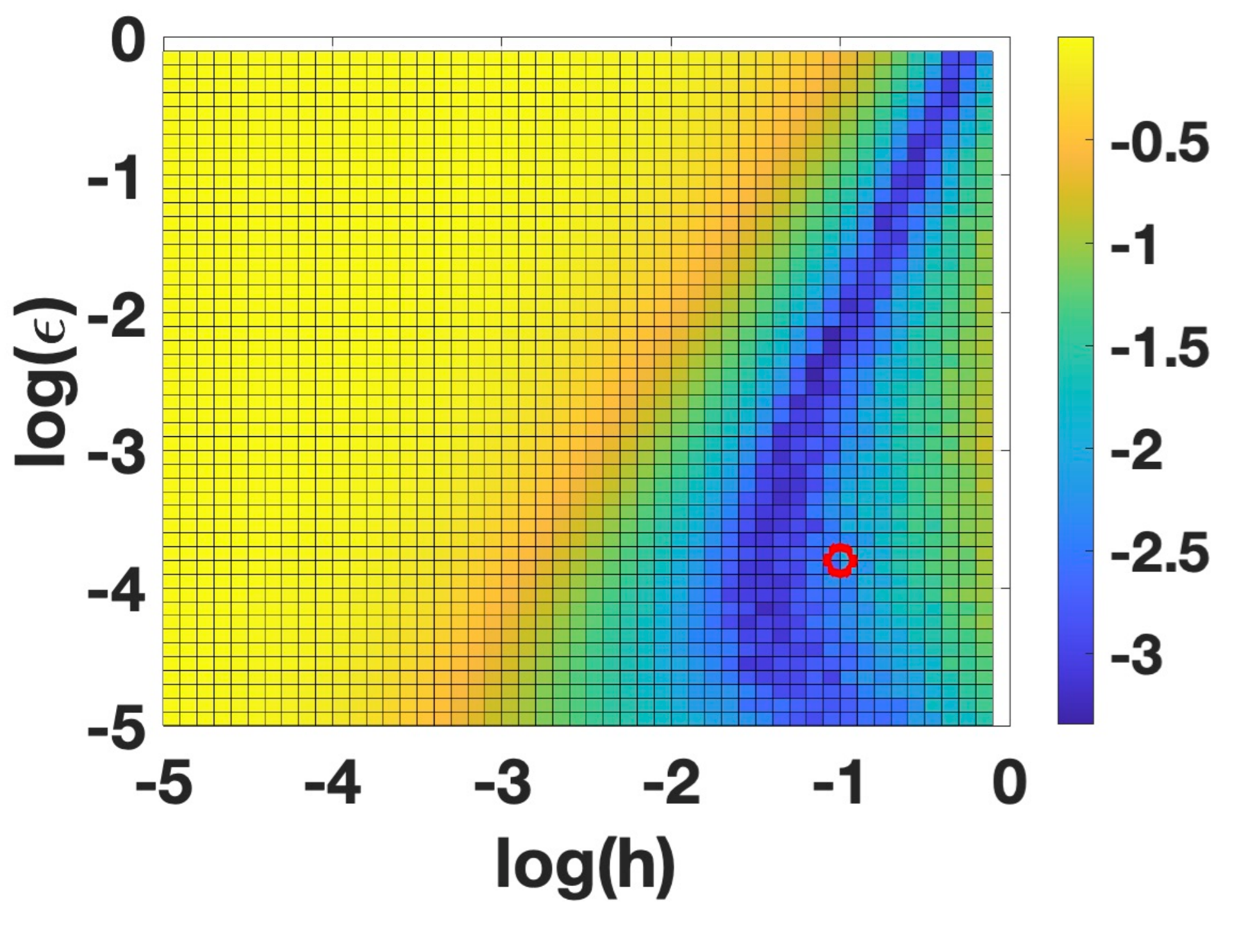}
	\caption{Torus}
	\end{subfigure}
	\begin{subfigure}[b]{0.3\textwidth}
	\centering	
	\includegraphics[width=\textwidth]{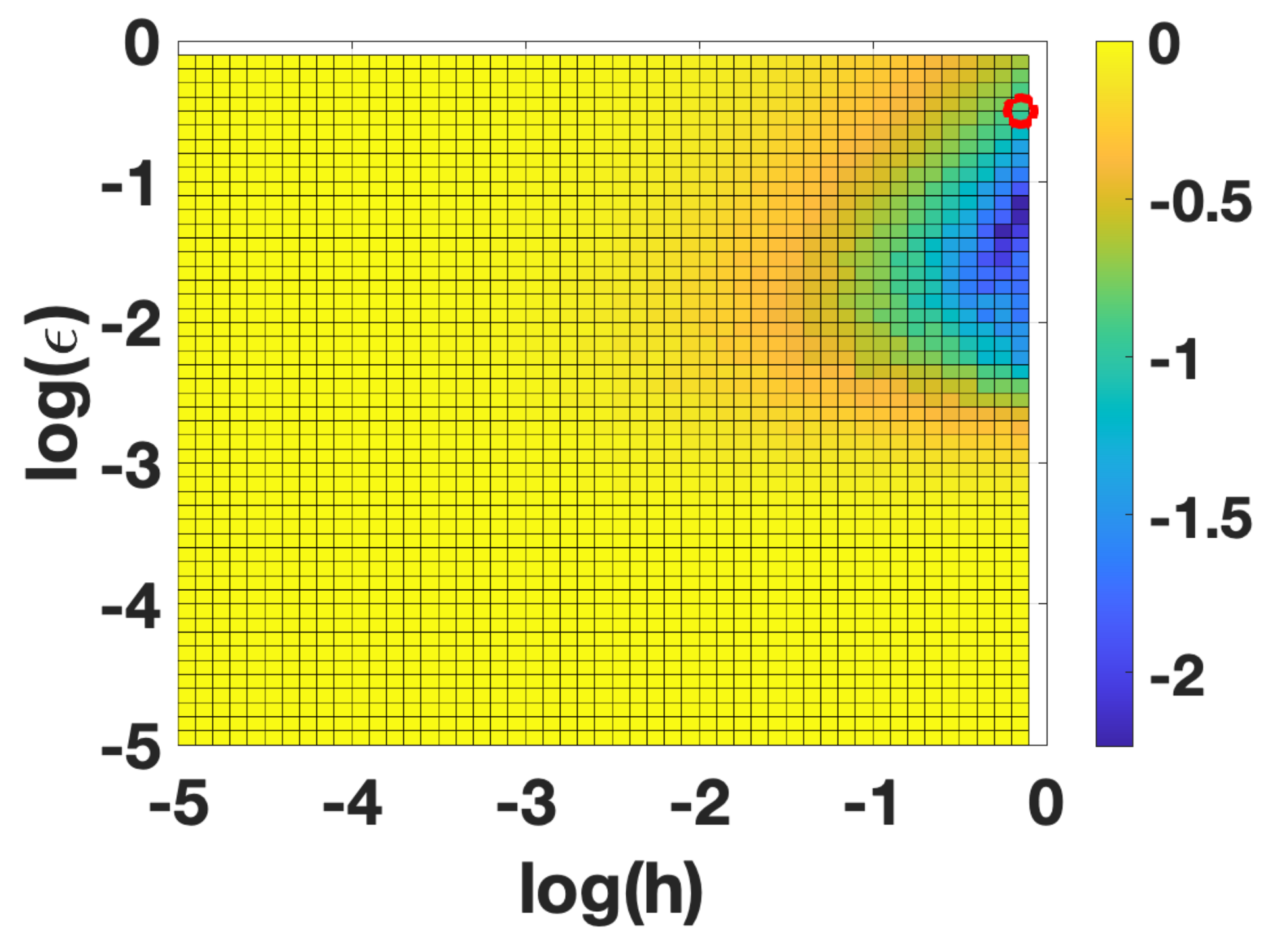}
	\caption{SafeGraph Inc. COVID mobility}
	\end{subfigure}
	\begin{subfigure}[b]{0.3\textwidth}
	\centering	
	\includegraphics[width=\textwidth]{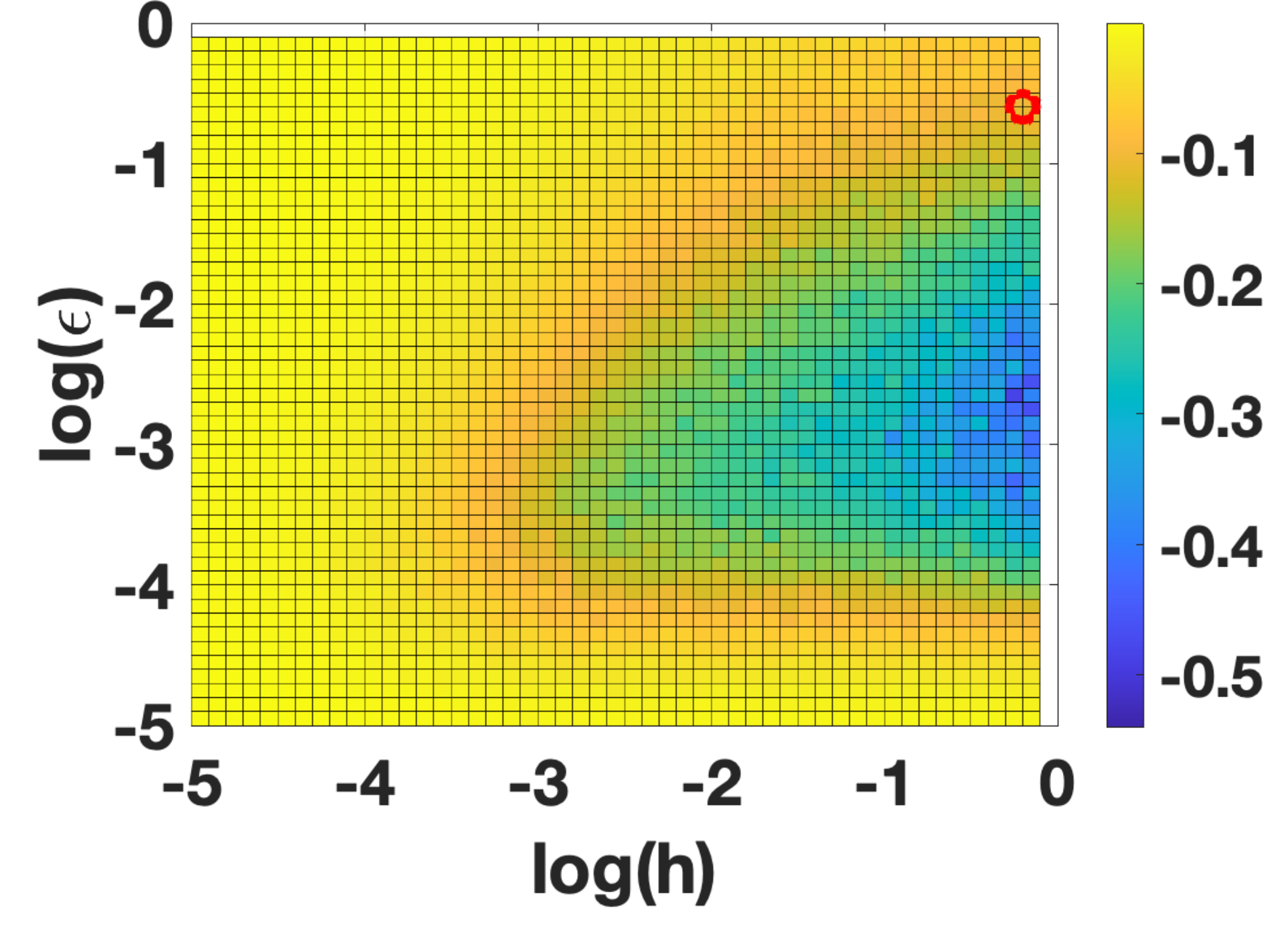}
	\caption{Google global COVID mobility}
	\end{subfigure}
	\begin{subfigure}[b]{0.3\textwidth}
	\centering	
	\includegraphics[width=\textwidth]{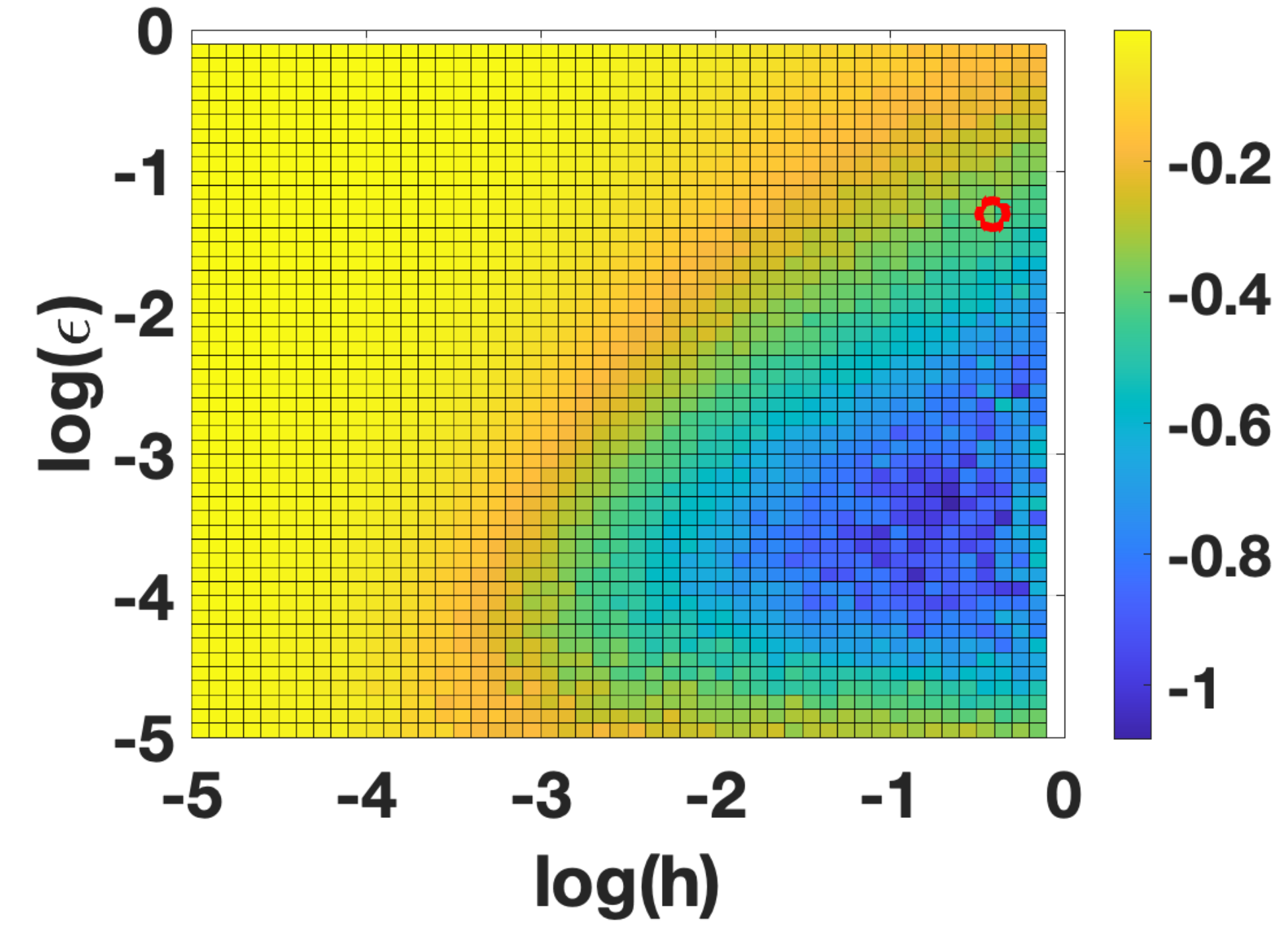}
	\caption{NBA player performance}
	\end{subfigure}

	\caption{ Plots of log of the average deviation measure $\mathcal{D}$ (averaged over 28 initial points on each dataset) as a function of $\log(\epsilon)$ and $\log(h)$ for the five datasets studied in this work. The red circle in each figure indicates the parameter values used in the analysis presented in the main article and this SI. \label{fig:D}}
\end{figure}

\end{document}